\documentclass[12pt]{article}
\usepackage{graphicx}
\usepackage{float}
\usepackage[flushleft]{caption2}
\usepackage{amsmath,amsfonts,amssymb}
\usepackage{amsthm}
\usepackage{mathrsfs}
\usepackage{enumitem}
\usepackage{tikz}
\usepackage[title]{appendix}
\usepackage{indentfirst}
\usepackage{epstopdf}
\usepackage{hyperref}
\hypersetup{hidelinks}

\DeclareMathOperator*{\Res}{Res}
\usepackage[sort&compress,numbers]{natbib}
\usetikzlibrary{arrows}
\newtheorem{lemma}{Lemma}[section]
\newtheorem{theorem}{Theorem}[section]

\newtheorem{proposition}{Proposition}[section]

\setitemize[1]{itemsep=0pt,partopsep=0pt,parsep=\parskip,topsep=5pt}
\numberwithin{equation}{section}
\newtheorem{thm}[theorem]{Theorem}
\newtheorem{RHP}{Riemann--Hilbert problem}
\newcommand{\R}{\mathbb{R}}
\newcommand{\Z}{\mathbb{Z}}
\newcommand{\C}{\mathbb{C}}

\def\d{{\rm d}}

\def\i{{\rm i}}
\def\1{\operatorname{Id}}

\def\Re{\operatorname{Re}}
\def\Im{\operatorname{Im}}
\def\exp{\operatorname{exp}}
\def\Bes{\mathrm{Bes}}
\def\Jac{\operatorname{dn}}

\def\le{\left}
\def\ri{\right}

\usepackage{tikz}
\usetikzlibrary{shapes,snakes,calc,fit}
\usetikzlibrary{decorations.text}
\usepgfmodule{shapes}
\usetikzlibrary{decorations.pathmorphing}
\usetikzlibrary{decorations.markings}
\usetikzlibrary{patterns}
\usetikzlibrary{automata}
\usetikzlibrary{positioning}
\usetikzlibrary{spy}
\usetikzlibrary{backgrounds}
\tikzset{partial ellipse/.style args={#1:#2:#3}{insert path={+ (#1:#3) arc (#1:#2:#3)} }}
\tikzset{->-/.style={decoration={ markings, mark=at position #1 with {\arrow{>}}},postaction={decorate}}}
\tikzset{-<-/.style={decoration={ markings, mark=at position #1 with {\arrow{<}}},postaction={decorate}}}

\title{\LARGE\bf Large-space and large-time asymptotics for the mKdV soliton gas with any odd genus}
\author{\hspace{0.6 cm}{Dedi Yan, Xianguo Geng\footnote{\footnotesize
 Corresponding author. {\sl Email address}: xggeng@zzu.edu.cn}, Kedong Wang}\\
\leftline{\hspace{0.6 cm}{\small{\sl School of Mathematics and Statistics, Zhengzhou University, 100 Kexue Road, Zhengzhou, }}}\\
\leftline{\hspace{0.6 cm}{\small{\sl Henan 450001, People's Republic of China}}}}
\date{}

\begin{document}
	\maketitle
	\begin{abstract}

	We study the large-space and large-time asymptotic behavior of the soliton gas of genus $2n-1$ for the mKdV equation with $n\in \mathbb{N}_+$. As $x \to +\infty$, we show that the large-space asymptotics of the mKdV soliton gas can be expressed with the Riemann-theta function of genus $2n-1$. For large $t$, based on the nonlinear steepest descent method and $g$-function approach, we establish a global large-time asymptotic description of the mKdV soliton gas. The half-plane $\{(x,t):-\infty<x<+\infty, t>0\}$ is divided into $2n+1$ separated regions. In each region, the large-time asymptotics of the mKdV soliton gas is given by using the Riemann-theta functions and uniform error estimation.\\

\noindent{\rm Keywords:} mKdV soliton gas, large-space and large-time asymptotics, Riemann--Hilbert problem, Riemann-theta function\\
\noindent{\rm Mathematics Subject Classification:} 35Q15, 35B40
	\end{abstract}
	
\section{Introduction}
The focusing modified Korteweg--de Vries (mKdV) equation
	\begin{align}\label{mKdV}
		u_{t}+u_{xxx}+6u^{2}u_{x}=0,\ \ -\infty<x<+\infty,\ t>0,
	\end{align}
is one of the fundamental models in the field of integrable systems, and has received extensive attention  because of its rich mathematical structure and physical applications. The global well-posedness of the Cauchy problem for the mKdV equation has been studied in \cite{CKSTT,KPV}. Miguel and Claudio discussed the nonlinear stability of breathers of the focusing mKdV equation by introducing a new Lyapunov
functional \cite{MC}. Ling and Sun constructed and analyzed the multi ellipticlocalized solutions under the background of elliptic function solutions for the mKdV equation \cite{Ling}. Germain et al. used the space-time resonance to study the stability of solitary waves and long time asymptotics with small initial data \cite{GPR}. Grava and Minakov obtained the long-time asymptotics for the mKdV equation with step-like initial data \cite{TM} with the help of the nonlinear steepest descent method \cite{DeiftZ,Yan,Wang}. Zhang and Yan used the Riemann--Hilbert (RH) approach to study the $N$-soliton solutions for the focusing and defocusing mKdV equation with nonzero boundary conditions \cite{ZhangYan}. Chen and Liu derived long time
asymptotics for the focusing and defocusing mKdV equation with initial data belongs to weighted Sobolev space $H^{2,1}(\R)$ \cite{CL1,CL2}. The $N$-soliton solution of the integrable nonlinear wave equation obtained by the RH approach provides a framework of nonlinear superposition for us to further analyze the elastic interaction. As $N\to\infty$, the $N$-soliton solution can create irregular statistical ensemble which can be regarded as a soliton gas.

The concept of soliton gas, representing the infinite statistical ensemble of interacting solitons, was initially introduced by Zakharov in the KdV equation \cite{Zak71}. Specifically, Zakharov's pioneering work focused on the interaction between a single soliton and an infinite number of widely spaced and isolated solitons, collectively referred to as a rarefied soliton gas. Building on this foundation, El et al. extended Zakharov's model to describe a dense KdV soliton gas by employing spectral theory to analyze the thermodynamic limit of finite-gap solutions \cite{E1,E2,E3}. Later on, this spectral theory approach was also applied to characterize soliton and breather gases for the focusing nonlinear Schr\"{o}dinger equation \cite{E4,TW2022}. In the past few years, researchers have developed a rapidly growing interest on the study of soliton gases as their dynamics are crucial for understanding several fundamental nonlinear wave phenomena \cite{SRADE2024}. In particular, these phenomena include spontaneous modulation instability and the formation of rogue waves \cite{GA,GAZ}. Recently, asymptotic analysis of soliton gas has become a focal point of research. Girotti et al. conducted a comprehensive study on the large-space and large-time asymptotics of the KdV soliton gas \cite{Girotti-1}. Their research was based on the $N$-soliton solution, which is characterized by the RH problem. This new perspective on soliton gas was inspired by the concept of the primitive potential introduced by Dyachenko, Zakharov et al. in \cite{DZZ}. The large $x$ asymptotics of the soliton gas of the nonlinear Schr\"{o}dinger equation was studied in \cite{HXF}. In a related development, Bertola et al. \cite{Grava-3} explored the phenomenon of soliton shielding  of the focusing nonlinear Schr\"{o}dinger equation. In  \cite{Girotti-2}, Girotti et al. analyzed the behavior of a dense mKdV soliton gas of genus one, as well as its large-time dynamics in the presence of a single trial soliton. So far, the research has remained in the case of low genus soliton gases, and there is no literature studying the large-time and large-space asymptotics of the mKdV soliton gas of any odd genus.

The main goal of this paper is to study the large-space and large-time asymptotic behavior of the soliton gas of genus $2n-1$ for the focusing mKdV equation \eqref{mKdV} based on the nonlinear steepest descent method, the $g$-function approach and the Riemann-theta function. The main innovations of this paper are as follows. First, we study the large-space and large-time asymptotic behavior of the mKdV soliton gas with any odd genus, and construct the solution of the corresponding model problem for the soliton gas with any odd genus, the leading asymptotic term of the mKdV soliton gas is expressed by the Riemann-theta functions with any odd genus. Second, with the increase of the genus, the large-time asymptotic analysis of the mKdV soliton gas becomes more difficult, and the key to the asymptotic analysis is to introduce the appropriate $g$-functions, and for the asymptotic regions of different genus, we construct the corresponding $g$-functions to replace the phase function $\theta(x,t;k)$. Third, for the high genus modulation region, the branch point $i\beta_m$ and $\xi$ satisfy the Whitham modulation equation such that $\xi$ is a function of $\beta_m$, but it is difficult to prove the monotonicity of this function. Fortunately, we write its derivative in the form of the product of the determinants of $m\times m$ matrices and prove rigorously that their determinants are all positive. Thus, using the implicit function theorem, we show that $\beta_m$ is monotonically increasing with respect to $\xi$. Moreover, we establish the inequality $\xi_{2m-1}<\xi_{2m}$. This also demonstrates that the method developed in this paper is applicable to the proof of the conjecture proposed in~\cite{Wang}.

 The outline of the paper is as follows. In section 2, we first consider a pure $nN$-soliton solution for the mKdV equation, whose spectrum is confined to the intervals $\bigcup\limits_{j=1}^{n}(i(-b_j,-a_j)\cup i(a_j,b_j))$, where $0<a_1<b_1<\cdots<a_n<b_n$. As $N\to +\infty$, we construct the RH problem for the mKdV soliton gas of genus $2n-1$. In section 3, we consider the large-space asymptotics of the mKdV soliton gas. In sections 4--6, we establish  a global long-time asymptotic description of the soliton gas solution $u(x,t)$ using the nonlinear steepest descent method and the  $g$-function approach. We show that there are $2n+1$ fundamental spatial domains, in which the potential $u(x,t)$ displays different asymptotic behaviors depending on the value of the parameter $\xi = x/t$. For $\xi<4a_1^2$, the potential decays exponentially. For $4a_1^2<\xi < \xi_{1}$, the potential is described by the periodic travelling wave with modulated coefficients. For $\xi_1<\xi<\xi_2$, the asymptotics is connected by a periodic travelling wave with fixed coefficients. For $\xi_{2m-2}<\xi<\xi_{2m-1}$, the leading-order term can be expressed with a  Riemann-theta function of genus $2m-1$ with modulated coefficients, while for $\xi_{2m-1}<\xi<\xi_{2m}$, the asymptotics is expressed by a  Riemann-theta function of genus $2m-1$ with fixed coefficients, where $m=2,3,\ldots,n$.			
\section{A RH characterization of the mKdV equation}
	The mKdV equation \eqref{mKdV} admits the Lax pair
	\begin{equation}\label{lax}
		\begin{aligned}
			&\Phi_x=\begin{pmatrix} -ik & u(x,t)\\
				-u(x,t) & ik \end{pmatrix}\Phi,\\
			&\Phi_t =\begin{pmatrix} -4ik^3+2{ik u^2}& 4k^2 u+2ik u_x-2u^3-u_{xx}\\ -4k^2 u+2ik u_x+2u^3+u_{xx}& 4ik^3-2ik u^2 \end{pmatrix}\Phi,
		\end{aligned}
	\end{equation}
where $\Phi(x,t;k)$ is a $2\times2$ matrix-valued function of $x, t$ and the spectral parameter  $k \in\C$.
The RH problem for the solution of the mKdV equation is described as follows:
\begin{RHP}\label{RHP1}
Find  a $2\times2$ matrix-valued function $M(x,t;k)$ with the following properties
\begin{enumerate}
\item $M(x,t;k)$ is meromorphic in $\mathbb{C}$, with simple poles at $i\kappa_{j} $ in $i\mathbb{R}_{+}$, and at the corresponding conjugate points $-i\kappa_{j} $ in $i\mathbb{R}_{-}$,\ ${j=1,2,\ldots,n N},\  n, N\in \Z_+$.
			
\item $M(x,t;k)$ satisfies the residue conditions
		\begin{equation}
			\begin{aligned}\label{residue_soliton}
\Res\limits_{k=i\kappa_{j}}{M}(k)&=\lim_{k\to i\kappa_{j}}{M}(k)\begin{pmatrix}0&0\\-i\chi_{j}e^{-2i\theta(x,t;k)}&0\\\end{pmatrix},\\
\Res\limits_{k=-i\kappa_j}{M}(k)&
=\lim_{k\to-i\kappa_j}{M}(k)\begin{pmatrix}0&-i\chi_je^{2i\theta(x,t;k)}\\0&0\end{pmatrix} ,
			\end{aligned}
		\end{equation}
		where $\theta(x,t;k)=4tk^3+xk$ and $\chi_{j}$ is a nonzero real constant.
\item $M(x,t;k)\rightarrow I ,\  k\rightarrow\infty$.
\item $M(x,t;k)$ satisfies the symmetry
		\[
		\ M(k)=\overline{ M(-\overline{k})}=\begin{pmatrix}
			0&-1\\1&0\end{pmatrix}M(-k)\begin{pmatrix}
			0&1\\-1&0\end{pmatrix}\ .
		\]
\end{enumerate}
\end{RHP}
The potential $u(x,t)$ is determined from $M(x,t;k)$ via
	\begin{equation}
		u(x,t)=2i\lim\limits_{k\to\infty}\le(k M(x,t;k)\ri)_{12}.
	\end{equation}
Without loss of generality,  it is assumed in the following text that $\chi_{j}>0$. As $N \to + \infty$, we assume that
	\begin{enumerate}
		\item  The poles $ \{i\kappa_j\}_{j=(l-1)N+1}^{l N}$ are given in interval $i[a_l,b_l]$ and satisfy $|\kappa_{j+1}-\kappa_j|=\frac{b_l-a_l}{N},\ j=(l-1)N+1,\ldots, lN-1$, where $l=1,2,\ldots,n$, and $0<a_1<b_1<a_2<b_2<\ldots <a_n<b_n$.
		\item The coefficients $\{i\chi_j\}_{j=1}^{nN}$ are purely imaginary and are assumed to be discretizations of the given functions:
		\begin{gather}
            \chi_j = \frac{ (b_l-a_l) r_l(i\kappa_j)}{N\pi},\ j=(l-1)N+1,\ldots,l N,\ l=1,2,\ldots,n,
		\end{gather}
		where $r_l(k)$ are non-vanishing,  real valued and positive functions for $k\in i( a_{l},  b_{l})$ and analytic near the intervals $i( a_{l},  b_{l}),\ l=1,2\ldots,n $.
	\end{enumerate}
	
	We introduce the closed curves $\Gamma_{l+}$ which located in the upper half-plane $\mathbb{C}_{+}$, encircling the poles  $i\kappa_{(l-1)N+1},\ldots,i\kappa_{l N}$ in a counterclockwise direction. Similarly, $\Gamma_{l-}$ are the clockwise curves that surrounds the poles $-i\kappa_{(l-1)N+1},\ldots,-i\kappa_{l N}$ in the lower half-plane $\mathbb{C}_{-}$ with $\ l=1,2,\ldots, n$.
We first remove the poles by defining
	\begin{equation}
		Z(k) = M(k)
		\begin{cases}
\begin{pmatrix} 1 & 0\\ \frac{i (b_l-a_l) }{N\pi}\displaystyle\sum_{j=(l-1)N+1}^{l N}\frac{ r_l(i\kappa_j)e^{-2i\theta(x,t;k) }}{k - i\kappa_j} & 1\end{pmatrix},
			&\text{$k$ inside $\Gamma_{l+},\  l=1,2,\ldots, n,$}\\[5ex]
\begin{pmatrix} 1 & \frac{i (b_l-a_l) }{N\pi}\displaystyle\sum_{j=(l-1)N+1}^{lN}\frac{ r_l(i\kappa_j)e^{2i\theta(x,t;k) }}{k +i\kappa_j}\\0 & 1\end{pmatrix},
			&\text{$k$ inside $\Gamma_{l-},\  l=1,2,\ldots, n$},\\
			 I ,
			&\text{otherwise},
		\end{cases}
	\end{equation}
where $ I $ is the $2\times2$ identity matrix.	
Then the $2\times 2$ matrix-valued function $Z(x,t;k)$ satisfies following jump relation
\begin{gather}
Z_+(k) = Z_-(k)\begin{cases}
\begin{pmatrix} 1 & 0\\ \frac{i (b_l-a_l) }{N\pi}\displaystyle\sum_{j=(l-1)N+1}^{lN}\frac{ r_l(i\kappa_j)e^{-2i\theta(x,t;k) }}{k - i\kappa_j} & 1\end{pmatrix},&\text{$k\in\Gamma_{l+},\  l=1,2,\ldots, n$},\\[5ex]
\begin{pmatrix} 1 & -\frac{i (b_l-a_l) }{N\pi}\displaystyle\sum_{j=(l-1)N+1}^{lN}\frac{ r_l(i\kappa_j)e^{2i\theta(x,t;k) }}{k + i\kappa_j}\\0 & 1\end{pmatrix},
			&\text{$k\in\Gamma_{l-},\ l=1,2,\ldots, n$ },\\
		\end{cases}
\end{gather}
where, for $k \in \Gamma_{l+}\cup\Gamma_{l-},\ l=1,2,\ldots, n$, the boundary values $Z_{+}(k)$ are taken from the left side of the contour  and the boundary values $Z_{-}(k)$ are taken from the right.
As $N \to \infty$, the following proposition holds.
\begin{proposition}
For any open set $K_{l+}$ containing the interval $i[ a_{l},b_{l}]$, and any open set  $K_{l-}$ containing the interval $i[- b_l,-a_l],\  l=1,2,\ldots, n$, the following limit holds uniformly for all  $k \in \mathbb{C}\backslash (\bigcup\limits_{l=1}^{n}K_{l+})$:
		\begin{gather}
			\lim_{N\to +\infty} \frac{i(b_l-a_l)}{N\pi}\sum_{j=(l-1)N+1}^{lN}\frac{r_l(i\kappa_j) }{k -i\kappa_j} =\int_{ia_l}^{ib_l} \frac{2i  r_l(\zeta)}{k - \zeta}\frac{\d \zeta}{2\pi i},\ l=1,2,\ldots, n,
		\end{gather}
and the following limit holds uniformly for all $k \in \mathbb{C} \backslash (\bigcup\limits_{l=1}^{n} K_{l-})$:
\begin{gather}
			\lim_{N\to +\infty} \frac{i(b_l-a_l)}{N\pi}\sum_{j=(l-1)N+1}^{lN}\frac{r_l(-i\kappa_j) }{k +i\kappa_j} = \int_{-ib_l}^{-ia_l} \frac{2i  r_l(\zeta)}{k - \zeta}\frac{\d \zeta}{2\pi i}, \ l=1,2,\ldots, n.
\end{gather}
\end{proposition}
\begin{proof}
The proof is similar as Proposition 2.3 in \cite{Girotti-1}.
\end{proof}
	According to the proposition above and a small norm argument, the jump conditions of $Z(k)$ transform into the following formula
	\begin{gather}
				Z_+(k) = Z_-(k)\begin{cases}
\begin{pmatrix} 1 & 0\\ e^{-2i\theta(x,t;k)}\int_{ia_j}^{ib_j} \frac{2i  r_j(\zeta)}{k - \zeta}\frac{\d \zeta}{2\pi i} & 1\end{pmatrix},
			&\text{$k\in\Gamma_{j+},\  j=1,2, \ldots, n $},\\[5ex]
\begin{pmatrix} 1 & -e^{2i\theta(x,t;k) }\int_{-ib_j}^{-ia_j} \frac{2i  r_j(\zeta)}{k - \zeta}\frac{\d \zeta}{2\pi i}\\0 & 1\end{pmatrix},
			&\text{$k\in\Gamma_{j-},\  j=1,2, \ldots, n $}.
		\end{cases}
			\end{gather}
	Next, we define the following transform to eliminate its jump on the contours $\Gamma_{j+}$ and $\Gamma_{j-}$,
	\begin{equation}
		X(k) =Z(k)
		\begin{cases}
\begin{pmatrix} 1 & 0\\ -e^{-2i\theta(x,t;k) }\int_{ia_j}^{ib_j} \frac{2i  r_j(\zeta)}{k - \zeta}\frac{\d \zeta}{2\pi i} & 1\end{pmatrix},
			&\text{$k$ inside the loop $\Gamma_{j+},\  j=1,2, \ldots, n $},\\[5ex]
			\begin{pmatrix} 1 & -e^{2i\theta(x,t;k) }\int_{-ib_j}^{-ia_j} \frac{2i  r_j(\zeta)}{k - \zeta}\frac{\d \zeta}{2\pi i}\\0 & 1\end{pmatrix},
			&\text{$k$ inside the loop $\Gamma_{j-},\  j=1,2, \ldots, n $},\\
			 I ,
			&\text{elsewhere}.
		\end{cases}
	\end{equation}
	Because the integrals have jumps across the intervals $i(a_j,b_j)$ and $i(-b_j,-a_j), j=1,2,\ldots, n$, there still have jumps across those intervals. We obtain a RH  problem for $X(x,t;k)$ by using the Sokhotski-Plemelj formula.
\begin{RHP}\label{RHP2}
Find a $2\times 2$ matrix-valued function $X(x,t;k)$ with the following properties
\begin{enumerate}
\item X(x,t;k) is analytic for $k\in \C \backslash (\bigcup\limits_{j=1}^{n}(i(a_j,b_j)\cup i(-b_j,-a_j)))$.
\item For $k \in \bigcup\limits_{j=1}^{n}(i(a_j,b_j)\cup i(-b_j,-a_j))$, the boundary values $X_{+}(k)$ are taken from the left side of the contour  and the boundary values $X_{-}(k)$ are taken from the right, then we have the following jump relation
	\begin{align}
		&X_{+}(k) = X_{-}(k) \begin{cases}
			\begin{pmatrix} {1}& 0\\   {2i   r_j(k) e^{ -2 i\theta(x,t;k)} }  & {1} \end{pmatrix}, &  k \in i(a_j,b_j),\ j=1,2,\ldots, n,\\
			\begin{pmatrix} {1}& {2i   r_j(k) e^{ 2 i\theta(x,t;k)}}\\0  & {1} \end{pmatrix}, &  k \in i(-b_j,-a_j),\ j=1,2,\ldots, n,
		\end{cases}
\end{align}	
where the contours $i(a_j,b_j)$ and $i(-b_j,-a_j)$ are  both oriented upwards.
\item $X(k) =  I  + \mathcal{O}\le(\frac{1}{k}\ri) ,\qquad k \rightarrow \infty.$
\end{enumerate}
\end{RHP}	
	As explained in \cite{Girotti-2}, we can recover the soliton gas solution $u(x,t)$ from
	\begin{gather}
		u(x,t)= 2i\lim_{k \to \infty} k X(x,t;k)_{12}.
	\end{gather}
\section{Behavior of the potential $u(x,0)$ as $x\to \pm\infty$}
	Here we consider the case $t=0$, we set $\Sigma_j=i(a_j,b_j)$, $\Sigma_{-j}=i(-b_j,-a_j)$, $j=1,2,\ldots, n$, then $X(x;k)$ satisfies the following jump relation
\begin{align}\label{X1}
		&X_+ (x;k) = X_-(x;k) \begin{cases}
\displaystyle \begin{pmatrix} 1 & 0\\ 2ir_j(k) e^{-2ik x} & 1 \end{pmatrix}, &\quad
k \in  \Sigma_j,\ j=1,2,\ldots, n,\\[3ex]
\displaystyle \begin{pmatrix} 1 & 2i r_j(k)e^{2ik x}\\0  & 1 \end{pmatrix}, & \quad
k \in  \Sigma_{-j},\ j=1,2,\ldots, n.
		\end{cases}
	\end{align}
When $x\to -\infty$, we notice that the off-diagonal elements of the jump matrices in  \eqref{X1} are exponentially decaying. Therefore, by a small norm argument, we find
 $$X(x;k) =  I  + \mathcal{O}\le(e^{-c|x|}\ri),\ x\to -\infty,$$
 where $c$ is a positive constant. Hence, the mKdV soliton gas solution $u(x,0)$ is exponentially decreasing as $x\to -\infty$.
To analyze the large $x$ asymptotics as  $x\to+\infty$, we introduce the scalar  function $g(k)$  which satisfies the following properties:

The function $g(k)$ is analytic for $k\in \C \backslash i[-b_n, b_n]$ and satisfies the jump conditions:
\begin{align}
& g_+(k) + g_-(k) = 2k ,& k \in \bigcup\limits_{j=1}^{n}\le(\Sigma_j\cup\Sigma_{-j}\ri) \label{gconstraint1}, \\
& g_+ (k) - g_-(k) = \Omega_0, & k \in i[-a_1,a_1]  \label{gconstraint2},\\
& g_+ (k) - g_-(k) = \Omega_j, & k \in i[b_{j},a_{j+1}],\ j=1,2,\ldots,n-1  \label{gconstraint3},\\
& g_+ (k) - g_-(k) = \Omega_{-j}, & k \in i[-a_{j+1},-b_{j}],\ j=1,2,\ldots,n-1 \label{gconstraint4},\\
& g(k) =\mathcal{O}\le(\frac{1}{k}\ri), & k \to  \infty \, \label{gconstraint5}.
\end{align}
The function $g(k)$ is given by
\begin{equation}\label{gk}
g(k)=k-\int_{ib_n}^{k}\frac{\zeta^{2n}+c_1\zeta^{2(n-1)}+\cdots+c_{n-1}\zeta^2+c_n}{R(\zeta)}\d\zeta,
\end{equation}
where $R(k)=\sqrt{\prod\limits_{j=1}^n(k^2+a_j^2)(k^2+b_j^2)}$, and we choose the cuts of $R(k)$ to be the intervals $\Sigma_j$ and $\Sigma_{-j}$, $j=1,2,\ldots, n$.
The constants $c_j$ are chosen so that
\begin{align}
&\int_{-ia_1}^{ia_1}\frac{\zeta^{2n}+c_1\zeta^{2(n-1)}+\cdots+c_{n-1}\zeta^2+c_n}{R(\zeta)}\d\zeta=0,\label{c1}\\
&\int_{ib_j}^{ia_{j+1}}\frac{\zeta^{2n}+c_1\zeta^{2(n-1)}+\cdots+c_{n-1}\zeta^2+c_n}{R(\zeta)}\d\zeta=0,\ j=1,2,\ldots,n-1.\label{c2}
\end{align}
To represent $c_{j}$ exactly, we first define the following $n\times n$ matrices $A$ and $D_j$, $j=1,2,\ldots, n$,
\begin{align*}
A=\begin{pmatrix}
\int_{ib_{n-1}}^{ia_{n}}\frac{\zeta^{2(n-1)}}{R(\zeta)}\d\zeta & \int_{ib_{n-1}}^{ia_{n}}\frac{\zeta^{2(n-2)}}{R(\zeta)}\d\zeta &
\cdots &
\int_{ib_{n-1}}^{ia_{n}}\frac{1}{R(\zeta)}\d\zeta\\
\int_{ib_{n-2}}^{ia_{n-1}}\frac{\zeta^{2(n-1)}}{R(\zeta)}\d\zeta & \int_{ib_{n-2}}^{ia_{n-1}}\frac{\zeta^{2(n-2)}}{R(\zeta)}\d\zeta &
\cdots &
\int_{ib_{n-2}}^{ia_{n-1}}\frac{1}{R(\zeta)}\d\zeta\\
\cdots&\cdots&\cdots&\cdots\\
\int_{-ia_1}^{ia_1}\frac{\zeta^{2(n-1)}}{R(\zeta)}\d\zeta & \int_{-ia_1}^{ia_1}\frac{\zeta^{2(n-2)}}{R(\zeta)}\d\zeta &
\cdots &
\int_{-ia_1}^{ia_1}\frac{1}{R(\zeta)}\d\zeta
\end{pmatrix},
\end{align*}
\begin{align*}
D_1=\begin{pmatrix}
\int_{ib_{n-1}}^{ia_{n}}\frac{-\zeta^{2n}}{R(\zeta)}\d\zeta & \int_{ib_{n-1}}^{ia_{n}}\frac{\zeta^{2(n-2)}}{R(\zeta)}\d\zeta &
\cdots &
\int_{ib_{n-1}}^{ia_{n}}\frac{1}{R(\zeta)}\d\zeta\\
\int_{ib_{n-2}}^{ia_{n-1}}\frac{-\zeta^{2n}}{R(\zeta)}\d\zeta & \int_{ib_{n-2}}^{ia_{n-1}}\frac{\zeta^{2(n-2)}}{R(\zeta)}\d\zeta &
\cdots &
\int_{ib_{n-2}}^{ia_{n-1}}\frac{1}{R(\zeta)}\d\zeta\\
\cdots&\cdots&\cdots&\cdots\\
\int_{-ia_1}^{ia_1}\frac{-\zeta^{2n}}{R(\zeta)}\d\zeta & \int_{-ia_1}^{ia_1}\frac{\zeta^{2(n-2)}}{R(\zeta)}\d\zeta &
\cdots &
\int_{-ia_1}^{ia_1}\frac{1}{R(\zeta)}\d\zeta
\end{pmatrix},
\end{align*}
\begin{equation*}
D_j=\begin{array}{l}\begin{pmatrix}
\scriptstyle\int_{ib_{n-1}}^{ia_{n}}\frac{\zeta^{2(n-1)}}{R(\zeta)}\d\zeta &
\scriptstyle\cdots &
\scriptstyle\int_{ib_{n-1}}^{ia_{n}}\frac{\zeta^{2(n-j+1)}}{R(\zeta)}\d\zeta &
\scriptstyle\int_{ib_{n-1}}^{ia_{n}}\frac{-\zeta^{2n}}{R(\zeta)}\d\zeta &
\scriptstyle\int_{ib_{n-1}}^{ia_{n}}\frac{\zeta^{2(n-j-1)}}{R(\zeta)}\d\zeta &
\scriptstyle\cdots &
\scriptstyle\int_{ib_{n-1}}^{ia_{n}}\frac{1}{R(\zeta)}\d\zeta\\
\scriptstyle\int_{ib_{n-2}}^{ia_{n-1}}\frac{\zeta^{2(n-1)}}{R(\zeta)}\d\zeta &
\scriptstyle\cdots &
\scriptstyle\int_{ib_{n-2}}^{ia_{n-1}}\frac{\zeta^{2(n-j+1)}}{R(\zeta)}\d\zeta &
\scriptstyle\int_{ib_{n-2}}^{ia_{n-1}}\frac{-\zeta^{2n}}{R(\zeta)}\d\zeta &
\scriptstyle\int_{ib_{n-2}}^{ia_{n-1}}\frac{\zeta^{2(n-j-1)}}{R(\zeta)}\d\zeta &
\scriptstyle\cdots &
\scriptstyle\int_{ib_{n-2}}^{ia_{n-1}}\frac{1}{R(\zeta)}\d\zeta\\
\scriptstyle\cdots&\scriptstyle\cdots&\scriptstyle\cdots&\scriptstyle\cdots&\scriptstyle\cdots&\scriptstyle\cdots&\scriptstyle\cdots\\
\scriptstyle\int_{-ia_1}^{ia_1}\frac{\zeta^{2(n-1)}}{R(\zeta)}\d\zeta &
\scriptstyle\cdots &
\scriptstyle\int_{-ia_1}^{ia_1}\frac{\zeta^{2(n-j+1)}}{R(\zeta)}\d\zeta &
\scriptstyle\int_{-ia_1}^{ia_1}\frac{-\zeta^{2n}}{R(\zeta)}\d\zeta &
\scriptstyle\int_{-ia_1}^{ia_1}\frac{\zeta^{2(n-j-1)}}{R(\zeta)}\d\zeta &
\scriptstyle\cdots &
\scriptstyle\int_{-ia_1}^{ia_1}\frac{1}{R(\zeta)}\d\zeta
\end{pmatrix},
\end{array}
\end{equation*}
for $j=2,3,\ldots,n-1$, and
\begin{align*}
D_n=\begin{pmatrix}
\int_{ib_{n-1}}^{ia_{n}}\frac{\zeta^{2(n-1)}}{R(\zeta)}\d\zeta & \int_{ib_{n-1}}^{ia_{n}}\frac{\zeta^{2(n-2)}}{R(\zeta)}\d\zeta &
\cdots &
\int_{ib_{n-1}}^{ia_{n}}\frac{\zeta^{2}}{R(\zeta)}\d\zeta&
\int_{ib_{n-1}}^{ia_{n}}\frac{-\zeta^{2n}}{R(\zeta)}\d\zeta\\
\int_{ib_{n-2}}^{ia_{n-1}}\frac{\zeta^{2(n-1)}}{R(\zeta)}\d\zeta & \int_{ib_{n-2}}^{ia_{n-1}}\frac{\zeta^{2(n-2)}}{R(\zeta)}\d\zeta &
\cdots &
\int_{ib_{n-2}}^{ia_{n-1}}\frac{\zeta^{2}}{R(\zeta)}\d\zeta&
\int_{ib_{n-2}}^{ia_{n-1}}\frac{-\zeta^{2n}}{R(\zeta)}\d\zeta\\
\cdots&\cdots&\cdots&\cdots&\ldots\\
\int_{-ia_1}^{ia_1}\frac{\zeta^{2(n-1)}}{R(\zeta)}\d\zeta & \int_{-ia_1}^{ia_1}\frac{\zeta^{2(n-2)}}{R(\zeta)}\d\zeta &
\cdots &
\int_{-ia_1}^{ia_1}\frac{\zeta^{2}}{R(\zeta)}\d\zeta&
\int_{-ia_1}^{ia_1}\frac{-\zeta^{2n}}{R(\zeta)}\d\zeta
\end{pmatrix}.
\end{align*}
Utilizing Cramer's Law, the constants $c_j$ can be expressed by
\begin{align}\label{c_j}
c_j=\frac{\det D_j}{\det A}, j=1,2,\ldots, n.
\end{align}
From \eqref{gconstraint2}--\eqref{gconstraint4}, the integral constants $\Omega_0$, $\Omega_j$ and $\Omega_{-j}$ can be expressed by
\begin{align}
&\Omega_{n-1}=-2\int_{ib_n}^{ia_n}\frac{\zeta^{2n}+c_1\zeta^{2(n-1)}+\cdots +c_n}{R_+(\zeta)}\d\zeta,\\
&\Omega_{j}=\Omega_{n-1}+\Omega_{n-2}+\cdots+\Omega_{j+1}-2\int_{ib_{j+1}}^{ia_{j+1}}\frac{\zeta^{2n}+c_1\zeta^{2(n-1)}+\cdots +c_n}{R_+(\zeta)}\d\zeta,\ j=0,1,\ldots, n-2,\\
&\Omega_{-j}=\Omega_{j}, \ j=1,2,\ldots,n-1.
\end{align}
Next, we define the function $f(k)$ which satisfies the following conditions:
\begin{align}
&f_+(k) f_-(k)=\frac{1}{2r_j(k)},\ \ &k\in\Sigma_j,\ j=1,2,\ldots,n,\\
&\frac{f_+(k)}{f_-(k)}=e^{i\Delta_j},\ \ &k\in i[b_j,a_{j+1}],\ j=1,2,\ldots,n-1,\\
&\frac{f_+(k)}{f_-(k)}=e^{i\Delta_0},\ \ &k\in i[-a_1,a_1],\\
&\frac{f_+(k)}{f_-(k)}=e^{i\Delta_{-j}},\ \ &k\in i[-a_{j+1},-b_j],\ j=1,2,\ldots,n-1,\\
&f_+(k) f_-(k)=2r_j(k),\ \ &k\in\Sigma_{-j},\ j=1,2,\ldots,n,\\
&f(k)=1+\mathcal{O}(\frac{1}{k}),\ \ &k \rightarrow  \infty.\label{fconstant}
\end{align}
By using the Sokhotski-Plemelj formula, we get
\begin{equation}\label{f}
\begin{aligned}
\begin{array}{l}
f(k)=\exp\le( \frac{R(k)}{2\pi i}\left(\sum\limits_{j=1}^{n}\int_{\Sigma_j}\frac{\log(\frac{1}{2r_j(\zeta)}) }{R_+(\zeta)(\zeta-k)}\d\zeta+
\sum\limits_{j=1}^{n-1}\int_{ ib_j}^{ia_{j+1}}\frac{i\Delta_j}{R(\zeta)(\zeta-k)}\d\zeta
+\int_{-ia_1}^{ia_1}\frac{i\Delta_0}{R(\zeta)(\zeta-k)}\d\zeta\right.\right.\\
\left.\left.\ \ \ \ \ \ \ \
+\sum\limits_{j=1}^{n}\int_{\Sigma_{-j}}\frac{\log(2r_j(\zeta))}{R_+(\zeta)(\zeta-k)}\d\zeta+
\sum\limits_{j=1}^{n-1}\int_{-ia_{j+1}}^{-ib_{j}}\frac{i\Delta_{-j}}{R(\zeta)(\zeta-k)}\d\zeta
\ri) \ri) .
\end{array}
\end{aligned}
\end{equation}
From the normalization condition \eqref{fconstant}, we obtain
\begin{equation}\label{Delta}
\begin{array}{l}
\sum\limits_{j=1}^{n}\int_{\Sigma_j}\frac{\log(\frac{1}{2r_j(\zeta)}) \zeta^l}{R_+(\zeta)}\d\zeta+
\sum\limits_{j=1}^{n-1}\int_{ ib_j}^{ia_{j+1}}\frac{i\Delta_j\zeta^l}{R(\zeta)}\d\zeta
+\int_{-ia_1}^{ia_1}\frac{i\Delta_0\zeta^l}{R(\zeta)}\d\zeta
+\sum\limits_{j=1}^{n}\int_{\Sigma_{-j}}\frac{\log(2r_j(\zeta))\zeta^l}{R_+(\zeta)}\d\zeta\\+
\sum\limits_{j=1}^{n-1}\int_{-ia_{j+1}}^{-ib_{j}}\frac{i\Delta_{-j}\zeta^l}{R(\zeta)}\d\zeta=0,\ l=0,1,\ldots,2n-2.
\end{array}
\end{equation}
Then the constants $\Delta_0$, $\Delta_j$ and $\Delta_{-j}$, $j=1,2,\ldots,n-1$, can be uniquely determined by the system \eqref{Delta}.

Using the functions $g(k)$ and $f(k)$, we define a new matrix-valued function
\begin{gather}
T(k)=X(k)e^{ixg(k)\sigma_3}f(k)^{\sigma_3},
\end{gather}
where
$$f(k)^{\sigma_3}=
\begin{pmatrix}
f(k) & 0 \\
0 & f(k)^{-1}
\end{pmatrix},\ \sigma_3:=\begin{pmatrix} 1&0\\0&-1\end{pmatrix},$$
then the jump relation for $T(k)$ is
	\begin{align}
		\label{RH_T}
		T_+(k)=T_-(k)\begin{cases}
			\displaystyle \begin{pmatrix} e^{ix\le(g_+(k) - g_-(k)\ri)}\dfrac{f_+(k)}{f_-(k)} & 0\\ i & e^{-ix\le(g_+(k) - g_-(k)\ri)}\dfrac{f_-(k)}{f_+(k)} \end{pmatrix},  k \in  \Sigma_j,\ j=1,2,\ldots,n,&\\[3ex]
			\displaystyle \begin{pmatrix} e^{ix\le(g_+(k) - g_-(k)\ri)}\dfrac{f_+(k)}{f_-(k)}&  i\\0 & e^{-ix\le(g_+(k) - g_-(k)\ri)}\dfrac{f_-(k)}{f_+(k)} \end{pmatrix},   k \in \Sigma_{-j},\ j=1,2,\ldots,n,& \\
			\displaystyle \begin{pmatrix} e^{i(x\Omega_j+\Delta_j) } & 0\\0& e^{-i(x\Omega_j+\Delta_j)} \end{pmatrix},  k \in  i[b_{j},a_{j+1}],\ j=1,2,\ldots,n-1,& \\
\displaystyle \begin{pmatrix} e^{i(x\Omega_0+\Delta_0) } & 0\\0& e^{-i(x\Omega_0+\Delta_0)} \end{pmatrix},  k \in  i[-a_1,a_1],&\\
\displaystyle \begin{pmatrix} e^{i(x\Omega_{-j}+\Delta_{-j}) } & 0\\0& e^{-i(x\Omega_{-j}+\Delta_{-j})} \end{pmatrix},  k \in  i[-a_{j+1},-b_{j}],\ j=1,2,\ldots,n-1.&
		\end{cases}
	\end{align}
Assume that the function $r_j(k)$ has analytic continuation off the imaginary axis:
\begin{gather}
\hat{r}_j(k)|_{k\in i[a_j,b_j]}=r_j(k),\ j=1,2,\ldots,n,
\end{gather}
where $\hat{r}_j(k)$ is analytic in the lens around $\Sigma_j$ (see Figure 1).
The jumps for $T(k)$ admit the following factorizations on $k\in \Sigma_j\cup\Sigma_{-j},\ j=1,2,\ldots,n$,
\begin{align*}
&\begin{pmatrix} (2r_j(k))^{-1} f_-(k)^{-2} e^{ix(g_+(k)-g_-(k))} & 0 \\ i  & (2r_j(k))^{-1} f_+(k)^{-2} e^{-ix(g_+(k)-g_-(k))} \end{pmatrix}=\\
&\quad\quad\begin{pmatrix} 1 & -i (2r_j(k))^{-1} f_-(k)^{-2} e^{-2ix(g_-(k)-k)} \\ 0 & 1 \end{pmatrix}
	   \begin{pmatrix} 0 & i \\ i & 0 \end{pmatrix}
	   \begin{pmatrix} 1 &  -i (2r_j(k))^{-1} f_+(k)^{-2} e^{-2ix(g_+(k)-k)}  \\ 0 & 1 \end{pmatrix},
\end{align*}
for $k \in \Sigma_{j},\ j=1,2,\ldots,n$,
\begin{align*}
&\begin{pmatrix} (2r_j(k))^{-1} f_+(k)^{2} e^{ix(g_+(k)-g_-(k))} & i \\ 0  & (2r_j(k))^{-1} f_-(k)^{2} e^{-ix(g_+(k)-g_-(k))} \end{pmatrix}=\\
&\quad\quad\begin{pmatrix} 1 & 0\\ -i (2r_j(k))^{-1} f_-(k)^{2} e^{2ix(g_-(k)-k)} & 1 \end{pmatrix}
	   \begin{pmatrix} 0 & i \\ i & 0 \end{pmatrix}
	   \begin{pmatrix} 1 & 0\\ -i (2r_j(k))^{-1} f_+(k)^{2} e^{2ix(g_+(k)-k)}  & 1 \end{pmatrix},
\end{align*}
for $k \in \Sigma_{-j},\ j=1,2,\ldots,n$.

Using the decomposition above, we define
\begin{equation}\label{S}
	 S(k) =T(k)
	\begin{cases}
	   \begin{pmatrix} 1 & -i (2\hat{r}_j(k))^{-1} f(k)^{-2} e^{-2i x(g(k)-k)} \\ 0 & 1 \end{pmatrix},
	& k\in \text{lens right of $\Sigma_{j},\ j=1,2,\ldots,n$}, \\
	   \begin{pmatrix} 1 & i (2\hat{r}_j(k))^{-1} f(k)^{-2} e^{-2i x(g(k)-k)} \\ 0 & 1 \end{pmatrix},
	& k\in \text{lens left of $\Sigma_{j},\ j=1,2,\ldots,n$}, \\
	   \begin{pmatrix} 1 & 0 \\  \frac{-i f(k)^2 e^{2i x(g(k)-k)}}{2\hat{r}_j(k)} & 1 \end{pmatrix},
	& k\in \text{lens right of $\Sigma_{-j},\ j=1,2,\ldots,n$}, \\
	   \begin{pmatrix} 1 & 0 \\  \frac{i f(k)^2 e^{2i x(g(k)-k)}}{2\hat{r}_j(k)} & 1 \end{pmatrix},
	& k\in \text{lens left of $\Sigma_{-j},\ j=1,2,\ldots,n$}, \\
	   I , & \text{elsewhere .}
	\end{cases}
\end{equation}
Then the matrix-valued function $S(k)$ satisfies the following RH problem.
\begin{RHP}\label{RHP3}
Find a $2\times 2$ matrix-valued function $S(x;k)$ with the following properties
\begin{enumerate}
\item $S(x;k)$ is analytic for $k\in \C\backslash( i[-b_n,b_n]\cup \bigcup\limits_{j=1}^{n}(\mathcal{ C}_j\cup\mathcal{ C}_{-j}))$, where the jump contours $\mathcal{C}_{j}$ and $\mathcal{C}_{-j}$ are depicted in Figure 1.
\item For $k\in i[-b_n,b_n]\cup \bigcup\limits_{j=1}^{n}(\mathcal{ C}_j\cup\mathcal{ C}_{-j})$, the boundary values $S_{+}(k)$ are taken from the left side of the contour  and the boundary values $S_{-}(k)$ are taken from the right, then we have the following jump relation
\begin{equation}
\label{VS}
\begin{split}
S_+(k)=S_-(k)\begin{cases}
\begin{pmatrix} 1 &  \dfrac{-ie^{-2i x(g(k)-k)}}{2\hat{r}_j(k)f(k)^{2}} \\ 0 & 1 \end{pmatrix},
	 &k\in\mathcal{C}_j ,\ j=1,2,\ldots,n, \\
\begin{pmatrix} 1 & 0 \\  \dfrac{-i f(k)^2 e^{2i x(g(k)-k)}}{2\hat{r}_j(k)} & 1 \end{pmatrix},
	 &k\in  \mathcal{C}_{-j},\ j=1,2,\ldots,n, \\
\begin{pmatrix} 0 & i \\ i & 0 \end{pmatrix},
	 &k\in \bigcup\limits_{j=1}^{n}(\Sigma_{j}\cup \Sigma_{-j}), \\
\begin{pmatrix} e^{i(x{\Omega_0}+\Delta_0)} & 0 \\ 0 & e^{-i(x{\Omega_0}+\Delta_0)} \end{pmatrix},
 &k\in i[-a_1,a_1],\\
\begin{pmatrix} e^{i(x{\Omega_j}+\Delta_j)} & 0 \\ 0 & e^{-i(x{\Omega_j}+\Delta_j)} \end{pmatrix},
 &k\in i[b_j,a_{j+1}],\ j=1,2,\ldots,n-1,\\
\begin{pmatrix} e^{i(x{\Omega_{-j}}+\Delta_{-j})} & 0 \\ 0 & e^{-i(x{\Omega_{-j}}+\Delta_{-j})} \end{pmatrix},
 &k\in i[-a_{j+1},-b_j],\ j=1,2,\ldots,n-1.\\
	\end{cases}
\end{split}
\end{equation}

\item $S(k) =  I  + \mathcal{O}\le(\frac{1}{k}\ri) ,\qquad k \rightarrow \infty.$
\end{enumerate}
\end{RHP}
In order to eliminate the jumps on $\mathcal{C}_{j}$ and $\mathcal{C}_{-j}$, we need the following lemma.
\begin{lemma}
The following inequalities hold
\begin{align}
\label{in1}
&\operatorname{Re} 2i\le(g(k)-k\ri) >0\, ,\quad k\in \bigcup_{j=1}^{n}{\mathcal C}_j \backslash\{ia_1,i b_1,\ldots,ia_n,i b_n\},\\
\label{in2}
&\operatorname{Re} 2i\le(g(k)-k \ri) <0\, ,\quad  k\in \bigcup_{j=1}^{n}{\mathcal C}_{-j}\backslash\{-ia_1,-i b_1,\ldots,-ia_n,-i b_n \}.
\end{align}
\end{lemma}
\begin{proof}
From equations \eqref{gk}, \eqref{c1} and \eqref{c2}, We can rewrite the derivative of the function $g(k)-k$ as
\begin{gather}
g'(k)-1=-\frac{\prod\limits_{j=0}^{n-1}(k^2+k_j^2)}{R(k)},
\end{gather}
where $k_0\in (0, a_1)$ and $k_j \in (b_j, a_{j+1}),\ j=1,\ldots, n-1$. It follows that, for $k\in \Sigma_{j},\ j=1,\ldots, n$,
\begin{gather}
\Im (g_+'(k)-1)=-\Im (g_-'(k)-1)>0.
\end{gather}
Hence, $\Im (g(k)-k)<0$ on the left and right of $\Sigma_{j}$ and  $k \in \mathcal{C}_{j}$, which implies that \eqref{in1} holds. The inequality \eqref{in2} can be proven similarly.
\end{proof}
Using the lemma above, we know that the off-diagonal entries in the jumps along the left and right lenses are exponentially small as $x\to +\infty$ and that these jumps asymptotically approximate the unit matrix outside a small neighborhood of the endpoints  $\pm i a_j$ and $\pm i b_j$, $j=1,2,\ldots,n$. Then we get the following model problem.
\begin{RHP}\label{RHP4}
Find a $2\times 2$ matrix-valued function $S^{\infty}(k)$ satisfies the following properties
\begin{enumerate}
\item $S^{\infty}(k)$ is analytic in $\mathbb{C}\setminus i[-b_n,b_n]$.
\item For $k\in i[-b_n, b_n] $, the boundary values $S^{\infty}_{\pm}(k)$ satisfy the following jump relation
\begin{gather}
\label{Sinfinity1}
S^{\infty}_+(k) = S^{\infty}_-(k)
\begin{cases}
\begin{pmatrix} 0 & i \\ i & 0 \end{pmatrix},
	& k\in \Sigma_{j}\cup \Sigma_{-j} ,\ j=1,2,\ldots,n, \\
\begin{pmatrix} e^{i(x{\Omega_0}+\Delta_0)} & 0 \\ 0 & e^{-i(x{\Omega_0}+\Delta_0)} \end{pmatrix},
& k\in i[-a_1,a_1],\\
\begin{pmatrix} e^{i(x{\Omega_j}+\Delta_j)} & 0 \\ 0 & e^{-i(x{\Omega_j}+\Delta_j)} \end{pmatrix},
& k\in i[b_j,a_{j+1}],\ j=1,2,\ldots,n-1,\\
\begin{pmatrix} e^{i(x{\Omega_{-j}}+\Delta_{-j})} & 0 \\ 0 & e^{-i(x{\Omega_{-j}}+\Delta_{-j})} \end{pmatrix},
& k\in i[-a_{j+1},-b_j],\ j=1,2,\ldots,n-1.\\
\end{cases}
\end{gather}
 \item
$
 S^{\infty}(k)= I +\mathcal{O}\le(\frac{1}{k}\ri),\quad k\to\infty.
$
\end{enumerate}
\end{RHP}
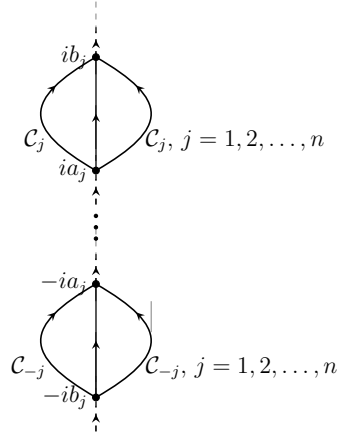
\begin{figure}
\centering
\scalebox{.75}{
\begin{tikzpicture}[>=stealth]op
\path (0,0) coordinate (O);
\coordinate (eta0) at (0,0.4);    \coordinate (eta0c) at ($-1*(eta0)$);
\coordinate (eta1) at (0,1);    \coordinate (eta1c) at ($-1*(eta1)$);
\coordinate (eta2) at (0,3);       \coordinate (eta2c) at ($-1*(eta2)$);
\coordinate (eta3) at (0,3.6);   \coordinate (eta3c) at ($-1*(eta3)$);
\coordinate (eta4) at (0,0.0);
\coordinate (eta5) at (0,0.2); \coordinate (eta5c) at ($-1*(eta5)$);

\draw[->- = .5, dashed,  thick] (eta2)--(eta3)
  node[pos=0.5, pin={ }] {};
\draw[->- = .5,  thick] (eta1)--(eta2)
  node[pos=0.5, pin={}] {};

\draw[->- = .5, dashed, thick] (eta3c)--(eta2c)
  node[pos=0.5, pin={ }] {};
\draw[->- = .5,  thick] (eta2c)--(eta1c)
  node[pos=0.5, pin={}] {};

\draw[->- = .5, dashed,  thick] (eta0)--(eta1)
  node[pos=0.5, pin={}] {};
\draw[->- = .5, dashed,  thick] (eta1c)--(eta0c)
  node[pos=0.5, pin={}] {};

\draw[->- = .7,thick] (eta1) .. controls + (30:1.5cm) and + (-30:1.5cm) .. (eta2)
  node[pos=.25, right] {$\mathcal{C}_{j}$, $j=1,2,\ldots,n$};
\draw[->- = .7,thick] (eta1) .. controls + (150:1.5cm) and + (-150:1.5cm) .. (eta2)
  node[pos=.25, left] {$\mathcal{C}_{j}$};

\draw[->- = .7,thick] (eta2c) .. controls + (30:1.5cm) and + (-30:1.5cm) .. (eta1c)
  node[pos=.25, right] {$\mathcal{C}_{-j}$, $j=1,2,\ldots,n$}
  node[pos=0.5, pin={}] {};
\draw[->- = .7,thick] (eta2c) .. controls + (150:1.5cm) and + (-150:1.5cm) .. (eta1c)
  node[pos=.25, left] {$\mathcal{C}_{-j}$};

\draw[fill] (eta4) circle [radius=0.03];
\draw[fill] (eta5) circle [radius=0.03];
\draw[fill] (eta5c) circle [radius=0.03];
\draw[fill] (eta2) circle [radius=0.06] node[left] {$ib_{j}$};
\draw[fill] (eta1) circle [radius=0.06] node[left] {$ia_{j}$};

\draw[fill] (eta2c) circle [radius=0.06] node[left] {$-i b_{j}$};
\draw[fill] (eta1c) circle [radius=0.06] node[left] {$-i a_{j}$};

\end{tikzpicture}
 }

\label{openinglenses1}
\caption{Opening lenses and the jump contours for $S(k)$.}
\end{figure}

In order to construct the solution of RH problem \ref{RHP4}, we first introduce a scalar function $h(x;k)$ which satisfies the following properties:
\begin{align}
& h_+(k) + h_-(k) = 0 ,& k \in \bigcup_{j=1}^{n}(\Sigma_{j}\cup\Sigma_{-j}), \\
& h_+ (k) - h_-(k) = i(x\Omega_0+\Delta_0), & k \in i[-a_1,a_1]  ,\\
& h_+ (k) - h_-(k) = i(x\Omega_j+\Delta_j), & k \in i[b_{j},a_{j+1}],\ j=1,2,\ldots,n-1,\\
& h_+ (k) - h_-(k) = i(x\Omega_{-j}+\Delta_{-j}), & k \in i[-a_{j+1},-b_{j}],\ j=1,2,\ldots,n-1.
\end{align}
Using the Sokhotski-Plemelj formula, the scalar function $h(x;k)$ can be expressed as
\begin{equation}
\begin{array}{l}
h(x;k)=\frac{R(k)}{2\pi i}\le(\sum\limits_{j=1}^{n-1}\int_{ ib_{j}}^{ia_{j+1}}\frac{i(x\Omega_j+\Delta_j)}{R(\zeta)(\zeta-k)}\d\zeta+\int_{ -ia_1}^{ia_1}\frac{i(x\Omega_0+\Delta_0)}{R(\zeta)(\zeta-k)}\d\zeta+
 \sum\limits_{j=1}^{n-1}\int_{ -ia_{j+1}}^{-ib_{j}}\frac{i(x\Omega_{-j}+\Delta_{-j})}{R(\zeta)(\zeta-k)}\d\zeta\ri).
 \end{array}
\end{equation}
Suppose that $R(k)$ has the following asymptotic expansion as k tends to infinity
\begin{gather}
R(k)=k^{2n}+R_1 k^{2(n-1)}+R_2 k^{2(n-2)}+\cdots+R_n+\mathcal{O}\le(\frac{1}{k^2}\ri),\ k\to \infty,
\end{gather}
where $R_{j}$ can be expressed by $a_j$ and $b_{j}$, $j=1,2,\ldots,n$.
Next, we consider the asymptotic expansion of $h(k)$ as $k$ approaches infinity,
\begin{align}
h(k)=h_{2n-1} k^{2n-1}+h_{2n-2} k^{2n-2}+\cdots+h_0+\mathcal{O}\le(\frac{1}{k}\ri),\ k\to \infty,
\end{align}
where
\begin{equation}\label{h0123}
\begin{aligned}\begin{array}{l}
h_{2n-1}=-\frac{1}{2\pi }\le(\sum\limits_{j=1}^{n-1}\int_{ ib_{j}}^{ia_{j+1}}\frac{x\Omega_j+\Delta_j}{R(\zeta)}\d\zeta+\int_{ -ia_1}^{ia_1}\frac{x\Omega_0+\Delta_0}{R(\zeta)}\d\zeta+\sum\limits_{j=1}^{n-1}\int_{ -ia_{j+1}}^{-ib_{j}}\frac{x\Omega_{-j}+\Delta_{-j}}{R(\zeta)}\d\zeta\ri),\\
h_{2n-1-2l}=-\frac{1}{2\pi }\le(\sum\limits_{j=1}^{n-1}\int_{ ib_{j}}^{ia_{j+1}}\frac{(x\Omega_j+\Delta_j)(\zeta^{2l}+\sum\limits_{p=1}^{l}R_p\zeta^{2(l-p)})}{R(\zeta)}\d\zeta+\int_{ -ia_1}^{ia_1}\frac{(x\Omega_0+\Delta_0)(\zeta^{2l}+\sum\limits_{p=1}^{l}R_p\zeta^{2(l-p)})}{R(\zeta)}\d\zeta\right.\\
\left.\ \ \ \ \ \ \ \  \ \ \ \ \ \ \ \ \
+\sum\limits_{j=1}^{n-1}\int_{ -ia_{j+1}}^{-ib_{j}}\frac{(x\Omega_{-j}+\Delta_{-j})(\zeta^{2l}+\sum\limits_{p=1}^{l}R_p\zeta^{2(l-p)})}{R(\zeta)}\d\zeta\ri),\ l=1,2,\ldots,n-1,\\
h_{2n-2}=-\frac{1}{2\pi }\le(\sum\limits_{j=1}^{n-1}\int_{ ib_{j}}^{ia_{j+1}}\frac{(x\Omega_j+\Delta_j)\zeta}{R(\zeta)}\d\zeta+\int_{ -ia_1}^{ia_1}\frac{(x\Omega_0+\Delta_0)\zeta}{R(\zeta)}\d\zeta+\sum\limits_{j=1}^{n-1}\int_{ -ia_{j+1}}^{-ib_{j}}\frac{(x\Omega_{-j}+\Delta_{-j})\zeta}{R(\zeta)}\d\zeta\ri), \\
h_{2n-2-2l}=-\frac{1}{2\pi }\le(\sum\limits_{j=1}^{n-1}\int_{ ib_{j}}^{ia_{j+1}}\frac{(x\Omega_j+\Delta_j)(\zeta^{2l}+\sum\limits_{p=1}^{l}R_p\zeta^{2(l-p)})\zeta}{R(\zeta)}\d\zeta+\int_{ -ia_1}^{ia_1}\frac{(x\Omega_0+\Delta_0)(\zeta^{2l}+\sum\limits_{p=1}^{l}R_p\zeta^{2(l-p)})\zeta}{R(\zeta)}\d\zeta\right.\\
\left.\ \ \ \ \ \ \ \  \ \ \ \ \ \ \ \ \
+\sum\limits_{j=1}^{n-1}\int_{ -ia_{j+1}}^{-ib_{j}}\frac{(x\Omega_{-j}+\Delta_{-j})(\zeta^{2l}+\sum\limits_{p=1}^{l}R_p\zeta^{2(l-p)})\zeta}{R(\zeta)}\d\zeta\ri), \ l=1,2,\ldots,n-1.
\end{array}
\end{aligned}
\end{equation}

Next, we define a matrix-valued function $Q(x;k)$ to eliminate the jumps over the intervals $i[-a_{j+1}, -b_{j}]$, $i[-a_1,a_1]$ and $i[b_j,a_{j+1}]$,
\begin{align}
Q(x;k)= e^{h_0\sigma_3}S^{\infty}(k) e^{-h(k)\sigma_3},
\end{align}
then $Q(k)$ satisfies
\begin{align}
&Q_+(k)=Q_-(k)\begin{pmatrix}0&i\\i&0  \end{pmatrix},\ \ k \in \bigcup\limits_{j=1}^{n}(\Sigma_j\cup\Sigma_{-j}),\label{Q1}\\
&Q(k)\rightarrow\begin{pmatrix}e^{-\sum\limits_{j=1}^{2n-1}h_{j}k^{j}}&0\\0&e^{\sum\limits_{j=1}^{2n-1}h_{j}k^{j}}\end{pmatrix},  \  k\to \infty\label{Q2}.
\end{align}
Next, we turn to construct a solution with respect to $Q(k)$ that satisfies conditions \eqref{Q1} and \eqref{Q2}. We  first introduce  a two-sheeted  Riemann surface $\mathfrak{X}$ of genus $2n-1$ (see Figure 2),
\begin{align}
\mathfrak{X}=\{(k,\eta)|\eta^2=\prod\limits_{j=1}^{n}(k^2+a_{j}^2)(k^2+b_{j}^2)\},
\end{align}
where the first sheet of $\mathfrak{X}$ is identified by the fact that $R(k)>0$ for $\Im k=0$. We fix a canonical homology basis on $\mathfrak{X}$ by choosing $\mathcal A_{j}$ and $\mathcal A_{n-1+j}$  to encircle $\Sigma_{-(n-j)}$ and $\Sigma_{j}$ clockwise on the first sheet, respectively. We choose $\mathcal B_{j}$  to pass from the positive side of $\Sigma_{-n}$ to $\Sigma_{-(n-j)}$ on sheet 1 and from the negative side of $\Sigma_{-(n-j)}$ to $\Sigma_{-n}$ on the sheet 2, and  $\mathcal B_{n-1+j}$  to pass from the positive side of $\Sigma_{-n}$ to $\Sigma_{j}$ on sheet 1 and from the negative side of $\Sigma_{j}$ to $\Sigma_{-n}$ on the sheet 2 for $j=1,2,\ldots,n-1$.
Set
\begin{align}
\omega_j(k)=\int_{-ia_n}^{k}\psi_j(\zeta)\d \zeta,\ j=1,2,\ldots,2n-1,
\end{align}
where $d \omega_j(\mathcal{P})$ is a basis holomorphic differential on $\mathfrak{X}$:
\begin{align}
\psi_j(\zeta)=\frac{\sum\limits_{l=1}^{2n-1}c_{jl}\zeta^{2n-1-l}}{R(\zeta)},
\end{align}
the coefficients $c_{jl}$ are determined by
$$\oint_{\mathcal{A}_l}\d{\omega}_j(\mathcal{P})=\delta_{jl},\ j,l=1,2,\ldots,2n-1.$$
The entries of $(2n-1)\times (2n-1)$ period matrix $B$ are given by
\begin{gather}
B_{jl}=\oint_{\mathcal{B}_l}\d{\omega}_j(\mathcal{P}),\ j,l=1,2,\ldots,2n-1,
\end{gather}
and the integral paths $\mathcal{A}_l$ and $\mathcal{B}_l$ are shown in Figure 2, then the lattice $\Lambda$ in $\C^{2n-1}$ is generated by the linear combinations, with integer coefficients, of the vectors $\mathbf{e}_j$ and $B\mathbf{e}_j$, where $\mathbf{e}_j=(0,\ldots,0,1,0,\ldots,0)^{\top}$ is the j-th basis vector in $\C^{2n-1}$.
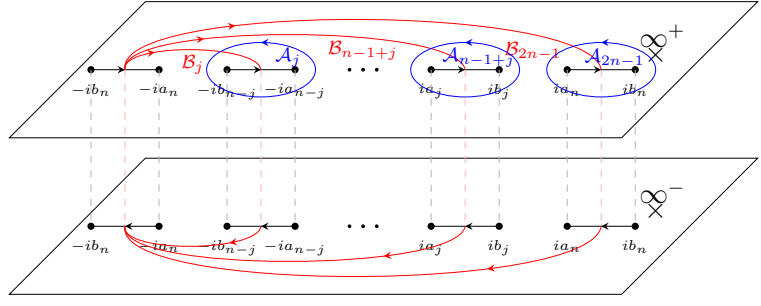
\begin{figure}[th]
\centering
\scalebox{.9}{
\begin{tikzpicture}[>=stealth]
\path (0,0) coordinate (O);

\coordinate (TL) at (-2,5);
\coordinate (TR) at (7,5);
\coordinate (BL) at (-4,3);
\coordinate (BR) at (5,3);
\coordinate (INF1) at (5.6,4.5);

\coordinate (shift) at (0,-2.3);
\coordinate (TL2) at ($  (TL) + (shift)  $);
\coordinate (TR2) at ($  (TR) + (shift)  $);
\coordinate (BL2) at ($  (BL) + (shift)  $);
\coordinate (BR2) at ($  (BR) + (shift)  $);
\coordinate (INF2) at ($ (INF1) + (shift) $);


\coordinate (eta1) at (2.2,4);
\coordinate (eta2) at (3.2,4);
\coordinate (eta3) at (4.2,4);
\coordinate (eta4) at (5.2,4);

\coordinate (-eta1) at (0.2,4);
\coordinate (-eta2) at (-0.8,4);
\coordinate (-eta3) at (-1.8,4);
\coordinate (-eta4) at (-2.8,4);
\draw (TL) -- (TR) -- (BR) -- (BL) -- cycle;
\node[ label={[label distance= -0.3cm, below, xshift= -0.1cm]$\times$}]  at (INF1) {$\infty^+$};
\draw (TL2) -- (TR2) -- (BR2) -- (BL2) -- cycle;
\node[ label={[label distance= -0.3cm, below, xshift= -0.1cm]$\times$}]  at (INF2) {$\infty^-$};

\draw[->-=0.5] (eta1) -- (eta2);
\draw[->-=0.5] (eta3) -- (eta4);
\draw[->-=0.5] (-eta2) -- (-eta1);
\draw[->-=0.5] (-eta4) -- (-eta3);

\draw[->-=0.5] ($(eta2)+(shift)$) -- ($ (eta1) +(shift) $);
\draw[->-=0.5] ($(eta4)+(shift)$) -- ($ (eta3) +(shift) $);
\draw[->-=0.5] ($(-eta1)+(shift)$) -- ($ (-eta2) +(shift) $);
\draw[->-=0.5] ($(-eta3)+(shift)$) -- ($ (-eta4) +(shift) $);

\coordinate (eta0) at (1.2,4);
\coordinate (eta01) at (1,4);
\coordinate (eta02) at (1.4,4);
\draw[fill] (eta0) circle [radius=0.02];
\draw[fill] (eta01) circle [radius=0.02];
\draw[fill] (eta02) circle [radius=0.02];
\coordinate (eta00) at (1.2,1.7);
\coordinate (eta001) at (1,1.7);
\coordinate (eta002) at (1.4,1.7);
\draw[fill] (eta00) circle [radius=0.02];
\draw[fill] (eta001) circle [radius=0.02];
\draw[fill] (eta002) circle [radius=0.02];
\foreach \pos/\label in {eta1/ia_j, eta2/ib_j, eta3/ia_n, eta4/i b_n, -eta1/-ia_{n-j},-eta2/-ib_{n-j},-eta3/-i a_{n}, -eta4/-i b_{n}}{
\node[circle,fill=black, inner sep=0pt,minimum size=3pt,label=below:{\tiny $\label$}] at  (\pos) {};
\node[circle,fill=black, inner sep=0pt,minimum size=3pt,label=below:{\tiny $\label$}] at  ($ (\pos)+(shift) $)  {};
}

\foreach \pos in {eta1,eta2,-eta1,-eta2}{
\draw[dashed, black!30] (\pos) -- ($ (\pos) + (shift) $);
}
\foreach \pos in {eta3,eta4,-eta3,-eta4}{
\draw[dashed, black!30] (\pos) -- ($ (\pos) + (shift) $);
}
\draw[->- = .25, red] ($ 0.5*(-eta4)+0.5*(-eta3)  $) .. controls + (100:1cm) and + (100:1cm) .. ($ 0.5*(eta4)+0.5*(eta3) $);
\draw[->- = .25, red] ($ 0.5*(-eta4)+0.5*(-eta3)  $) .. controls + (100:.7cm) and + (100:.7cm) .. ($ 0.5*(eta2)+0.5*(eta1) $);
\draw[->- = .25, red] ($ 0.5*(-eta4)+0.5*(-eta3)  $) .. controls + (100:.4cm) and + (100:.4cm) .. ($ 0.5*(-eta1)+0.5*(-eta2) $);
\draw[->- = .25, red]  ($ 0.5*(-eta1)+0.5*(-eta2)  + (shift) $) .. controls + (-100:0.4cm) and + (-100:.4cm) .. ($ 0.5*(-eta4)+0.5*(-eta3) +(shift) $);
\draw[->- = .25, red]  ($ 0.5*(eta1)+0.5*(eta2)  + (shift) $) .. controls + (-100:0.7cm) and + (-100:.7cm) .. ($ 0.5*(-eta4)+0.5*(-eta3) +(shift) $);
\draw[->- = .25, red]  ($ 0.5*(eta3)+0.5*(eta4)  + (shift) $) .. controls + (-100:1cm) and + (-100:1cm) .. ($ 0.5*(-eta4)+0.5*(-eta3) +(shift) $);

\draw[red!30, dashed] ($ 0.5*(-eta1)+0.5*(-eta2)  $) -- ++ (shift);
\draw[red!30, dashed] ($ 0.5*(-eta3)+0.5*(-eta4)  $) -- ++ (shift);
\draw[red!30, dashed] ($ 0.5*(eta1)+0.5*(eta2)  $) -- ++ (shift);
\draw[red!30, dashed] ($ 0.5*(eta3)+0.5*(eta4)  $) -- ++ (shift);

\node[above, red] at (-1.3,3.8) {\small $\scriptstyle{\mathcal B}_{j}$};
\node[above, red] at (1.2,4) {\small $\scriptstyle{\mathcal B}_{n-1+j}$};
\node[above, red] at (3.7,4) {\small $\scriptstyle{\mathcal B}_{2n-1}$};
\draw[->- = .25, blue] ( $0.5*(eta1)+0.5*(eta2) $) ellipse (0.8cm and .4cm);
\draw[->- = .25, blue] ( $0.5*(eta3)+0.5*(eta4) $) ellipse (0.8cm and .4cm);
\draw[->- = .25, blue] ( $0.5*(-eta1)+0.5*(-eta2) $) ellipse (0.8cm and .4cm);
\node[blue, above] at (4.9,3.9) {\small $\scriptstyle\mathcal A_{2n-1}$};
\node[blue, above] at (2.9,3.9) {\small $\scriptstyle\mathcal A_{n-1+j}$};
\node[blue, above] at (0.1,3.9) {\small $\scriptstyle\mathcal A_{j}$};
\end{tikzpicture}
}
\caption{The homology basis for the Riemann surface
$\mathfrak{X}$ associated with
$R(k) = \sqrt{ \prod\limits_{j=1}^{n}(k^2+a_{j}^2)(k^2 +b_{j}^2)}$.
}
\label{fig:homology}
\end{figure}

Next, we define the Abel mapping $A_{j}(\mathcal{P}): \mathfrak{X}\rightarrow\C^{2n-1}/\Lambda$ by
\begin{gather}\label{A}
A_{j}(\mathcal{P})=\int_{\mathcal{P}_0}^{\mathcal{P}}\d \omega_j(\mathcal{\tilde{P}}),\ j=1,2,\ldots,2n-1,
\end{gather}
where $\mathcal{\tilde{P}}$ is the integration variable, and the fixed point $\mathcal{P}_0$ satisfies $\tilde{\pi}(\mathcal{P}_0)=-ia_n$, $k=\tilde{\pi}(\mathcal{P})$ is the standard projection of $\mathcal{P}=(k,\eta)\in \mathfrak{X}$ on the Riemann sphere $ \mathbb{CP}^1$. The Abelian integral $\mathbf{A}(k)=(A_1(k), A_2(k), \ldots, A_{2n-1}(k))^{\top}$ is considered in the upper sheet of $\mathfrak{X}$, we observe that
\begin{equation}\label{Ajump}
\begin{array}{l}
\begin{aligned}
&\mathbf{A} _+(k)-\mathbf{A} _-(k)=0 \ (mod\  \Z^{2n-1}), &k\in \bigcup\limits_{j=1}^{n-1}(i(-a_{j+1},-b_{j})\cup i(b_{j},a_{j+1}))\cup i(-a_1,a_1),\\
&\mathbf{A} _+(k)+\mathbf{A} _-(k)=0\ (mod\  \Z^{2n-1}), &k\in \Sigma_{-n},\\
&\mathbf{A} _+(k)+\mathbf{A} _-(k)=B \mathbf{e_j}, &k\in \Sigma_{-(n-j)},\ j=1,2,\ldots,n-1,\\
&\mathbf{A} _+(k)+\mathbf{A} _-(k)=B \mathbf{e_{n-1+j}}, &k\in \Sigma_{j},\ j=1,2,\ldots,n.
\end{aligned}
\end{array}
\end{equation}
Next, we define the additional $2n-1$ Abel integrals as follows:
\begin{align}
\varsigma_j(k)=\int_{-ia_n}^{k}\varphi_j(\zeta)\d \zeta, \ \varphi_j(\zeta)=\frac{\sum\limits_{l=1}^{4n-1}s_{jl}\zeta^{4n-1-l}}{R(\zeta)},\ j=1,2,\ldots,2n-1,
\end{align}
where $s_{jl}$ are determined by
\begin{align}\label{varinfy}
&\varsigma_j(k)\to k^{j}+\mathcal{O}(1),\ \ \text{as}\  k \to \infty^+,\ j=1,2,\ldots,2n-1,\\
&\oint_{\mathcal{A}_{l}}\varphi_j(\zeta)\d  \zeta=0,\ j,l=1,2,\ldots,2n-1.
\end{align}
Then we set
\begin{align}\label{UVW}
U_{j,l}=\oint_{\mathcal{B}_{l}}\varphi_{j}(\zeta)\d  \zeta,\ j,l=1,2,\ldots,2n-1,
\end{align}
and
\begin{align}\label{J123}
J_{j}=\lim\limits_{k\to \infty}\varsigma_{j}(k)-k^{j},\ j=1,2,\ldots,2n-1.
\end{align}
We obtain that
\begin{align}\label{varsigmajump}
&\varsigma_{j+}(k)+\varsigma_{j-}(k)=U_{j,l},\ k\in \Sigma_{-(n-l)},\ l=1,2,\ldots,n-1,\ j=1,2,\ldots, 2n-1,\\
&\varsigma_{j+}(k)+\varsigma_{j-}(k)=U_{j,n-1+l},\ k\in \Sigma_{l},\ l=1,2,\ldots,n,\ j=1,2,\ldots, 2n-1.
\end{align}
Next, we introduce the Riemann-theta function
\begin{align}
\Theta(\mathbf{v})=\sum\limits_{\mathbf{w}\in \Z^{2n-1}}e^{\pi i(B\mathbf{w},\mathbf{w})+2\pi i (\mathbf{w},\mathbf{v})}
\end{align}
where $\mathbf{v}=(v_1,v_2,\ldots,v_{2n-1})^{\top}$, $\mathbf{w}=(w_1,w_2,\ldots,w_{2n-1})^{\top}$ and $(\mathbf{w},\mathbf{v})=w_1 v_1+w_2 v_2+\cdots+w_{2n-1} v_{2n-1}$. The Riemann-theta function satisfies the periodicity relations
\begin{align}\label{Theta}
\Theta(\mathbf{v}+\mathbf{e}_j)=\Theta(\mathbf{v}),\ \ \Theta(\mathbf{v}+B\mathbf{e}_j)=e^{-\pi iB_{jj}- 2\pi iv_j}\Theta(\mathbf{v}),\ \ j=1,2,\ldots, 2n-1.
\end{align}
Finally, we define a function $\gamma(k)$ by
\begin{align}
\gamma(k)=\le(\prod\limits_{j=1}^{n}\frac{(k-ib_{j})(k+i a_{j})}{(k-ia_{j})(k+i b_{j})} \ri)^{\frac{1}{4}},
\end{align}
which is analytic in $\C\backslash(\bigcup\limits_{j=1}^{n} (\Sigma_{j}\cup\Sigma_{-j}))$ and normalized such that $\gamma(k)\to 1$ as $k\to \infty^+$, $\gamma(k)$ satisfies the following jump condition
\begin{align}\label{gamma}
\gamma_+(k)=i\gamma_-(k),\ k\in \bigcup\limits_{j=1}^{n} (\Sigma_{j}\cup\Sigma_{-j}) .
\end{align}
Then the function $\gamma(k)-\frac{1}{\gamma(k)}$ has $2n$ zeros on $\mathfrak{X}$ counting multiplicity. We use $\infty^{+}$, $\mathcal{P}_1$, $\mathcal{P}_2$, $\ldots$, $\mathcal{P}_{2n-1}$ represents these $2n$  zeros. Let $D$ denote the divisor $D:=\mathcal{P}_1+\mathcal{P}_2+\cdots+\mathcal{P}_{2n-1}$, the constant vector $\mathbf{d}$ is defined by
\begin{align}\label{d}
\mathbf{d}=\mathbf{A}(D)+\mathbf{K},
\end{align}
where  $\mathbf{A}(D)=\sum\limits_{j=1}^{2n-1}\mathbf{A}(\mathcal{P}_j)$ and $\mathbf{K}$ is the Riemann-theta constant vector defined by
\begin{align}
K_{j}=\frac{1}{2}\sum_{l=1}^{2n-1}B_{lj}-\frac{j}{2},\ j=1,2,\ldots,2n-1.
\end{align}
Here we have utilized the strategy in \cite{DeiftItsZhou,KVSD} such that the zeros of $\Theta(\mathbf{A}(k)+\mathbf{d})$ and $\Theta(-\mathbf{A}(k)+\mathbf{d})$ can be eliminated by the zeros of $\gamma(k)+\frac{1}{\gamma(k)}$ and $\gamma(k)-\frac{1}{\gamma(k)}$.
\begin{proposition}
The matrix solution $Q(k)$ of RH problem \eqref{Q1}--\eqref{Q2} is given by following formula
\begin{equation}\label{3.62}
\begin{array}{l}
Q_{11}(x;k)=\frac{1}{2}(\gamma(k)+\frac{1}{\gamma(k)})\frac{\Theta(\mathbf{A}(\infty)+\mathbf{d})}{\Theta(\mathbf{A}(k)+\mathbf{d})} \frac{\Theta(\mathbf{A}(k)+\mathbf{d}-\frac{1}{2 \pi i}\sum\limits_{j=1}^{2n-1}\mathbf{U}_{j}h_{j})}{\Theta(\mathbf{A}(\infty)+\mathbf{d}-\frac{1}{2 \pi i}\sum\limits_{j=1}^{2n-1}\mathbf{U}_{j}h_{j})} \exp(-\sum\limits_{j=1}^{2n-1}(\varsigma_j-J_j)h_j),\\
Q_{12}(x;k)=\frac{1}{2}(\gamma(k)-\frac{1}{\gamma(k)})\frac{\Theta(\mathbf{A}(\infty)+\mathbf{d})}{\Theta(\mathbf{A}(k)-\mathbf{d})} \frac{\Theta(\mathbf{A}(k)-\mathbf{d}+\frac{1}{2 \pi i}\sum\limits_{j=1}^{2n-1}\mathbf{U}_{j}h_{j})}{\Theta(\mathbf{A}(\infty)+\mathbf{d}-\frac{1}{2 \pi i}\sum\limits_{j=1}^{2n-1}\mathbf{U}_{j}h_{j})} \exp(\sum\limits_{j=1}^{2n-1}(\varsigma_j+J_j)h_j),\\
Q_{21}(x;k)=\frac{1}{2}(\gamma(k)-\frac{1}{\gamma(k)})\frac{\Theta(\mathbf{A}(\infty)+\mathbf{d})}{\Theta(\mathbf{A}(k)-\mathbf{d})} \frac{\Theta(\mathbf{A}(k)-\mathbf{d}-\frac{1}{2 \pi i}\sum\limits_{j=1}^{2n-1}\mathbf{U}_{j}h_{j})}{\Theta(\mathbf{A}(\infty)+\mathbf{d}+\frac{1}{2 \pi i}\sum\limits_{j=1}^{2n-1}\mathbf{U}_{j}h_{j})} \exp(-\sum\limits_{j=1}^{2n-1}(\varsigma_j+J_j)h_j),\\
Q_{22}(x;k)=\frac{1}{2}(\gamma(k)+\frac{1}{\gamma(k)})\frac{\Theta(\mathbf{A}(\infty)+\mathbf{d})}{\Theta(\mathbf{A}(k)+\mathbf{d})} \frac{\Theta(\mathbf{A}(k)+\mathbf{d}+\frac{1}{2 \pi i}\sum\limits_{j=1}^{2n-1}\mathbf{U}_{j}h_{j})}{\Theta(\mathbf{A}(\infty)+\mathbf{d}+\frac{1}{2 \pi i}\sum\limits_{j=1}^{2n-1}\mathbf{U}_{j}h_{j})} \exp(\sum\limits_{j=1}^{2n-1}(\varsigma_j-J_j)h_j),
\end{array}
\end{equation}
where $\mathbf{U}_{j}=(U_{j,1},U_{j,2},\ldots,U_{j,2n-1})^{\top}\in \C^{2n-1}$.
\end{proposition}
\begin{proof}
Using the relation \eqref{Ajump}, \eqref{varsigmajump} and \eqref{Theta}, we obtain
\begin{equation*}
\begin{array}{l}
\frac{\Theta(\mathbf{A}(\infty)+\mathbf{d})}{\Theta(\mathbf{A}_{+}(k)+\mathbf{d})} \frac{\Theta(\mathbf{A}_{+}(k)+\mathbf{d}-\frac{1}{2 \pi i}\sum\limits_{j=1}^{2n-1}\mathbf{U}_{j}h_{j})}{\Theta(\mathbf{A}(\infty)+\mathbf{d}-\frac{1}{2 \pi i}\sum\limits_{j=1}^{2n-1}\mathbf{U}_{j}h_{j})} \exp(-\sum\limits_{j=1}^{2n-1}(\varsigma_{j+}-J_j)h_j)\\
=\frac{\Theta(\mathbf{A}(\infty)+\mathbf{d})}{\Theta(-\mathbf{A}_{-}(k)+B\mathbf{e}_l+\mathbf{d})} \frac{\Theta(-\mathbf{A}_{-}(k)+B\mathbf{e}_l+\mathbf{d}-\frac{1}{2 \pi i}\sum\limits_{j=1}^{2n-1}\mathbf{U}_{j}h_{j})}{\Theta(\mathbf{A}(\infty)+\mathbf{d}-\frac{1}{2 \pi i}\sum\limits_{j=1}^{2n-1}\mathbf{U}_{j}h_{j})} \exp(-\sum\limits_{j=1}^{2n-1}(-\varsigma_{j-}+U_{j,l}-J_j)h_j)\\
=\frac{\Theta(\mathbf{A}(\infty)+\mathbf{d})}{\Theta(-\mathbf{A}_{-}(k)+\mathbf{d})} \frac{\Theta(-\mathbf{A}_{-}(k)+\mathbf{d}-\frac{1}{2 \pi i}\sum\limits_{j=1}^{2n-1}\mathbf{U}_{j}h_{j})}{\Theta(\mathbf{A}(\infty)+\mathbf{d}-\frac{1}{2 \pi i}\sum\limits_{j=1}^{2n-1}\mathbf{U}_{j}h_{j})} \exp(-\sum\limits_{j=1}^{2n-1}(-\varsigma_{j-}-J_j)h_j)\\
=\frac{\Theta(\mathbf{A}(\infty)+\mathbf{d})}{\Theta(\mathbf{A}_{-}(k)-\mathbf{d})} \frac{\Theta(\mathbf{A}_{-}(k)-\mathbf{d}+\frac{1}{2 \pi i}\sum\limits_{j=1}^{2n-1}\mathbf{U}_{j}h_{j})}{\Theta(\mathbf{A}(\infty)+\mathbf{d}-\frac{1}{2 \pi i}\sum\limits_{j=1}^{2n-1}\mathbf{U}_{j}h_{j})} \exp(\sum\limits_{j=1}^{2n-1}(\varsigma_{j-}+J_j)h_j),
\end{array}
\end{equation*}
for $k\in \Sigma_{-(n-l)},\ l=1,2,\ldots,n-1$. Noting \eqref{gamma}, we have $Q_{11+}(k)=i Q_{12-}(k)$ for $k\in \Sigma_{-(n-l)}$. Similarly, for $k\in \bigcup\limits_{j=1}^{n}(\Sigma_{j}\cup \Sigma_{-j})$, we have $Q_{11+}(k)=i Q_{12-}(k)$, $Q_{12+}(k)=i Q_{11-}(k)$, $Q_{21+}(k)=i Q_{22-}(k)$ and $Q_{22+}(k)=iQ_{21-}(k)$.
Using the fact  $\gamma(k)\to 1$ as $k\to \infty$, \eqref{varinfy} and \eqref{J123}, we obtain
$$Q(k)\to \begin{pmatrix}e^{-\sum\limits_{j=1}^{2n-1}h_j k^j}&0\\0&e^{\sum\limits_{j=1}^{2n-1}h_j k^j}\end{pmatrix}, \ \ k\to \infty,$$ namely, the condition \eqref{Q2} is satisfied.
\end{proof}

Finally, the matrix solution $S^{\infty}(k)$ of RH problem \ref{RHP4} is given by
\begin{align}\label{3.63}
S^{\infty}(k)=e^{-h_0\sigma_3} Q(k) e^{h(k)\sigma_3}.
\end{align}
Next, we show the construction of a local parametrix at $k=ib_n$ with the help of the modified Bessel functions. Introduce a local variable $\mu=\mu(x;k)$ in a disk $B_{\rho}^{(ib_n)}$ as follows:
\begin{gather}
\sqrt{\mu}=\frac12[ix(g(k)-k)],\ k\in B_{\rho}^{(ib_n)}=\le\{ k \in \mathbb{C} \le| \,  \le|k - ib_{n}\ri|< \rho \ri.  \ri\},
\end{gather}
with branch cut for $\sqrt{\mu}$ is $\mu\in (-\infty,0]$, which corresponds to $k\in i(a_n, b_n)$. Then we introduce the model parametrix $\Psi_{Bes}(\mu)$ as in \cite{Girotti-1}, which satisfies the following jump relation
\begin{align}
\Psi_{Bes+}(\mu)=\Psi_{Bes-}(\mu)\begin{cases}
\begin{pmatrix}
1&0\\1&1
\end{pmatrix},\ \  \arg \mu=\pm \frac{2\pi}{3} ,\\
\begin{pmatrix}
0&1\\-1&0
\end{pmatrix},\ \ \arg \mu=\pi,
\end{cases}
\end{align}
with asymptotics at infinity
\begin{align}
\Psi_{Bes}(\mu)=(2\pi \mu^{\frac12})^{-\frac12 \sigma_3}\frac{1}{\sqrt{2}}\begin{pmatrix}1&i\\i&1 \end{pmatrix}\le(  I +\mathcal{O}(\frac{1}{\mu^{\frac12}})\ri)e^{2\mu^{\frac12}\sigma_3}.
\end{align}
Then we consider
\begin{equation}\label{Eq12}
\begin{array}{l}
P^{(1)}(k) =  S(k) \le(\frac{ e^{  -i \pi  /4}}{\sqrt{2\hat{r}_n(k)} f(k)}
\right)^{\sigma_{3}},\quad k \in B^{(ib_n)}_{\rho},
\end{array}
\end{equation}
and
\begin{equation}\label{Eq13}
\begin{array}{l}
P^{(2)}(\mu) = P^{(1)}(k(\mu)) e^{-2 \mu^{\frac{1}{2}} \sigma_3} \begin{pmatrix}0&1\\1&0 \end{pmatrix}\ , \quad \mu \in \C.
\end{array}
\end{equation}
It is easy to check that $P^{(2)}(\mu)$ has the same jump as $\Psi_{Bes}(\mu)$, hence the local parametrix around the point $ib_n$ is
\begin{gather}
P^{ib_n}(k)=D(k)\Psi_\Bes(\mu(k)) \begin{pmatrix} 0&1\\1&0\end{pmatrix} e^{2\mu(k)^{\frac{1}{2}}\sigma_3} \left(
\frac{ e^{ -i \pi  /4}}{\sqrt{2 \hat{r}_n(k)} f(k)}
\right)^{-\sigma_{3}},\ k \in B^{(ib_n)}_{\rho},
\end{gather}
where
\begin{gather}
D(k)
=S^{\infty}(k)
\left(
\frac{ e^{ -i \pi  /4}}{\sqrt{2\hat{r}_n(k)} f(k)}
\right)^{\sigma_{3}}
\frac{1}{\sqrt{2}} \begin{pmatrix}-i& 1 \\ 1  &  -i\end{pmatrix} \le(2\pi \mu^{\frac{1}{2}}\ri)^{\frac{1}{2}\sigma_3}, \quad  k \in B^{(ib_n)}_{\rho}.
\end{gather}
Therefore, we have the following matching condition on the boundary $\partial B^{(ib_n)}_{\rho} $:
\begin{gather}
P^{ib_n}(k)(S^{\infty}(k))^{-1}= I +\mathcal{O}(\frac{1}{x}),\  \text{as}\ x\to +\infty.
\end{gather}
The local parametrix of other endpoints can be constructed similarly.
 We define the following error function $ {\mathcal{E}}(k)$:
 \begin{gather}
	 {\mathcal{E}}(k) =  {S}(k) \left( {P}(k) \right)^{-1},
\text{where}\\
	 {P}(k) = \begin{cases}
		 {S}^{\infty}(k),  & k\in\mathbb{C}\backslash (\bigcup\limits_{j=1}^{n}B^{(\pm i a_j)}_{\rho}\cup \bigcup\limits_{j=1}^{n}B^{(\pm i b_j)}_{\rho}), \\
         {P}^{ia_j}(k), & k\in B^{(ia_j)}_{\rho},\ j=1,2,\ldots,n, \\
         {P}^{ib_j}(k), & k\in B^{(ib_j)}_{\rho},\ j=1,2,\ldots,n, \\
         {P}^{-ia_j}(k), & k\in B^{(-ia_j)}_{\rho},\ j=1,2,\ldots,n, \\
         {P}^{-ib_j}(k), & k\in B^{(-ib_j)}_{\rho},\ j=1,2,\ldots,n.
	\end{cases}
\end{gather}

For some $c>0$, the matrix $ {\mathcal{E}}(x;k)$ satisfies
 \begin{gather}
 \begin{array}{l}
  {\mathcal{E}}_+(k) = \begin{cases}
  {\mathcal{E}}_-(k) (  I  + \mathcal{O}( \displaystyle { e^{-cx} }) ),\  k\in \bigcup\limits_{j=1}^{n}({\mathcal{C}}_j\cup {\mathcal{C}}_{-j})\backslash (\bigcup\limits_{j=1}^{n}(B^{(\pm i a_j)}_{\rho}\cup B^{(\pm i b_j)}_{\rho})), \\
  {\mathcal{E}}_-(k) (  I  + \mathcal{O}(\displaystyle x^{-1} ) ),\ k\in \bigcup\limits_{j=1}^{n}(\partial B^{(\pm i a_j)}_{\rho}\cup \partial B^{(\pm i b_j)}_{\rho}),
 \end{cases}
 \end{array}
 \end{gather}
 and
 \begin{gather}  {\mathcal{E}}(k) =  I  + \mathcal{O}(\frac{1}{k}), \qquad \text{as }k \to \infty \ .
 \end{gather}
Therefore, by a small norm argument, we conclude that
\begin{gather}
\mathcal{E}(x;k)= I +\mathcal{O}(\frac{1}{x}),
\end{gather}
uniformly in $k$ as $x\to +\infty$. Taking into account all the transformations we performed, we find that
\begin{gather}\label{3.77}
u(x,0)=\lim\limits_{k\to\infty}2i k X_{12}(x;k)=\lim\limits_{k\to\infty}2i k S^{\infty}_{12}(x;k)+\mathcal{O}(\frac{1}{x}).
\end{gather}
Substituting \eqref{3.62} and \eqref{3.63} into \eqref{3.77}, we obtain the following theorem.
\begin{thm}
In the regime $x \to +\infty$, $t=0$, the mKdV soliton gas solution $u(x,0)$ has the following asymptotic behaviour
\begin{gather}
\begin{array}{l}
u(x,0)=\alpha_1\frac{\Theta(\mathbf{A}(\infty)-\mathbf{d}+\frac{1}{2\pi i}\sum\limits_{j=1}^{2n-1}\mathbf{U}_{j}h_{j})}{\Theta(\mathbf{A}(\infty)+\mathbf{d}-\frac{1}{2\pi i}\sum\limits_{j=1}^{2n-1}\mathbf{U}_{j}h_{j})} \frac{\Theta(\mathbf{A}(\infty)+\mathbf{d})}{\Theta(\mathbf{A}(\infty)-\mathbf{d})} \exp(2\sum\limits_{j=1}^{2n-1}J_jh_j-2h_0)+\mathcal{O}(\frac{1}{x}),
\end{array}
\end{gather}
where $\mathbf{A}$, $\mathbf{d}$, $\mathbf{U}_{j}$, $h_j$ and $J_j$  are defined by \eqref{A}, \eqref{d}, \eqref{UVW}, \eqref{h0123}, \eqref{J123}, respectively, and $\alpha_1=\sum\limits_{j=1}^{n}(b_{j}-a_{j}).$
\end{thm}
\section{Behavior of the potential $u(x,t)$ as $t\rightarrow +\infty$}
 Recall that the RH  problem to $X(x,t;k)$ for the soliton gas
 \begin{align}\label{Xtjump}
		&X_{+}(k) = X_{-}(k) \begin{cases}
			\begin{pmatrix} {1}& 0\\   {2i   r_j(k) e^{ -2 i\theta(x,t;k)} }  & {1} \end{pmatrix}, &  k \in i(a_j,b_j),\ j=1,2,\ldots, n,\\
			\begin{pmatrix} {1}& {2i   r_j(k) e^{ 2 i\theta(x,t;k)}}\\0  & {1} \end{pmatrix}, &  k \in i(-b_j,-a_j),\ j=1,2,\ldots, n.
		\end{cases}
\end{align}		
In the region  $\xi=\frac{x}{t}<4 a_1^2$, it becomes evident that the phases in the jumps are exponentially decaying as $t \to + \infty$. Through a small norm argument, we can infer that	
\begin{gather}
X(k) = I  + \mathcal{O}\le( {e^{-2ta_1(4a_1^2-\xi)}} \ri), \qquad \text{as } t\to + \infty \text{ with }  \xi<4a_1^2,
\end{gather}
and the potential $u(x,t)$ becomes trivial.

As $4a_1^2<\xi<\xi_2$, where $\xi_2$ is defined by \eqref{xi2}, there exist a rarefaction wave region and a periodic travelling wave region. We introduce a speed $\xi_1$ to distinguish these two regions
\begin{align}\label{xi1}
\xi_1=b_1^2W_1(\frac{a_1^2}{b_1^2}),\quad W_1(m_{\beta_1})=\frac{4(1-m_{\beta_1})K(m_{\beta_1})}{E(m_{\beta_1})}+2(1+m_{\beta_1}),
\end{align}
where $K(m_{\beta_1})=\int_0^{1}\frac{\d s}{\sqrt{(1-s^2)(1-m_{\beta_1} s^2)}}$ and $E(m_{\beta_1})=\int_0^{1}\frac{\sqrt{1-m_{\beta_1} s^2}}{\sqrt{1-s^2}}\d s$ are the complete elliptic integrals of first and second kind, respectively. The $g$-function introduced in the literature \cite{Girotti-2} can be directly applied to region $(4a_1^2,\xi_2)$, where the jumps on $\bigcup\limits_{j=2}^{n}(\Sigma_{j}\cup\Sigma_{-j})$ can be eliminated, with a proof similar to Lemma 4.9 in \cite{Girotti-2}. Then for $4a_1^2<\xi<\xi_2$, the asymptotic behavior of the mKdV soliton gas solution in the large time limit is
\begin{equation}
				\label{u_dnt}
				u(x,t) = (\beta_1 + a_1)
				\Jac \left( (\beta_1 + a_1) ( x - 2(a_1^2+\beta_1^2) t - x_0), m_{1\beta_1} \right)+ \mathcal{O}\le(t^{-1}\ri),
			\end{equation}
			where $\Jac\le(z,m_{1\beta_1} \ri)$ is the Jacobi elliptic function with modulus $m_{1\beta_1}=\frac{4\beta_1a_1}{(\beta_1+a_1)^2}$,
 $x_0 = \frac{K(m_{\beta_1})(\Delta_{\beta_1}-\pi )}{\beta_1 \pi }$ and $\Delta_{\beta_1}=-i (\int_{0}^{ia_1}\frac{\d s}{\sqrt{(s^2+\beta_1^2)(s^2+a_1^2)}})^{-1}(\int_{\Sigma_{1,\beta_1}}\frac{\log (2r_1(s))}{\sqrt{(s^2+\beta_1^2)(s^2+a_1^2)}_+}\d s)$.
 The parameter $\beta_1=\beta_1(\xi) \in (a_1, b_1)$ for $\xi\in (4a_1^2, \xi_1)$ and satisfies the Whitham modulation equation $\xi=W_1(\frac{a_1^2}{\beta_1^2})$, while for $\xi \in (\xi_1, \xi_2)$, $\beta_1(\xi)=b_1$ (see \cite{Girotti-2} for details). When $\xi>\xi_2$, there are $2n-2$ distinct asymptotic regions associated with high genus Riemann-theta functions, which we will investigate in the following two sections.

\section{The first genus $2m-1$ sector : $\xi_{2m-2} <\xi< \xi_{2m-1}$}

Assume that $\xi_{2m-2}$ and $\xi_{2m-1}$, $(m=2,3,\ldots,n)$, are defined by \eqref{xi2} and \eqref{xi3}, respectively. For $\xi_{2m-2} <\xi< \xi_{2m-1}$, we split the contours $\Sigma_m$ and $\Sigma_{-m}$ in the following way: let ${\beta_{m}}\in (a_m, b_m)$ and define the sub intervals $\Sigma_{m,{\beta_{m}}}=i(a_m,{\beta_{m}})$ and $\Sigma_{-m,{\beta_{m}}}=i(-{\beta_{m}},-a_m)$, where ${\beta_{m}}$ is determined by the Whitham evolution equation \eqref{alpha}. Then the off-diagonal terms of jump conditions \eqref{Xtjump} on two growing bands $i(a_m, {\beta_{m}})$ and $i(-{\beta_{m}}, -a_m)$ are exponentially increasing as $t\to\infty$. In this setting, we need to introduce two new scalar functions $g_{{\beta_{m}}}(x,t;k)$ and $f_{{\beta_{m}}}(k)$.

The function $g_{{\beta_{m}}}(x,t;k)$ needs to satisfy the following properties:
\begin{enumerate}
\item{}
$g_{{\beta_{m}}}(x,t;k)$ is analytic in $\C\backslash i[-{\beta_{m}},{\beta_{m}}]$.

\item{}
$g_{{\beta_{m}}}(x,t;k)$ satisfies the following jump conditions:
\begin{align}
&g_{{\beta_{m}}+}(k)+ g_{{\beta_{m}}-}(k)= 2kx+8k^3t, & k\in \Sigma_{m,\beta_{m}}\cup\Sigma_{-m,\beta_{m}} \cup\bigcup\limits_{j=1}^{m-1}(\Sigma_j\cup\Sigma_{-j}) ,\label{g31}\\
&g_{{\beta_{m}}+}(k)-g_{{\beta_{m}}-}(k)=\Omega_{{\beta_{m}},0} ,& k \in  i[-a_1,a_1], \label{g32}\\
&g_{{\beta_{m}}+}(k)-g_{{\beta_{m}}-}(k)=\Omega_{{\beta_{m}},j} ,& k \in i[b_{j},a_{j+1}],\ j=1,2,\ldots,m-1, \label{g33}\\
&g_{{\beta_{m}}+}(k)-g_{{\beta_{m}}-}(k)=\Omega_{{\beta_{m}},{-j}} ,& k \in  i[-a_{j+1},-b_{j}],\ j=1,2,\ldots,m-1. \label{g34}
\end{align}
\item{}
$g_{{\beta_{m}}}(k) = \mathcal{O}\le(\frac{1}{k}\ri), k \rightarrow  \infty \ .$

\item{}
$g_{{\beta_{m}}}(x,t;k)$ has the following asymptotic behavior near the endpoints $\pm ia_j$, $\pm ib_j$ and $\pm i{\beta_{m}}$
\begin{align}
&[g_{{\beta_{m}}}(k)-kx-4k^3t]^{'}=\mathcal{O}\le((k\mp ia_j)^{-\frac{1}{2}}\ri), k \rightarrow \pm ia_{j}, \ j=1,2,\ldots,m,\\
&[g_{{\beta_{m}}}(k)-kx-4k^3t]^{'}=\mathcal{O}\le((k\mp ib_j)^{-\frac{1}{2}}\ri), k \rightarrow \pm ib_{j},\ j=1,2,\ldots,m-1,\\
&g_{{\beta_{m}}}(k)-kx-4k^3t=\mathcal{O}\le((k\mp i{\beta_{m}})^\frac{3}{2}\ri), k \rightarrow \pm i{\beta_{m}}.\label{axi}
\end{align}
\item{}
$g_{{\beta_{m}}}(x,t;k)$ satisfies the following inequalities
\begin{align}
&\Im (g_{{\beta_{m}}}(k)-xk-4k^3t)>c t, k\in i({\beta_{m}},b_m]\cup\bigcup\limits_{j=m+1}^{n}\Sigma_{j},\\
&\Im (g_{{\beta_{m}}}(k)-xk-4k^3t)<-c t, k\in i[-b_m,-{\beta_{m}})\cup\bigcup\limits_{j=m+1}^{n}\Sigma_{-j}.
\end{align}
\end{enumerate}
To solve the scalar RH problem, we observe that the derivative of $g_{{\beta_{m}}}(k)$ with respect to $k$ satisfies the following properties:
\begin{enumerate}
\item{}
$g_{{\beta_{m}}}'(x,t;k)$ is analytic in $\C\backslash i[-{\beta_{m}},{\beta_{m}}]$.

\item{}
$g_{{\beta_{m}}}'(x,t;k)$ satisfies the following jump conditions:
\begin{align*}
&g_{{\beta_{m}}+}'(k)+ g_{{\beta_{m}}-}'(k)= 2x+24k^2t, &k\in\Sigma_{m,{\beta_{m}}}\cup\Sigma_{-m,{\beta_{m}}}\cup\bigcup\limits_{j=1}^{m-1}(\Sigma_j\cup\Sigma_{-j}) ,\\
&g_{{\beta_{m}}+}'(k)-g_{{\beta_{m}}-}'(k)=0 ,& k \in  i[-a_1,a_1]\cup \bigcup\limits_{j=1}^{m-1}(i[b_j,a_{j+1}] \cup i[-a_{j+1}, -b_j]).
\end{align*}
\item{}
$g_{{\beta_{m}}}'(k) = \mathcal{O}\le(\frac{1}{k^2}\ri), k \rightarrow  \infty \ .$
\end{enumerate}
From above, we can write the explicit expression of the derivative of $g_{{\beta_{m}}}(k)$ as follows:
\begin{equation}\label{gad}
\begin{array}{l}
g_{{\beta_{m}}}'(k)=x+12k^2t-x \frac{k^{2m}+\sum\limits_{j=1}^{m} \tilde{c}_j k^{2(m-j)}}{R_{{\beta_{m}}}(k)}-12t \frac{k^{2(m+1)}+\frac{1}{2}(a_m^2+\beta_{m}^2+\sum\limits_{j=1}^{m-1}(a_j^2+b_j^2)) k^{2m}+\sum\limits_{j=1}^{m} \hat{c}_j k^{2(m-j)}}{R_{{\beta_{m}}}(k)},
\end{array}
\end{equation}
where $$R_{{\beta_{m}}}(k)=\sqrt{(k^2+a_m^2)(k^2+\beta_{m}^2)\prod\limits_{j=1}^{m-1}(k^2+a_j^2)(k^2+b_j^2)},$$ and the conditions \eqref{g31}--\eqref{g34} imply that  the constants $ \tilde{c}_j $ and $ \hat{c}_j $ have to satisfy the following equations
\begin{align}
&\int_{ib_l}^{ia_{l+1}}\frac{\zeta^{2m}+\sum\limits_{j=1}^{m} \tilde{c}_j \zeta^{2(m-j)}}{R_{{\beta_{m}}}(\zeta)}\d\zeta=0,\ l=1,2,\ldots,m-1,\label{tildec1}\\ &\int_{-ia_1}^{ia_1}\frac{\zeta^{2m}+\sum\limits_{j=1}^{m} \tilde{c}_j \zeta^{2(m-j)}}{R_{{\beta_{m}}}(\zeta)}\d\zeta=0,\label{tildec2}\\
&\int_{ib_l}^{ia_{l+1}}\frac{\zeta^{2(m+1)}+\frac{1}{2}(a_m^2+\beta_{m}^2+\sum\limits_{j=1}^{m-1}(a_j^2+b_j^2)) \zeta^{2m}+\sum\limits_{j=1}^{m} \hat{c}_j \zeta^{2(m-j)}}{R_{{\beta_{m}}}(\zeta)}\d\zeta=0,\ l=1,2,\ldots,m-1,\label{hatc1}\\ &\int_{-ia_1}^{ia_1}\frac{\zeta^{2(m+1)}+\frac{1}{2}(a_m^2+\beta_{m}^2+\sum\limits_{j=1}^{m-1}(a_j^2+b_j^2)) \zeta^{2m}+\sum\limits_{j=1}^{m} \hat{c}_j \zeta^{2(m-j)}}{R_{{\beta_{m}}}(\zeta)}\d\zeta=0\label{hatc2}.
\end{align}
In order to solve the system of equations  \eqref{tildec1}--\eqref{tildec2} regarding $ \tilde{c}_j $, we define the following matrices
\begin{align}
\tilde{A}=\begin{pmatrix}
\int_{ib_{m-1}}^{ia_{m}}\frac{\zeta^{2(m-1)}}{R_{\beta_{m}}(\zeta)}\d\zeta & \int_{ib_{m-1}}^{ia_{m}}\frac{\zeta^{2(m-2)}}{R_{\beta_{m}}(\zeta)}\d\zeta &
\cdots &
\int_{ib_{m-1}}^{ia_{m}}\frac{1}{R_{\beta_{m}}(\zeta)}\d\zeta\\
\int_{ib_{m-2}}^{ia_{m-1}}\frac{\zeta^{2(m-1)}}{R_{\beta_{m}}(\zeta)}\d\zeta & \int_{ib_{m-2}}^{ia_{m-1}}\frac{\zeta^{2(m-2)}}{R_{\beta_{m}}(\zeta)}\d\zeta &
\cdots &
\int_{ib_{m-2}}^{ia_{m-1}}\frac{1}{R_{\beta_{m}}(\zeta)}\d\zeta\\
\cdots&\cdots&\cdots&\cdots\\
\int_{-ia_1}^{ia_1}\frac{\zeta^{2(m-1)}}{R_{\beta_{m}}(\zeta)}\d\zeta & \int_{-ia_1}^{ia_1}\frac{\zeta^{2(m-2)}}{R_{\beta_{m}}(\zeta)}\d\zeta &
\cdots &
\int_{-ia_1}^{ia_1}\frac{1}{R_{\beta_{m}}(\zeta)}\d\zeta
\end{pmatrix},
\end{align}
\begin{align}
\tilde{D}_1=\begin{pmatrix}
\int_{ib_{m-1}}^{ia_{m}}\frac{-\zeta^{2m}}{R_{\beta_{m}}(\zeta)}\d\zeta & \int_{ib_{m-1}}^{ia_{m}}\frac{\zeta^{2(m-2)}}{R_{\beta_{m}}(\zeta)}\d\zeta &
\cdots &
\int_{ib_{m-1}}^{ia_{m}}\frac{1}{R_{\beta_{m}}(\zeta)}\d\zeta\\
\int_{ib_{m-2}}^{ia_{m-1}}\frac{-\zeta^{2m}}{R_{\beta_{m}}(\zeta)}\d\zeta & \int_{ib_{m-2}}^{ia_{m-1}}\frac{\zeta^{2(m-2)}}{R_{\beta_{m}}(\zeta)}\d\zeta &
\cdots &
\int_{ib_{m-2}}^{ia_{m-1}}\frac{1}{R_{\beta_{m}}(\zeta)}\d\zeta\\
\cdots&\cdots&\cdots&\cdots\\
\int_{-ia_1}^{ia_1}\frac{-\zeta^{2m}}{R_{\beta_{m}}(\zeta)}\d\zeta & \int_{-ia_1}^{ia_1}\frac{\zeta^{2(m-2)}}{R_{\beta_{m}}(\zeta)}\d\zeta &
\cdots &
\int_{-ia_1}^{ia_1}\frac{1}{R_{\beta_{m}}(\zeta)}\d\zeta
\end{pmatrix},
\end{align}
\begin{gather*}
\tilde{D}_j=\begin{array}{l}\begin{pmatrix}
\scriptstyle \int_{ib_{m-1}}^{ia_{m}}\frac{\zeta^{2(m-1)}}{R_{\beta_{m}}(\zeta)}\d\zeta &
\scriptstyle \cdots &
\scriptstyle \int_{ib_{m-1}}^{ia_{m}}\frac{\zeta^{2(m-j+1)}}{R_{\beta_{m}}(\zeta)}\d\zeta &
\scriptstyle \int_{ib_{m-1}}^{ia_{m}}\frac{-\zeta^{2m}}{R_{\beta_{m}}(\zeta)}\d\zeta &
\scriptstyle \int_{ib_{m-1}}^{ia_{m}}\frac{\zeta^{2(m-j-1)}}{R_{\beta_{m}}(\zeta)}\d\zeta &
\scriptstyle \cdots &
\scriptstyle \int_{ib_{m-1}}^{ia_{m}}\frac{1}{R_{\beta_{m}}(\zeta)}\d\zeta\\
\scriptstyle \int_{ib_{m-2}}^{ia_{m-1}}\frac{\zeta^{2(m-1)}}{R_{\beta_{m}}(\zeta)}\d\zeta &
\scriptstyle \cdots &
\scriptstyle \int_{ib_{m-2}}^{ia_{m-1}}\frac{\zeta^{2(m-j+1)}}{R_{\beta_{m}}(\zeta)}\d\zeta &
\scriptstyle \int_{ib_{m-2}}^{ia_{m-1}}\frac{-\zeta^{2m}}{R_{\beta_{m}}(\zeta)}\d\zeta &
\scriptstyle \int_{ib_{m-2}}^{ia_{m-1}}\frac{\zeta^{2(m-j-1)}}{R_{\beta_{m}}(\zeta)}\d\zeta &
\scriptstyle \cdots &
\scriptstyle \int_{ib_{m-2}}^{ia_{m-1}}\frac{1}{R_{\beta_{m}}(\zeta)}\d\zeta\\
\scriptstyle \cdots& \scriptstyle \cdots& \scriptstyle\cdots& \scriptstyle\cdots& \scriptstyle\cdots&\scriptstyle\cdots&\scriptstyle\cdots\\
\scriptstyle \int_{-ia_1}^{ia_1}\frac{\zeta^{2(m-1)}}{R_{\beta_{m}}(\zeta)}\d\zeta &
\scriptstyle \cdots &
\scriptstyle \int_{-ia_1}^{ia_1}\frac{\zeta^{2(m-j+1)}}{R_{\beta_{m}}(\zeta)}\d\zeta &
\scriptstyle \int_{-ia_1}^{ia_1}\frac{-\zeta^{2m}}{R_{\beta_{m}}(\zeta)}\d\zeta &
\scriptstyle \int_{-ia_1}^{ia_1}\frac{\zeta^{2(m-j-1)}}{R_{\beta_{m}}(\zeta)}\d\zeta &
\scriptstyle \cdots &
\scriptstyle \int_{-ia_1}^{ia_1}\frac{1}{R_{\beta_{m}}(\zeta)}\d\zeta
\end{pmatrix},
\end{array}
\end{gather*}
for $j=2,3,\ldots,m-1$,
\begin{align*}
\tilde{D}_m=\begin{pmatrix}
\int_{ib_{m-1}}^{ia_{m}}\frac{\zeta^{2(m-1)}}{R_{\beta_{m}}(\zeta)}\d\zeta & \int_{ib_{m-1}}^{ia_{m}}\frac{\zeta^{2(m-2)}}{R_{\beta_{m}}(\zeta)}\d\zeta &
\cdots &
\int_{ib_{m-1}}^{ia_{m}}\frac{\zeta^{2}}{R_{\beta_{m}}(\zeta)}\d\zeta&
\int_{ib_{m-1}}^{ia_{m}}\frac{-\zeta^{2m}}{R_{\beta_{m}}(\zeta)}\d\zeta\\
\int_{ib_{m-2}}^{ia_{m-1}}\frac{\zeta^{2(m-1)}}{R_{\beta_{m}}(\zeta)}\d\zeta & \int_{ib_{m-2}}^{ia_{m-1}}\frac{\zeta^{2(m-2)}}{R_{\beta_{m}}(\zeta)}\d\zeta &
\cdots &
\int_{ib_{m-2}}^{ia_{m-1}}\frac{\zeta^{2}}{R_{\beta_{m}}(\zeta)}\d\zeta&
\int_{ib_{m-2}}^{ia_{m-1}}\frac{-\zeta^{2m}}{R_{\beta_{m}}(\zeta)}\d\zeta\\
\cdots&\cdots&\cdots&\cdots&\ldots\\
\int_{-ia_1}^{ia_1}\frac{\zeta^{2(m-1)}}{R_{\beta_{m}}(\zeta)}\d\zeta & \int_{-ia_1}^{ia_1}\frac{\zeta^{2(m-2)}}{R_{\beta_{m}}(\zeta)}\d\zeta &
\cdots &
\int_{-ia_1}^{ia_1}\frac{\zeta^{2}}{R_{\beta_{m}}(\zeta)}\d\zeta&
\int_{-ia_1}^{ia_1}\frac{-\zeta^{2m}}{R_{\beta_{m}}(\zeta)}\d\zeta
\end{pmatrix}.
\end{align*}
Then the constants $ \tilde{c}_j $ can be represented as
\begin{gather}
 \tilde{c}_j =\frac{\det \tilde{D}_j}{\det \tilde{A}},\ j=1,2,\ldots,m.
\end{gather}
In order to solve the system \eqref{hatc1}--\eqref{hatc2}, we define
\begin{align}
\hat{D}_1=\begin{pmatrix}
\int_{ib_{m-1}}^{ia_{m}}\frac{-\zeta^{2(m+1)}}{R_{\beta_{m}}(\zeta)}\d\zeta & \int_{ib_{m-1}}^{ia_{m}}\frac{\zeta^{2(m-2)}}{R_{\beta_{m}}(\zeta)}\d\zeta &
\cdots &
\int_{ib_{m-1}}^{ia_{m}}\frac{1}{R_{\beta_{m}}(\zeta)}\d\zeta\\
\int_{ib_{m-2}}^{ia_{m-1}}\frac{-\zeta^{2(m+1)}}{R_{\beta_{m}}(\zeta)}\d\zeta & \int_{ib_{m-2}}^{ia_{m-1}}\frac{\zeta^{2(m-2)}}{R_{\beta_{m}}(\zeta)}\d\zeta &
\cdots &
\int_{ib_{m-2}}^{ia_{m-1}}\frac{1}{R_{\beta_{m}}(\zeta)}\d\zeta\\
\cdots&\cdots&\cdots&\cdots\\
\int_{-ia_1}^{ia_1}\frac{-\zeta^{2(m+1)}}{R_{\beta_{m}}(\zeta)}\d\zeta & \int_{-ia_1}^{ia_1}\frac{\zeta^{2(m-2)}}{R_{\beta_{m}}(\zeta)}\d\zeta &
\cdots &
\int_{-ia_1}^{ia_1}\frac{1}{R_{\beta_{m}}(\zeta)}\d\zeta
\end{pmatrix},
\end{align}
\begin{equation*}
\hat{D}_j=\begin{array}{l}
\begin{pmatrix}
\scriptstyle\int_{ib_{m-1}}^{ia_{m}}\frac{\zeta^{2(m-1)}}{R_{\beta_{m}}(\zeta)}\d\zeta &
\scriptstyle\cdots &
\scriptstyle\int_{ib_{m-1}}^{ia_{m}}\frac{\zeta^{2(m-j+1)}}{R_{\beta_{m}}(\zeta)}\d\zeta &
\scriptstyle\int_{ib_{m-1}}^{ia_{m}}\frac{-\zeta^{2(m+1)}}{R_{\beta_{m}}(\zeta)}\d\zeta &
\scriptstyle\int_{ib_{m-1}}^{ia_{m}}\frac{\zeta^{2(m-j-1)}}{R_{\beta_{m}}(\zeta)}\d\zeta &
\scriptstyle\cdots &
\scriptstyle\int_{ib_{m-1}}^{ia_{m}}\frac{1}{R_{\beta_{m}}(\zeta)}\d\zeta\\
\scriptstyle\int_{ib_{m-2}}^{ia_{m-1}}\frac{\zeta^{2(m-1)}}{R_{\beta_{m}}(\zeta)}\d\zeta &
\scriptstyle\cdots &
\scriptstyle\int_{ib_{m-2}}^{ia_{m-1}}\frac{\zeta^{2(m-j+1)}}{R_{\beta_{m}}(\zeta)}\d\zeta &
\scriptstyle\int_{ib_{m-2}}^{ia_{m-1}}\frac{-\zeta^{2(m+1)}}{R_{\beta_{m}}(\zeta)}\d\zeta &
\scriptstyle\int_{ib_{m-2}}^{ia_{m-1}}\frac{\zeta^{2(m-j-1)}}{R_{\beta_{m}}(\zeta)}\d\zeta &
\scriptstyle\cdots &
\scriptstyle\int_{ib_{m-2}}^{ia_{m-1}}\frac{1}{R_{\beta_{m}}(\zeta)}\d\zeta\\
\scriptstyle\cdots&\scriptstyle\cdots&\scriptstyle\cdots&\scriptstyle\cdots&\scriptstyle\cdots&\scriptstyle\cdots&\scriptstyle\cdots\\
\scriptstyle\int_{-ia_1}^{ia_1}\frac{\zeta^{2(m-1)}}{R_{\beta_{m}}(\zeta)}\d\zeta &
\scriptstyle\cdots &
\scriptstyle\int_{-ia_1}^{ia_1}\frac{\zeta^{2(m-j+1)}}{R_{\beta_{m}}(\zeta)}\d\zeta &
\scriptstyle\int_{-ia_1}^{ia_1}\frac{-\zeta^{2(m+1)}}{R_{\beta_{m}}(\zeta)}\d\zeta &
\scriptstyle\int_{-ia_1}^{ia_1}\frac{\zeta^{2(m-j-1)}}{R_{\beta_{m}}(\zeta)}\d\zeta &
\scriptstyle\cdots &
\scriptstyle\int_{-ia_1}^{ia_1}\frac{1}{R_{\beta_{m}}(\zeta)}\d\zeta
\end{pmatrix},
\end{array}
\end{equation*}
for $j=2,3,\ldots,m-1$,
\begin{align*}
\hat{D}_m=\begin{pmatrix}
\int_{ib_{m-1}}^{ia_{m}}\frac{\zeta^{2(m-1)}}{R_{\beta_{m}}(\zeta)}\d\zeta & \int_{ib_{m-1}}^{ia_{m}}\frac{\zeta^{2(m-2)}}{R_{\beta_{m}}(\zeta)}\d\zeta &
\cdots &
\int_{ib_{m-1}}^{ia_{m}}\frac{\zeta^{2}}{R_{\beta_{m}}(\zeta)}\d\zeta&
\int_{ib_{m-1}}^{ia_{m}}\frac{-\zeta^{2(m+1)}}{R_{\beta_{m}}(\zeta)}\d\zeta\\
\int_{ib_{m-2}}^{ia_{m-1}}\frac{\zeta^{2(m-1)}}{R_{\beta_{m}}(\zeta)}\d\zeta & \int_{ib_{m-2}}^{ia_{m-1}}\frac{\zeta^{2(m-2)}}{R_{\beta_{m}}(\zeta)}\d\zeta &
\cdots &
\int_{ib_{m-2}}^{ia_{m-1}}\frac{\zeta^{2}}{R_{\beta_{m}}(\zeta)}\d\zeta&
\int_{ib_{m-2}}^{ia_{m-1}}\frac{-\zeta^{2(m+1)}}{R_{\beta_{m}}(\zeta)}\d\zeta\\
\cdots&\cdots&\cdots&\cdots&\ldots\\
\int_{-ia_1}^{ia_1}\frac{\zeta^{2(m-1)}}{R_{\beta_{m}}(\zeta)}\d\zeta & \int_{-ia_1}^{ia_1}\frac{\zeta^{2(m-2)}}{R_{\beta_{m}}(\zeta)}\d\zeta &
\cdots &
\int_{-ia_1}^{ia_1}\frac{\zeta^{2}}{R_{\beta_{m}}(\zeta)}\d\zeta&
\int_{-ia_1}^{ia_1}\frac{-\zeta^{2(m+1)}}{R_{\beta_{m}}(\zeta)}\d\zeta
\end{pmatrix}.
\end{align*}
Then the constants $\hat{c}_{j}$ can be determined by the following equations
\begin{gather}
\hat{c}_{j}=\frac{1}{2}(a_m^2+\beta_{m}^2+\sum\limits_{l=1}^{m-1}(a_l^2+b_l^2))  \tilde{c}_j +\frac{\det \hat{D}_j}{\det \tilde{A}},\ j=1,2,\ldots,m.
\end{gather}

In order to construct local models to account for the locally non-uniform behavior around the endpoints $\pm i{\beta_{m}}$, we require that the  function $g_{{\beta_{m}}}(k)-kx-4k^3 t$ behaves as $(k \mp i{\beta_{m}})^{\frac{3}{2}}$ near $k=\pm i{\beta_{m}} $. Hence, for $\xi\in(\xi_{2m-2},\xi_{2m-1})$, ${\beta_{m}}(\xi)\in (a_m,b_m)$ and the Whitham evolution equation holds
\begin{equation}
\label{alpha}
\xi=\dfrac{x}{t}= -12\dfrac{(i{\beta_{m}})^{2(m+1)}+\frac{1}{2}(a_m^2+\beta_{m}^2+\sum\limits_{j=1}^{m-1}(a_j^2+b_j^2))(i{\beta_{m}})^{2m}
+\sum\limits_{j=1}^{m} \hat{c}_j (i{\beta_{m}})^{2(m-j)}}{(i{\beta_{m}})^{2m}+\sum\limits_{j=1}^{m} \tilde{c}_j (i{\beta_{m}})^{2(m-j)}},
\end{equation}
which gives the modulation equation determining the motion of branch point $\beta_m$ as a function of $\xi$.
\begin{lemma}
Function ${\beta_{m}}(\xi)$ determined by \eqref{alpha} is a monotonically increasing function with respect to $\xi$, this is $\dfrac{\d {\beta_{m}}(\xi)}{\d \xi}>0$. Moreover, $\lim\limits_{\xi\to\xi_{2m-2}}{\beta_{m}}(\xi)=a_{m}$ and $\lim\limits_{\xi\to\xi_{2m-1}}{\beta_{m}}(\xi)=b_{m}$.
\end{lemma}
\begin{proof}
Define $\hat{R}_{{\beta_{m}}}(\zeta)=\sqrt{(\zeta^2-a_m^2)(\zeta^2-\beta_{m}^2)\prod\limits_{j=1}^{m-1}(\zeta^2-a_j^2)(\zeta^2-b_j^2)}$. From equation \eqref{alpha}, we arrive at
\begin{equation}
\begin{aligned}
\xi
&=12\dfrac{\beta_{m}^{2(m+1)}-\frac{1}{2}(a_m^2+\beta_{m}^2+\sum\limits_{l=1}^{m-1}(a_l^2+b_l^2))\beta_{m}^{2m}
-\sum\limits_{j=1}^{m} \hat{c}_j (-1)^{j}\beta_{m}^{2(m-j)}}
{\beta_{m}^{2m}+\sum\limits_{j=1}^{m} \tilde{c}_j (-1)^{j}\beta_{m}^{2(m-j)}}\\
&=12\le(-\frac{1}{2}(a_m^2+\sum\limits_{l=1}^{m-1}(a_l^2+b_l^2))+\frac{1}{2}{\beta_{m}}^2+\frac{\det H}{\det G} \ri),
\end{aligned}
\end{equation}
where $H$ and $G$ are two $m\times m$ matrices given by
\begin{align}
H=\begin{pmatrix}
\int_{b_{m-1}}^{a_{m}}\frac{\zeta^{2m}(\beta_{m}^2-\zeta^2)}{\hat{R}_{\beta_{m}}(\zeta)}\d\zeta & \int_{b_{m-1}}^{a_{m}}\frac{\zeta^{2(m-2)}(\beta_{m}^2-\zeta^2)}{\hat{R}_{\beta_{m}}(\zeta)}\d\zeta &
\cdots &
\int_{b_{m-1}}^{a_{m}}\frac{\beta_{m}^2-\zeta^2}{\hat{R}_{\beta_{m}}(\zeta)}\d\zeta\\
\int_{b_{m-2}}^{a_{m-1}}\frac{\zeta^{2m}(\beta_{m}^2-\zeta^2)}{\hat{R}_{\beta_{m}}(\zeta)}\d\zeta & \int_{b_{m-2}}^{a_{m-1}}\frac{\zeta^{2(m-2)}(\beta_{m}^2-\zeta^2)}{\hat{R}_{\beta_{m}}(\zeta)}\d\zeta &
\cdots &
\int_{b_{m-2}}^{a_{m-1}}\frac{\beta_{m}^2-\zeta^2}{\hat{R}_{\beta_{m}}(\zeta)}\d\zeta\\
\cdots&\cdots&\cdots&\cdots\\
\int_{-a_1}^{a_1}\frac{\zeta^{2m}(\beta_{m}^2-\zeta^2)}{\hat{R}_{\beta_{m}}(\zeta)}\d\zeta & \int_{-a_1}^{a_1}\frac{\zeta^{2(m-2)}(\beta_{m}^2-\zeta^2)}{\hat{R}_{\beta_{m}}(\zeta)}\d\zeta &
\cdots &
\int_{-a_1}^{a_1}\frac{\beta_{m}^2-\zeta^2}{\hat{R}_{\beta_{m}}(\zeta)}\d\zeta
\end{pmatrix},
\end{align}

\begin{align}
G=\begin{pmatrix}
\int_{b_{m-1}}^{a_{m}}\frac{\zeta^{2(m-1)}(\beta_{m}^2-\zeta^2)}{\hat{R}_{\beta_{m}}(\zeta)}\d\zeta & \int_{b_{m-1}}^{a_{m}}\frac{\zeta^{2(m-2)}(\beta_{m}^2-\zeta^2)}{\hat{R}_{\beta_{m}}(\zeta)}\d\zeta &
\cdots &
\int_{b_{m-1}}^{a_{m}}\frac{\beta_{m}^2-\zeta^2}{\hat{R}_{\beta_{m}}(\zeta)}\d\zeta\\
\int_{b_{m-2}}^{a_{m-1}}\frac{\zeta^{2(m-1)}(\beta_{m}^2-\zeta^2)}{\hat{R}_{\beta_{m}}(\zeta)}\d\zeta & \int_{b_{m-2}}^{a_{m-1}}\frac{\zeta^{2(m-2)}(\beta_{m}^2-\zeta^2)}{\hat{R}_{\beta_{m}}(\zeta)}\d\zeta &
\cdots &
\int_{b_{m-2}}^{a_{m-1}}\frac{\beta_{m}^2-\zeta^2}{\hat{R}_{\beta_{m}}(\zeta)}\d\zeta\\
\cdots&\cdots&\cdots&\cdots\\
\int_{-a_1}^{a_1}\frac{\zeta^{2(m-1)}(\beta_{m}^2-\zeta^2)}{\hat{R}_{\beta_{m}}(\zeta)}\d\zeta & \int_{-a_1}^{a_1}\frac{\zeta^{2(m-2)}(\beta_{m}^2-\zeta^2)}{\hat{R}_{\beta_{m}}(\zeta)}\d\zeta &
\cdots &
\int_{-a_1}^{a_1}\frac{\beta_{m}^2-\zeta^2}{\hat{R}_{\beta_{m}}(\zeta)}\d\zeta
\end{pmatrix}.
\end{align}
Next, we consider the derivation of $\xi(\beta_{m})$ with respect to ${\beta_{m}}$
\begin{equation}
\dfrac{\d\xi(\beta_{m})}{\d{\beta_{m}}}=
12\frac{\beta_{m}(\det G)^2+\frac{\d (\det H)}{\d {\beta_m}} \det G-\frac{\d (\det G)}{\d {\beta_m}} \det H }{(\det G)^2}.
\end{equation}
Then $\frac{\d (\det H)}{\d {\beta_m}}$ can be written in the following form
\begin{gather*}
\frac{\d (\det H)}{\d {\beta_m}}=\beta_m\sum\limits_{j=1}^{m} \det H_j,
\end{gather*}
where $H_j$ are $m\times m$ matrices defined by
\begin{align*}
H_1=\begin{pmatrix}
\int_{b_{m-1}}^{a_{m}}\frac{\zeta^{2m}}{\hat{R}_{\beta_{m}}(\zeta)}\d\zeta & \int_{b_{m-1}}^{a_{m}}\frac{\zeta^{2(m-2)}(\beta_{m}^2-\zeta^2)}{\hat{R}_{\beta_{m}}(\zeta)}\d\zeta &
\cdots &
\int_{b_{m-1}}^{a_{m}}\frac{\beta_{m}^2-\zeta^2}{\hat{R}_{\beta_{m}}(\zeta)}\d\zeta\\
\int_{b_{m-2}}^{a_{m-1}}\frac{\zeta^{2m}}{\hat{R}_{\beta_{m}}(\zeta)}\d\zeta & \int_{b_{m-2}}^{a_{m-1}}\frac{\zeta^{2(m-2)}(\beta_{m}^2-\zeta^2)}{\hat{R}_{\beta_{m}}(\zeta)}\d\zeta &
\cdots &
\int_{b_{m-2}}^{a_{m-1}}\frac{\beta_{m}^2-\zeta^2}{\hat{R}_{\beta_{m}}(\zeta)}\d\zeta\\
\cdots&\cdots&\cdots&\cdots\\
\int_{-a_1}^{a_1}\frac{\zeta^{2m}}{\hat{R}_{\beta_{m}}(\zeta)}\d\zeta & \int_{-a_1}^{a_1}\frac{\zeta^{2(m-2)}(\beta_{m}^2-\zeta^2)}{\hat{R}_{\beta_{m}}(\zeta)}\d\zeta &
\cdots &
\int_{-a_1}^{a_1}\frac{\beta_{m}^2-\zeta^2}{\hat{R}_{\beta_{m}}(\zeta)}\d\zeta
\end{pmatrix},
\end{align*}
\begin{align*}
H_2=\begin{pmatrix}
\int_{b_{m-1}}^{a_{m}}\frac{\zeta^{2m}(\beta_{m}^2-\zeta^2)}{\hat{R}_{\beta_{m}}(\zeta)}\d\zeta & \int_{b_{m-1}}^{a_{m}}\frac{\zeta^{2(m-2)}}{\hat{R}_{\beta_{m}}(\zeta)}\d\zeta &
\int_{b_{m-1}}^{a_{m}}\frac{\zeta^{2(m-3)}(\beta_{m}^2-\zeta^2)}{\hat{R}_{\beta_{m}}(\zeta)}\d\zeta &
\cdots &
\int_{b_{m-1}}^{a_{m}}\frac{\beta_{m}^2-\zeta^2}{\hat{R}_{\beta_{m}}(\zeta)}\d\zeta\\
\int_{b_{m-2}}^{a_{m-1}}\frac{\zeta^{2m}(\beta_{m}^2-\zeta^2)}{\hat{R}_{\beta_{m}}(\zeta)}\d\zeta & \int_{b_{m-2}}^{a_{m-1}}\frac{\zeta^{2(m-2)}}{\hat{R}_{\beta_{m}}(\zeta)}\d\zeta &
\int_{b_{m-2}}^{a_{m-1}}\frac{\zeta^{2(m-3)}(\beta_{m}^2-\zeta^2)}{\hat{R}_{\beta_{m}}(\zeta)}\d\zeta &
\cdots &
\int_{b_{m-2}}^{a_{m-1}}\frac{\beta_{m}^2-\zeta^2}{\hat{R}_{\beta_{m}}(\zeta)}\d\zeta\\
\cdots&\cdots&\cdots&\cdots&\cdots\\
\int_{-a_1}^{a_1}\frac{\zeta^{2m}(\beta_{m}^2-\zeta^2)}{\hat{R}_{\beta_{m}}(\zeta)}\d\zeta & \int_{-a_1}^{a_1}\frac{\zeta^{2(m-2)}}{\hat{R}_{\beta_{m}}(\zeta)}\d\zeta &
\int_{-a_1}^{a_{1}}\frac{\zeta^{2(m-3)}(\beta_{m}^2-\zeta^2)}{\hat{R}_{\beta_{m}}(\zeta)}\d\zeta &
\cdots &
\int_{-a_1}^{a_1}\frac{\beta_{m}^2-\zeta^2}{\hat{R}_{\beta_{m}}(\zeta)}\d\zeta
\end{pmatrix},
\end{align*}
\begin{align*}
H_j=\begin{pmatrix}
\int_{b_{m-1}}^{a_{m}}\frac{\zeta^{2m}(\beta_{m}^2-\zeta^2)}{\hat{R}_{\beta_{m}}(\zeta)}\d\zeta & \int_{b_{m-1}}^{a_{m}}\frac{\zeta^{2(m-2)}(\beta_{m}^2-\zeta^2)}{\hat{R}_{\beta_{m}}(\zeta)}\d\zeta &
\cdots &
\int_{b_{m-1}}^{a_{m}}\frac{\zeta^{2(m-j)}}{\hat{R}_{\beta_{m}}(\zeta)}\d\zeta &
\cdots&
\int_{b_{m-1}}^{a_{m}}\frac{\beta_{m}^2-\zeta^2}{\hat{R}_{\beta_{m}}(\zeta)}\d\zeta\\
\int_{b_{m-2}}^{a_{m-1}}\frac{\zeta^{2m}(\beta_{m}^2-\zeta^2)}{\hat{R}_{\beta_{m}}(\zeta)}\d\zeta & \int_{b_{m-2}}^{a_{m-1}}\frac{\zeta^{2(m-2)}(\beta_{m}^2-\zeta^2)}{\hat{R}_{\beta_{m}}(\zeta)}\d\zeta &
\cdots &
\int_{b_{m-2}}^{a_{m-1}}\frac{\zeta^{2(m-j)}}{\hat{R}_{\beta_{m}}(\zeta)}\d\zeta &
\cdots&
\int_{b_{m-2}}^{a_{m-1}}\frac{\beta_{m}^2-\zeta^2}{\hat{R}_{\beta_{m}}(\zeta)}\d\zeta\\
\cdots&\cdots&\cdots&\cdots&\cdots&\cdots\\
\int_{-a_1}^{a_1}\frac{\zeta^{2m}(\beta_{m}^2-\zeta^2)}{\hat{R}_{\beta_{m}}(\zeta)}\d\zeta & \int_{-a_1}^{a_1}\frac{\zeta^{2(m-2)}(\beta_{m}^2-\zeta^2)}{\hat{R}_{\beta_{m}}(\zeta)}\d\zeta &
\cdots &
\int_{-a_1}^{a_{1}}\frac{\zeta^{2(m-j)}}{\hat{R}_{\beta_{m}}(\zeta)}\d\zeta &
\cdots&
\int_{-a_1}^{a_1}\frac{\beta_{m}^2-\zeta^2}{\hat{R}_{\beta_{m}}(\zeta)}\d\zeta
\end{pmatrix},
\end{align*}
for $j=3,4,\ldots,m$.
$\frac{\d (\det G)}{\d {\beta_m}}$ can be expressed as
\begin{gather*}
\frac{\d (\det G)}{\d {\beta_m}}=\beta_m\sum\limits_{j=1}^{m} \det G_j,
\end{gather*}
where $G_j$ are $m\times m$  matrices given by
\begin{align*}
G_1=\begin{pmatrix}
\int_{b_{m-1}}^{a_{m}}\frac{\zeta^{2(m-1)}}{\hat{R}_{\beta_{m}}(\zeta)}\d\zeta & \int_{b_{m-1}}^{a_{m}}\frac{\zeta^{2(m-2)}(\beta_{m}^2-\zeta^2)}{\hat{R}_{\beta_{m}}(\zeta)}\d\zeta &
\cdots &
\int_{b_{m-1}}^{a_{m}}\frac{\beta_{m}^2-\zeta^2}{\hat{R}_{\beta_{m}}(\zeta)}\d\zeta\\
\int_{b_{m-2}}^{a_{m-1}}\frac{\zeta^{2(m-1)}}{\hat{R}_{\beta_{m}}(\zeta)}\d\zeta & \int_{b_{m-2}}^{a_{m-1}}\frac{\zeta^{2(m-2)}(\beta_{m}^2-\zeta^2)}{\hat{R}_{\beta_{m}}(\zeta)}\d\zeta &
\cdots &
\int_{b_{m-2}}^{a_{m-1}}\frac{\beta_{m}^2-\zeta^2}{\hat{R}_{\beta_{m}}(\zeta)}\d\zeta\\
\cdots&\cdots&\cdots&\cdots\\
\int_{-a_1}^{a_1}\frac{\zeta^{2(m-1)}}{\hat{R}_{\beta_{m}}(\zeta)}\d\zeta & \int_{-a_1}^{a_1}\frac{\zeta^{2(m-2)}(\beta_{m}^2-\zeta^2)}{\hat{R}_{\beta_{m}}(\zeta)}\d\zeta &
\cdots &
\int_{-a_1}^{a_1}\frac{\beta_{m}^2-\zeta^2}{\hat{R}_{\beta_{m}}(\zeta)}\d\zeta
\end{pmatrix},
\end{align*}
\begin{align*}
G_2=\begin{pmatrix}
\int_{b_{m-1}}^{a_{m}}\frac{\zeta^{2(m-1)}(\beta_{m}^2-\zeta^2)}{\hat{R}_{\beta_{m}}(\zeta)}\d\zeta & \int_{b_{m-1}}^{a_{m}}\frac{\zeta^{2(m-2)}}{\hat{R}_{\beta_{m}}(\zeta)}\d\zeta &
\int_{b_{m-1}}^{a_{m}}\frac{\zeta^{2(m-3)}(\beta_{m}^2-\zeta^2)}{\hat{R}_{\beta_{m}}(\zeta)}\d\zeta &
\cdots &
\int_{b_{m-1}}^{a_{m}}\frac{\beta_{m}^2-\zeta^2}{\hat{R}_{\beta_{m}}(\zeta)}\d\zeta\\
\int_{b_{m-2}}^{a_{m-1}}\frac{\zeta^{2(m-1)}(\beta_{m}^2-\zeta^2)}{\hat{R}_{\beta_{m}}(\zeta)}\d\zeta & \int_{b_{m-2}}^{a_{m-1}}\frac{\zeta^{2(m-2)}}{\hat{R}_{\beta_{m}}(\zeta)}\d\zeta &
\int_{b_{m-2}}^{a_{m-1}}\frac{\zeta^{2(m-3)}(\beta_{m}^2-\zeta^2)}{\hat{R}_{\beta_{m}}(\zeta)}\d\zeta &
\cdots &
\int_{b_{m-2}}^{a_{m-1}}\frac{\beta_{m}^2-\zeta^2}{\hat{R}_{\beta_{m}}(\zeta)}\d\zeta\\
\cdots&\cdots&\cdots&\cdots&\cdots\\
\int_{-a_1}^{a_1}\frac{\zeta^{2(m-1)}(\beta_{m}^2-\zeta^2)}{\hat{R}_{\beta_{m}}(\zeta)}\d\zeta & \int_{-a_1}^{a_1}\frac{\zeta^{2(m-2)}}{\hat{R}_{\beta_{m}}(\zeta)}\d\zeta &
\int_{-a_1}^{a_{1}}\frac{\zeta^{2(m-3)}(\beta_{m}^2-\zeta^2)}{\hat{R}_{\beta_{m}}(\zeta)}\d\zeta &
\cdots &
\int_{-a_1}^{a_1}\frac{\beta_{m}^2-\zeta^2}{\hat{R}_{\beta_{m}}(\zeta)}\d\zeta
\end{pmatrix},
\end{align*}
\begin{align*}
G_j=\begin{pmatrix}
\int_{b_{m-1}}^{a_{m}}\frac{\zeta^{2(m-1)}(\beta_{m}^2-\zeta^2)}{\hat{R}_{\beta_{m}}(\zeta)}\d\zeta & \int_{b_{m-1}}^{a_{m}}\frac{\zeta^{2(m-2)}(\beta_{m}^2-\zeta^2)}{\hat{R}_{\beta_{m}}(\zeta)}\d\zeta &
\cdots &
\int_{b_{m-1}}^{a_{m}}\frac{\zeta^{2(m-j)}}{\hat{R}_{\beta_{m}}(\zeta)}\d\zeta &
\cdots&
\int_{b_{m-1}}^{a_{m}}\frac{\beta_{m}^2-\zeta^2}{\hat{R}_{\beta_{m}}(\zeta)}\d\zeta\\
\int_{b_{m-2}}^{a_{m-1}}\frac{\zeta^{2(m-1)}(\beta_{m}^2-\zeta^2)}{\hat{R}_{\beta_{m}}(\zeta)}\d\zeta & \int_{b_{m-2}}^{a_{m-1}}\frac{\zeta^{2(m-2)}(\beta_{m}^2-\zeta^2)}{\hat{R}_{\beta_{m}}(\zeta)}\d\zeta &
\cdots &
\int_{b_{m-2}}^{a_{m-1}}\frac{\zeta^{2(m-j)}}{\hat{R}_{\beta_{m}}(\zeta)}\d\zeta &
\cdots&
\int_{b_{m-2}}^{a_{m-1}}\frac{\beta_{m}^2-\zeta^2}{\hat{R}_{\beta_{m}}(\zeta)}\d\zeta\\
\cdots&\cdots&\cdots&\cdots&\cdots&\cdots\\
\int_{-a_1}^{a_1}\frac{\zeta^{2(m-1)}(\beta_{m}^2-\zeta^2)}{\hat{R}_{\beta_{m}}(\zeta)}\d\zeta & \int_{-a_1}^{a_1}\frac{\zeta^{2(m-2)}(\beta_{m}^2-\zeta^2)}{\hat{R}_{\beta_{m}}(\zeta)}\d\zeta &
\cdots &
\int_{-a_1}^{a_{1}}\frac{\zeta^{2(m-j)}}{\hat{R}_{\beta_{m}}(\zeta)}\d\zeta &
\cdots&
\int_{-a_1}^{a_1}\frac{\beta_{m}^2-\zeta^2}{\hat{R}_{\beta_{m}}(\zeta)}\d\zeta
\end{pmatrix},
\end{align*}
for $j=3,4,\ldots,m$.
We find that
\begin{align*}
&(\det G)^2+\det H_1 \det G- \det G_1 \det H\\
&=\det \begin{pmatrix}
\int_{b_{m-1}}^{a_{m}}\frac{\zeta^{2(m-1)}(\beta_{m}^2-\zeta^2)^2}{\hat{R}_{\beta_{m}}(\zeta)}\d\zeta & \int_{b_{m-1}}^{a_{m}}\frac{\zeta^{2(m-2)}(\beta_{m}^2-\zeta^2)}{\hat{R}_{\beta_{m}}(\zeta)}\d\zeta &
\cdots &
\int_{b_{m-1}}^{a_{m}}\frac{\beta_{m}^2-\zeta^2}{\hat{R}_{\beta_{m}}(\zeta)}\d\zeta\\
\int_{b_{m-2}}^{a_{m-1}}\frac{\zeta^{2(m-1)}(\beta_{m}^2-\zeta^2)^2}{\hat{R}_{\beta_{m}}(\zeta)}\d\zeta & \int_{b_{m-2}}^{a_{m-1}}\frac{\zeta^{2(m-2)}(\beta_{m}^2-\zeta^2)}{\hat{R}_{\beta_{m}}(\zeta)}\d\zeta &
\cdots &
\int_{b_{m-2}}^{a_{m-1}}\frac{\beta_{m}^2-\zeta^2}{\hat{R}_{\beta_{m}}(\zeta)}\d\zeta\\
\cdots&\cdots&\cdots&\cdots\\
\int_{-a_1}^{a_1}\frac{\zeta^{2(m-1)}(\beta_{m}^2-\zeta^2)^2}{\hat{R}_{\beta_{m}}(\zeta)}\d\zeta & \int_{-a_1}^{a_1}\frac{\zeta^{2(m-2)}(\beta_{m}^2-\zeta^2)}{\hat{R}_{\beta_{m}}(\zeta)}\d\zeta &
\cdots &
\int_{-a_1}^{a_1}\frac{\beta_{m}^2-\zeta^2}{\hat{R}_{\beta_{m}}(\zeta)}\d\zeta
\end{pmatrix}
\\ &\times
\det \begin{pmatrix}
\int_{b_{m-1}}^{a_{m}}\frac{\zeta^{2(m-1)}}{\hat{R}_{\beta_{m}}(\zeta)}\d\zeta & \int_{b_{m-1}}^{a_{m}}\frac{\zeta^{2(m-2)}(\beta_{m}^2-\zeta^2)}{\hat{R}_{\beta_{m}}(\zeta)}\d\zeta &
\cdots &
\int_{b_{m-1}}^{a_{m}}\frac{\beta_{m}^2-\zeta^2}{\hat{R}_{\beta_{m}}(\zeta)}\d\zeta\\
\int_{b_{m-2}}^{a_{m-1}}\frac{\zeta^{2(m-1)}}{\hat{R}_{\beta_{m}}(\zeta)}\d\zeta & \int_{b_{m-2}}^{a_{m-1}}\frac{\zeta^{2(m-2)}(\beta_{m}^2-\zeta^2)}{\hat{R}_{\beta_{m}}(\zeta)}\d\zeta &
\cdots &
\int_{b_{m-2}}^{a_{m-1}}\frac{\beta_{m}^2-\zeta^2}{\hat{R}_{\beta_{m}}(\zeta)}\d\zeta\\
\cdots&\cdots&\cdots&\cdots\\
\int_{-a_1}^{a_1}\frac{\zeta^{2(m-1)}}{\hat{R}_{\beta_{m}}(\zeta)}\d\zeta & \int_{-a_1}^{a_1}\frac{\zeta^{2(m-2)}(\beta_{m}^2-\zeta^2)}{\hat{R}_{\beta_{m}}(\zeta)}\d\zeta &
\cdots &
\int_{-a_1}^{a_1}\frac{\beta_{m}^2-\zeta^2}{\hat{R}_{\beta_{m}}(\zeta)}\d\zeta
\end{pmatrix}.
\end{align*}
Using Laplace's theorem and computing the determinant of  $G$, $H$, $H_2$ and $G_2$ by expanding their first two columns yield
\begin{align*}
&\det H_2 \det G- \det G_2 \det H\\
&=\det\begin{pmatrix}
\int_{b_{m-1}}^{a_{m}}\frac{\zeta^{2m}(\beta_{m}^2-\zeta^2)}{\hat{R}_{\beta_{m}}(\zeta)}\d\zeta & \int_{b_{m-1}}^{a_{m}}\frac{\zeta^{2(m-1)}(\beta_{m}^2-\zeta^2)}{\hat{R}_{\beta_{m}}(\zeta)}\d\zeta &
\int_{b_{m-1}}^{a_{m}}\frac{\zeta^{2(m-3)}(\beta_{m}^2-\zeta^2)}{\hat{R}_{\beta_{m}}(\zeta)}\d\zeta &
\cdots &
\int_{b_{m-1}}^{a_{m}}\frac{\beta_{m}^2-\zeta^2}{\hat{R}_{\beta_{m}}(\zeta)}\d\zeta\\
\int_{b_{m-2}}^{a_{m-1}}\frac{\zeta^{2m}(\beta_{m}^2-\zeta^2)}{\hat{R}_{\beta_{m}}(\zeta)}\d\zeta & \int_{b_{m-2}}^{a_{m-1}}\frac{\zeta^{2(m-1)}(\beta_{m}^2-\zeta^2)}{\hat{R}_{\beta_{m}}(\zeta)}\d\zeta &
\int_{b_{m-2}}^{a_{m-1}}\frac{\zeta^{2(m-3)}(\beta_{m}^2-\zeta^2)}{\hat{R}_{\beta_{m}}(\zeta)}\d\zeta &
\cdots &
\int_{b_{m-2}}^{a_{m-1}}\frac{\beta_{m}^2-\zeta^2}{\hat{R}_{\beta_{m}}(\zeta)}\d\zeta\\
\cdots&\cdots&\cdots&\cdots&\cdots\\
\int_{-a_1}^{a_1}\frac{\zeta^{2m}(\beta_{m}^2-\zeta^2)}{\hat{R}_{\beta_{m}}(\zeta)}\d\zeta & \int_{-a_1}^{a_1}\frac{\zeta^{2(m-1)}(\beta_{m}^2-\zeta^2)}{\hat{R}_{\beta_{m}}(\zeta)}\d\zeta &
\int_{-a_1}^{a_{1}}\frac{\zeta^{2(m-3)}(\beta_{m}^2-\zeta^2)}{\hat{R}_{\beta_{m}}(\zeta)}\d\zeta &
\cdots &
\int_{-a_1}^{a_1}\frac{\beta_{m}^2-\zeta^2}{\hat{R}_{\beta_{m}}(\zeta)}\d\zeta
\end{pmatrix}\\
&\times \det \begin{pmatrix}
\int_{b_{m-1}}^{a_{m}}\frac{-\zeta^{2(m-2)}(\beta_{m}^2-\zeta^2)}{\hat{R}_{\beta_{m}}(\zeta)}\d\zeta & \int_{b_{m-1}}^{a_{m}}\frac{\zeta^{2(m-2)}}{\hat{R}_{\beta_{m}}(\zeta)}\d\zeta &
\int_{b_{m-1}}^{a_{m}}\frac{\zeta^{2(m-3)}(\beta_{m}^2-\zeta^2)}{\hat{R}_{\beta_{m}}(\zeta)}\d\zeta &
\cdots &
\int_{b_{m-1}}^{a_{m}}\frac{\beta_{m}^2-\zeta^2}{\hat{R}_{\beta_{m}}(\zeta)}\d\zeta\\
\int_{b_{m-2}}^{a_{m-1}}\frac{-\zeta^{2(m-2)}(\beta_{m}^2-\zeta^2)}{\hat{R}_{\beta_{m}}(\zeta)}\d\zeta & \int_{b_{m-2}}^{a_{m-1}}\frac{\zeta^{2(m-2)}}{\hat{R}_{\beta_{m}}(\zeta)}\d\zeta &
\int_{b_{m-2}}^{a_{m-1}}\frac{\zeta^{2(m-3)}(\beta_{m}^2-\zeta^2)}{\hat{R}_{\beta_{m}}(\zeta)}\d\zeta &
\cdots &
\int_{b_{m-2}}^{a_{m-1}}\frac{\beta_{m}^2-\zeta^2}{\hat{R}_{\beta_{m}}(\zeta)}\d\zeta\\
\cdots&\cdots&\cdots&\cdots&\cdots\\
\int_{-a_1}^{a_1}\frac{-\zeta^{2(m-2)}(\beta_{m}^2-\zeta^2)}{\hat{R}_{\beta_{m}}(\zeta)}\d\zeta & \int_{-a_1}^{a_1}\frac{\zeta^{2(m-2)}}{\hat{R}_{\beta_{m}}(\zeta)}\d\zeta &
\int_{-a_1}^{a_{1}}\frac{\zeta^{2(m-3)}(\beta_{m}^2-\zeta^2)}{\hat{R}_{\beta_{m}}(\zeta)}\d\zeta &
\cdots &
\int_{-a_1}^{a_1}\frac{\beta_{m}^2-\zeta^2}{\hat{R}_{\beta_{m}}(\zeta)}\d\zeta
\end{pmatrix}.
\end{align*}
As for $\det H_j \det G- \det G_j \det H$, $j=3,\ldots,m$, we can use the column transformation to shift the j-th column of $G$, $H$, $H_j$ and $G_j$  to the second column, transforming it to the above form.
Then we obtain
\begin{equation}
\dfrac{\d\xi(\beta_{m})}{\d{\beta_{m}}}=
12\beta_{m}\frac{(\det G)^2+\det G \sum\limits_{j=1}^{m}\det H_j- \det H \sum\limits_{j=1}^m\det G_j }{(\det G)^2}=
12\beta_{m}\frac{\det \tilde{M} \det\tilde{N}}{(\det G)^2},
\end{equation}
where $\tilde{M}$ and $\tilde{N}$ are $m\times m$  matrices by
\begin{equation}
\begin{aligned}
\begin{array}{l}
\tilde{M}=\begin{pmatrix}
\int_{b_{m-1}}^{a_{m}}\frac{\zeta^{2(m-1)}}{\hat{R}_{\beta_{m}}(\zeta)}\d\zeta & \int_{b_{m-1}}^{a_{m}}\frac{\zeta^{2(m-2)}}{\hat{R}_{\beta_{m}}(\zeta)}\d\zeta &
\cdots &
\int_{b_{m-1}}^{a_{m}}\frac{1}{\hat{R}_{\beta_{m}}(\zeta)}\d\zeta\\
\int_{b_{m-2}}^{a_{m-1}}\frac{\zeta^{2(m-1)}}{\hat{R}_{\beta_{m}}(\zeta)}\d\zeta & \int_{b_{m-2}}^{a_{m-1}}\frac{\zeta^{2(m-2)}}{\hat{R}_{\beta_{m}}(\zeta)}\d\zeta &
\cdots &
\int_{b_{m-2}}^{a_{m-1}}\frac{1}{\hat{R}_{\beta_{m}}(\zeta)}\d\zeta\\
\cdots&\cdots&\cdots&\cdots\\
\int_{-a_1}^{a_1}\frac{\zeta^{2(m-1)}}{\hat{R}_{\beta_{m}}(\zeta)}\d\zeta & \int_{-a_1}^{a_1}\frac{\zeta^{2(m-2)}}{\hat{R}_{\beta_{m}}(\zeta)}\d\zeta &
\cdots &
\int_{-a_1}^{a_1}\frac{1}{\hat{R}_{\beta_{m}}(\zeta)}\d\zeta
\end{pmatrix},
\end{array}
\end{aligned}
\end{equation}
and
\begin{equation}
\begin{aligned}
\begin{array}{l}
\tilde{N}=\begin{pmatrix}
\int_{b_{m-1}}^{a_{m}}\frac{\zeta^{2(m-1)}(\beta_{m}^2-\zeta^2)^2}{\hat{R}_{\beta_{m}}(\zeta)}\d\zeta & \int_{b_{m-1}}^{a_{m}}\frac{\zeta^{2(m-2)}(\beta_{m}^2-\zeta^2)^2}{\hat{R}_{\beta_{m}}(\zeta)}\d\zeta &
\cdots &
\int_{b_{m-1}}^{a_{m}}\frac{(\beta_{m}^2-\zeta^2)^2}{\hat{R}_{\beta_{m}}(\zeta)}\d\zeta\\
\int_{b_{m-2}}^{a_{m-1}}\frac{\zeta^{2(m-1)}(\beta_{m}^2-\zeta^2)^2}{\hat{R}_{\beta_{m}}(\zeta)}\d\zeta & \int_{b_{m-2}}^{a_{m-1}}\frac{\zeta^{2(m-2)}(\beta_{m}^2-\zeta^2)^2}{\hat{R}_{\beta_{m}}(\zeta)}\d\zeta &
\cdots &
\int_{b_{m-2}}^{a_{m-1}}\frac{(\beta_{m}^2-\zeta^2)^2}{\hat{R}_{\beta_{m}}(\zeta)}\d\zeta\\
\cdots&\cdots&\cdots&\cdots\\
\int_{-a_1}^{a_1}\frac{\zeta^{2(m-1)}(\beta_{m}^2-\zeta^2)^2}{\hat{R}_{\beta_{m}}(\zeta)}\d\zeta & \int_{-a_1}^{a_1}\frac{\zeta^{2(m-2)}(\beta_{m}^2-\zeta^2)^2}{\hat{R}_{\beta_{m}}(\zeta)}\d\zeta &
\cdots &
\int_{-a_1}^{a_1}\frac{(\beta_{m}^2-\zeta^2)^2}{\hat{R}_{\beta_{m}}(\zeta)}\d\zeta
\end{pmatrix}.
\end{array}
\end{aligned}
\end{equation}
Next we prove that the determinants of $\tilde{M}$ and $\tilde{N}$ are both positive. In fact, their determinants can be written in the following form
\begin{align}
\det \tilde{M}=\int_{-a_1}^{a_1}\d\zeta_1\int_{b_1}^{a_2}\d\zeta_2\cdots\int_{b_{m-1}}^{a_m}
\frac{\prod\limits_{1\leq l<j\leq m}(\zeta_{j}^2-\zeta_{l}^{2})}{\prod\limits_{j=1}^{m}\hat{R}_{\beta_{m}}(\zeta_j)}\d \zeta_{m},
\end{align}
\begin{align}
\det \tilde{N}=\int_{-a_1}^{a_1}\d\zeta_1\int_{b_1}^{a_2}\d\zeta_2\cdots\int_{b_{m-1}}^{a_m}
\le(\prod\limits_{1\leq l<j\leq m}(\zeta_{j}^2-\zeta_{l}^{2})\ri)
\prod\limits_{j=1}^{m}\frac{(\beta_{m}^2-\zeta_{j}^2)^2}{\hat{R}_{\beta_{m}}(\zeta_j)}\d \zeta_{m}.
\end{align}
Since $\zeta_1<\zeta_2<\cdots<\zeta_m$,
we obtain that $\frac{\d \xi(\beta_{m})}{\d \beta_{m}}>0$ for ${\beta_{m}}\in [a_m, b_m]$. Hence, by using the implicit function theorem, equation \eqref{alpha} defines ${\beta_{m}}$ as a  monotonically increasing function of $\xi$ for $\xi\in [\xi_{2m-2},\xi_{2m-1}]$. Here, where $\xi_{2m-2}$ and $\xi_{2m-1}$ are defined by
\begin{align}\label{xi2}
\xi_{2m-2}=12\le(-\frac{1}{2}\sum\limits_{l=1}^{m-1}(a_l^2+b_l^2)+\frac{\det H_{a_m}}{\det G_{a_m}} \ri),
\end{align}
\begin{align}\label{xi3}
\xi_{2m-1}=12\le(-\frac{1}{2}\sum\limits_{l=1}^{m-1}(a_l^2+b_l^2)-\frac{1}{2}a_m^2+\frac{1}{2}b_m^2+\frac{\det H_{b_m}}{\det G_{b_m}} \ri),
\end{align}
where \begin{align}
H_{a_{m}}=\begin{pmatrix}
\int_{b_{m-1}}^{a_{m}}\frac{\zeta^{2m}(a_{m}^2-\zeta^2)}{\hat{R}_{a_{m}}(\zeta)}\d\zeta & \int_{b_{m-1}}^{a_{m}}\frac{\zeta^{2(m-2)}(a_{m}^2-\zeta^2)}{\hat{R}_{a_{m}}(\zeta)}\d\zeta &
\cdots &
\int_{b_{m-1}}^{a_{m}}\frac{(a_{m}^2-\zeta^2)}{\hat{R}_{a_{m}}(\zeta)}\d\zeta\\
\int_{b_{m-2}}^{a_{m-1}}\frac{\zeta^{2m}(a_{m}^2-\zeta^2)}{\hat{R}_{a_{m}}(\zeta)}\d\zeta & \int_{b_{m-2}}^{a_{m-1}}\frac{\zeta^{2(m-2)}(a_{m}^2-\zeta^2)}{\hat{R}_{a_{m}}(\zeta)}\d\zeta &
\cdots &
\int_{b_{m-2}}^{a_{m-1}}\frac{(a_{m}^2-\zeta^2)}{\hat{R}_{a_{m}}(\zeta)}\d\zeta\\
\cdots&\cdots&\cdots&\cdots\\
\int_{-a_1}^{a_1}\frac{\zeta^{2m}(a_{m}^2-\zeta^2)}{\hat{R}_{a_{m}}(\zeta)}\d\zeta & \int_{-a_1}^{a_1}\frac{\zeta^{2(m-2)}(a_{m}^2-\zeta^2)}{\hat{R}_{a_{m}}(\zeta)}\d\zeta &
\cdots &
\int_{-a_1}^{a_1}\frac{(a_{m}^2-\zeta^2)}{\hat{R}_{a_{m}}(\zeta)}\d\zeta
\end{pmatrix},
\end{align}
\begin{align}
G_{a_m}=\begin{pmatrix}
\int_{b_{m-1}}^{a_{m}}\frac{\zeta^{2(m-1)}(a_{m}^2-\zeta^2)}{\hat{R}_{a_{m}}(\zeta)}\d\zeta & \int_{b_{m-1}}^{a_{m}}\frac{\zeta^{2(m-2)}(a_{m}^2-\zeta^2)}{\hat{R}_{a_{m}}(\zeta)}\d\zeta &
\cdots &
\int_{b_{m-1}}^{a_{m}}\frac{(a_{m}^2-\zeta^2)}{\hat{R}_{a_{m}}(\zeta)}\d\zeta\\
\int_{b_{m-2}}^{a_{m-1}}\frac{\zeta^{2(m-1)}(a_{m}^2-\zeta^2)}{\hat{R}_{a_{m}}(\zeta)}\d\zeta & \int_{b_{m-2}}^{a_{m-1}}\frac{\zeta^{2(m-2)}(a_{m}^2-\zeta^2)}{\hat{R}_{a_{m}}(\zeta)}\d\zeta &
\cdots &
\int_{b_{m-2}}^{a_{m-1}}\frac{(a_{m}^2-\zeta^2)}{\hat{R}_{a_{m}}(\zeta)}\d\zeta\\
\cdots&\cdots&\cdots&\cdots\\
\int_{-a_1}^{a_1}\frac{\zeta^{2(m-1)}(a_{m}^2-\zeta^2)}{\hat{R}_{a_{m}}(\zeta)}\d\zeta & \int_{-a_1}^{a_1}\frac{\zeta^{2(m-2)}(a_{m}^2-\zeta^2)}{\hat{R}_{a_{m}}(\zeta)}\d\zeta &
\cdots &
\int_{-a_1}^{a_1}\frac{(a_{m}^2-\zeta^2)}{\hat{R}_{a_{m}}(\zeta)}\d\zeta
\end{pmatrix},
\end{align}
with $\hat{R}_{a_{m}}(\zeta)=\sqrt{(\zeta^2-a_m^2)(\zeta^2-a_{m}^2)\prod\limits_{j=1}^{m-1}(\zeta^2-a_j^2)(\zeta^2-b_j^2)}$, and
\begin{align}
H_{b_{m}}=\begin{pmatrix}
\int_{b_{m-1}}^{a_{m}}\frac{\zeta^{2m}(b_{m}^2-\zeta^2)}{\hat{R}_{b_{m}}(\zeta)}\d\zeta & \int_{b_{m-1}}^{a_{m}}\frac{\zeta^{2(m-2)}(b_{m}^2-\zeta^2)}{\hat{R}_{b_{m}}(\zeta)}\d\zeta &
\cdots &
\int_{b_{m-1}}^{a_{m}}\frac{(b_{m}^2-\zeta^2)}{\hat{R}_{b_{m}}(\zeta)}\d\zeta\\
\int_{b_{m-2}}^{a_{m-1}}\frac{\zeta^{2m}(b_{m}^2-\zeta^2)}{\hat{R}_{b_{m}}(\zeta)}\d\zeta & \int_{b_{m-2}}^{a_{m-1}}\frac{\zeta^{2(m-2)}(b_{m}^2-\zeta^2)}{\hat{R}_{b_{m}}(\zeta)}\d\zeta &
\cdots &
\int_{b_{m-2}}^{a_{m-1}}\frac{(b_{m}^2-\zeta^2)}{\hat{R}_{b_{m}}(\zeta)}\d\zeta\\
\cdots&\cdots&\cdots&\cdots\\
\int_{-a_1}^{a_1}\frac{\zeta^{2m}(b_{m}^2-\zeta^2)}{\hat{R}_{b_{m}}(\zeta)}\d\zeta & \int_{-a_1}^{a_1}\frac{\zeta^{2(m-2)}(b_{m}^2-\zeta^2)}{\hat{R}_{b_{m}}(\zeta)}\d\zeta &
\cdots &
\int_{-a_1}^{a_1}\frac{(b_{m}^2-\zeta^2)}{\hat{R}_{b_{m}}(\zeta)}\d\zeta
\end{pmatrix},
\end{align}
\begin{align}
G_{b_m}=\begin{pmatrix}
\int_{b_{m-1}}^{a_{m}}\frac{\zeta^{2(m-1)}(b_{m}^2-\zeta^2)}{\hat{R}_{b_{m}}(\zeta)}\d\zeta & \int_{b_{m-1}}^{a_{m}}\frac{\zeta^{2(m-2)}(b_{m}^2-\zeta^2)}{\hat{R}_{b_{m}}(\zeta)}\d\zeta &
\cdots &
\int_{b_{m-1}}^{a_{m}}\frac{(b_{m}^2-\zeta^2)}{\hat{R}_{b_{m}}(\zeta)}\d\zeta\\
\int_{b_{m-2}}^{a_{m-1}}\frac{\zeta^{2(m-1)}(b_{m}^2-\zeta^2)}{\hat{R}_{b_{m}}(\zeta)}\d\zeta & \int_{b_{m-2}}^{a_{m-1}}\frac{\zeta^{2(m-2)}(b_{m}^2-\zeta^2)}{\hat{R}_{b_{m}}(\zeta)}\d\zeta &
\cdots &
\int_{b_{m-2}}^{a_{m-1}}\frac{(b_{m}^2-\zeta^2)}{\hat{R}_{b_{m}}(\zeta)}\d\zeta\\
\cdots&\cdots&\cdots&\cdots\\
\int_{-a_1}^{a_1}\frac{\zeta^{2(m-1)}(b_{m}^2-\zeta^2)}{\hat{R}_{b_{m}}(\zeta)}\d\zeta & \int_{-a_1}^{a_1}\frac{\zeta^{2(m-2)}(b_{m}^2-\zeta^2)}{\hat{R}_{b_{m}}(\zeta)}\d\zeta &
\cdots &
\int_{-a_1}^{a_1}\frac{(b_{m}^2-\zeta^2)}{\hat{R}_{b_{m}}(\zeta)}\d\zeta
\end{pmatrix},
\end{align}
with $\hat{R}_{b_{m}}(\zeta)=\sqrt{
\prod\limits_{j=1}^{m}(\zeta^2-a_j^2)(\zeta^2-b_j^2)}$ .
\end{proof}
Note that for the case $m=1$, the formula \eqref{xi3} for $\xi_{2m-1}$ still holds.
\begin{lemma}
$\xi_{2m-1}<\xi_{2m}$, for $m=1,2,\ldots,n$.
\end{lemma}
\begin{proof}
It is easy to check that
\begin{align}
\xi_{2m}-\xi_{2m-1}=12(\frac{\det H_{a_{m+1}}}{\det G_{a_{m+1}}}-b_m^2-\frac{\det H_{b_m}}{\det G_{b_m}}),\ m=1,2,\ldots,n,
\end{align}
since $\det G_{a_{m+1}}>0$ and $\det G_{b_m}>0$, we only need to prove that $$\det H_{a_{m+1}}\det G_{b_m}-b_m^2\det G_{a_{m+1}}\det G_{b_m}-\det H_{b_m}\det G_{a_{m+1}}>0.$$
By calculation,
\begin{align*}
&\det H_{a_{m+1}}\det G_{b_m}-b_m^2\det G_{a_{m+1}}\det G_{b_m}-\det H_{b_m}\det G_{a_{m+1}}\\
&=\det\begin{pmatrix}
\int_{b_{m}}^{a_{m+1}}\frac{\zeta^{2m}(\zeta^2-b_{m}^2)}{\hat{R}_{b_{m}}(\zeta)}\d\zeta & \int_{b_{m}}^{a_{m+1}}\frac{\zeta^{2(m-2)}(\zeta^2-b_{m}^2)}{\hat{R}_{b_{m}}(\zeta)}\d\zeta &
\cdots &
\int_{b_{m}}^{a_{m+1}}\frac{(\zeta^2-b_{m}^2)}{\hat{R}_{b_{m}}(\zeta)}\d\zeta
&\int_{b_{m}}^{a_{m+1}}\frac{1}{\hat{R}_{b_{m}}(\zeta)}\d\zeta\\
\int_{b_{m-1}}^{a_{m}}\frac{\zeta^{2m}(\zeta^2-b_{m}^2)}{\hat{R}_{b_{m}}(\zeta)}\d\zeta & \int_{b_{m-1}}^{a_{m}}\frac{\zeta^{2(m-2)}(\zeta^2-b_{m}^2)}{\hat{R}_{b_{m}}(\zeta)}\d\zeta &
\cdots &
\int_{b_{m-1}}^{a_{m}}\frac{(\zeta^2-b_{m}^2)}{\hat{R}_{b_{m}}(\zeta)}\d\zeta
&\int_{b_{m-1}}^{a_{m}}\frac{1}{\hat{R}_{b_{m}}(\zeta)}\d\zeta\\
\cdots&\cdots&\cdots&\cdots\\
\int_{-a_1}^{a_1}\frac{\zeta^{2m}(\zeta^2-b_{m}^2)}{\hat{R}_{b_{m}}(\zeta)}\d\zeta & \int_{-a_1}^{a_1}\frac{\zeta^{2(m-2)}(\zeta^2-b_{m}^2)}{\hat{R}_{b_{m}}(\zeta)}\d\zeta &
\cdots &
\int_{-a_1}^{a_1}\frac{(\zeta^2-b_{m}^2)}{\hat{R}_{b_{m}}(\zeta)}\d\zeta
&\int_{-a_1}^{a_1}\frac{1}{\hat{R}_{b_{m}}(\zeta)}\d\zeta
\end{pmatrix}_{(m+1)\times(m+1)}\\
&\times \det G_{b_m}\\
&-\det\begin{pmatrix}
\int_{b_{m}}^{a_{m+1}}\frac{\zeta^{2(m-1)}(\zeta^2-b_{m}^2)}{\hat{R}_{b_{m}}(\zeta)}\d\zeta & \int_{b_{m}}^{a_{m+1}}\frac{\zeta^{2(m-2)}(\zeta^2-b_{m}^2)}{\hat{R}_{b_{m}}(\zeta)}\d\zeta &
\cdots &
\int_{b_{m}}^{a_{m+1}}\frac{(\zeta^2-b_{m}^2)}{\hat{R}_{b_{m}}(\zeta)}\d\zeta
&\int_{b_{m}}^{a_{m+1}}\frac{1}{\hat{R}_{b_{m}}(\zeta)}\d\zeta\\
\int_{b_{m-1}}^{a_{m}}\frac{\zeta^{2(m-1)}(\zeta^2-b_{m}^2)}{\hat{R}_{b_{m}}(\zeta)}\d\zeta & \int_{b_{m-1}}^{a_{m}}\frac{\zeta^{2(m-2)}(\zeta^2-b_{m}^2)}{\hat{R}_{b_{m}}(\zeta)}\d\zeta &
\cdots &
\int_{b_{m-1}}^{a_{m}}\frac{(\zeta^2-b_{m}^2)}{\hat{R}_{b_{m}}(\zeta)}\d\zeta
&\int_{b_{m-1}}^{a_{m}}\frac{1}{\hat{R}_{b_{m}}(\zeta)}\d\zeta\\
\cdots&\cdots&\cdots&\cdots\\
\int_{-a_1}^{a_1}\frac{\zeta^{2(m-1)}(\zeta^2-b_{m}^2)}{\hat{R}_{b_{m}}(\zeta)}\d\zeta & \int_{-a_1}^{a_1}\frac{\zeta^{2(m-2)}(\zeta^2-b_{m}^2)}{\hat{R}_{b_{m}}(\zeta)}\d\zeta &
\cdots &
\int_{-a_1}^{a_1}\frac{(\zeta^2-b_{m}^2)}{\hat{R}_{b_{m}}(\zeta)}\d\zeta
&\int_{-a_1}^{a_1}\frac{1}{\hat{R}_{b_{m}}(\zeta)}\d\zeta
\end{pmatrix}_{(m+1)\times(m+1)}\\
&\times \det H_{b_{m}}\\
&=\det\begin{pmatrix}
\int_{b_{m}}^{a_{m+1}}\frac{\zeta^{2m}(\zeta^2-b_{m}^2)}{\hat{R}_{b_{m}}(\zeta)}\d\zeta & \int_{b_{m}}^{a_{m+1}}\frac{\zeta^{2(m-1)}(b_{m}^2-\zeta^2)}{\hat{R}_{b_{m}}(\zeta)}\d\zeta &
\cdots &
\int_{b_{m}}^{a_{m+1}}\frac{\zeta^2(b_{m}^2-\zeta^2)}{\hat{R}_{b_{m}}(\zeta)}\d\zeta
&\int_{b_{m}}^{a_{m+1}}\frac{b_{m}^2-\zeta^2}{\hat{R}_{b_{m}}(\zeta)}\d\zeta\\
\int_{b_{m-1}}^{a_{m}}\frac{\zeta^{2m}(\zeta^2-b_{m}^2)}{\hat{R}_{b_{m}}(\zeta)}\d\zeta & \int_{b_{m-1}}^{a_{m}}\frac{\zeta^{2(m-1)}(b_{m}^2-\zeta^2)}{\hat{R}_{b_{m}}(\zeta)}\d\zeta &
\cdots &
\int_{b_{m-1}}^{a_{m}}\frac{\zeta^2(b_{m}^2-\zeta^2)}{\hat{R}_{b_{m}}(\zeta)}\d\zeta
&\int_{b_{m-1}}^{a_{m}}\frac{b_{m}^2-\zeta^2}{\hat{R}_{b_{m}}(\zeta)}\d\zeta\\
\cdots&\cdots&\cdots&\cdots\\
\int_{-a_1}^{a_1}\frac{\zeta^{2m}(\zeta^2-b_{m}^2)}{\hat{R}_{b_{m}}(\zeta)}\d\zeta & \int_{-a_1}^{a_1}\frac{\zeta^{2(m-1)}(b_{m}^2-\zeta^2)}{\hat{R}_{b_{m}}(\zeta)}\d\zeta &
\cdots &
\int_{-a_1}^{a_1}\frac{\zeta^2(b_{m}^2-\zeta^2)}{\hat{R}_{b_{m}}(\zeta)}\d\zeta
&\int_{-a_1}^{a_1}\frac{b_{m}^2-\zeta^2}{\hat{R}_{b_{m}}(\zeta)}\d\zeta
\end{pmatrix}_{(m+1)\times(m+1)}\\
&\times \det\begin{pmatrix}
\int_{b_{m-1}}^{a_{m}}\frac{\zeta^{2(m-1)}}{\hat{R}_{b_{m}}(\zeta)}\d\zeta & \int_{b_{m-1}}^{a_{m}}\frac{\zeta^{2(m-2)}}{\hat{R}_{b_{m}}(\zeta)}\d\zeta &
\cdots &
\int_{b_{m-1}}^{a_{m}}\frac{\zeta^2}{\hat{R}_{b_{m}}(\zeta)}\d\zeta
&\int_{b_{m-1}}^{a_{m}}\frac{1}{\hat{R}_{b_{m}}(\zeta)}\d\zeta
\\
\int_{b_{m-2}}^{a_{m-1}}\frac{\zeta^{2(m-1)}}{\hat{R}_{b_{m}}(\zeta)}\d\zeta & \int_{b_{m-2}}^{a_{m-1}}\frac{\zeta^{2(m-2)}}{\hat{R}_{b_{m}}(\zeta)}\d\zeta &
\cdots &
\int_{b_{m-2}}^{a_{m-1}}\frac{\zeta^2}{\hat{R}_{b_{m}}(\zeta)}\d\zeta
&\int_{b_{m-2}}^{a_{m-1}}\frac{1}{\hat{R}_{b_{m}}(\zeta)}\d\zeta
\\
\cdots&\cdots&\cdots&\cdots\\
\int_{-a_1}^{a_1}\frac{\zeta^{2(m-1)}}{\hat{R}_{b_{m}}(\zeta)}\d\zeta & \int_{-a_1}^{a_1}\frac{\zeta^{2(m-2)}}{\hat{R}_{b_{m}}(\zeta)}\d\zeta &
\cdots &
\int_{-a_1}^{a_1}\frac{\zeta^2}{\hat{R}_{b_{m}}(\zeta)}\d\zeta
&\int_{-a_1}^{a_1}\frac{1}{\hat{R}_{b_{m}}(\zeta)}\d\zeta
\end{pmatrix}_{m\times m}\\
&>0.
\end{align*}
The last equality can be verified by expanding the first row of the $(m+1)\times(m+1)$ determinant. For example, corresponding to the term $\int_{b_{m}}^{a_{m+1}}\frac{\zeta^{2(m-2)}(b_{m}^2-\zeta^2)}{\hat{R}_{b_{m}}(\zeta)}\d\zeta$, using the properties of determinant operations and Laplace's theorem, the following equations hold.
\begin{align*}
&\det\begin{pmatrix}
\int_{b_{m-1}}^{a_{m}}\frac{\zeta^{2m}(\zeta^2-b_{m}^2)}{\hat{R}_{b_{m}}(\zeta)}\d\zeta & \int_{b_{m-1}}^{a_{m}}\frac{\zeta^{2(m-3)}(\zeta^2-b_{m}^2)}{\hat{R}_{b_{m}}(\zeta)}\d\zeta &
\cdots &
\int_{b_{m-1}}^{a_{m}}\frac{(\zeta^2-b_{m}^2)}{\hat{R}_{b_{m}}(\zeta)}\d\zeta
&\int_{b_{m-1}}^{a_{m}}\frac{1}{\hat{R}_{b_{m}}(\zeta)}\d\zeta\\
\int_{b_{m-2}}^{a_{m-1}}\frac{\zeta^{2m}(\zeta^2-b_{m}^2)}{\hat{R}_{b_{m}}(\zeta)}\d\zeta & \int_{b_{m-2}}^{a_{m-1}}\frac{\zeta^{2(m-3)}(\zeta^2-b_{m}^2)}{\hat{R}_{b_{m}}(\zeta)}\d\zeta &
\cdots &
\int_{b_{m-2}}^{a_{m-1}}\frac{(\zeta^2-b_{m}^2)}{\hat{R}_{b_{m}}(\zeta)}\d\zeta
&\int_{b_{m-2}}^{a_{m-1}}\frac{1}{\hat{R}_{b_{m}}(\zeta)}\d\zeta\\
\cdots&\cdots&\cdots&\cdots\\
\int_{-a_1}^{a_1}\frac{\zeta^{2m}(\zeta^2-b_{m}^2)}{\hat{R}_{b_{m}}(\zeta)}\d\zeta & \int_{-a_1}^{a_1}\frac{\zeta^{2(m-3)}(\zeta^2-b_{m}^2)}{\hat{R}_{b_{m}}(\zeta)}\d\zeta &
\cdots &
\int_{-a_1}^{a_1}\frac{(\zeta^2-b_{m}^2)}{\hat{R}_{b_{m}}(\zeta)}\d\zeta
&\int_{-a_1}^{a_1}\frac{1}{\hat{R}_{b_{m}}(\zeta)}\d\zeta
\end{pmatrix}_{m\times m}\\
&\times \det G_{b_m}\\
&-\det\begin{pmatrix}
\int_{b_{m-1}}^{a_{m}}\frac{\zeta^{2(m-1)}(\zeta^2-b_{m}^2)}{\hat{R}_{b_{m}}(\zeta)}\d\zeta & \int_{b_{m-1}}^{a_{m}}\frac{\zeta^{2(m-3)}(\zeta^2-b_{m}^2)}{\hat{R}_{b_{m}}(\zeta)}\d\zeta &
\cdots &
\int_{b_{m-1}}^{a_{m}}\frac{(\zeta^2-b_{m}^2)}{\hat{R}_{b_{m}}(\zeta)}\d\zeta
&\int_{b_{m-1}}^{a_{m}}\frac{1}{\hat{R}_{b_{m}}(\zeta)}\d\zeta\\
\int_{b_{m-2}}^{a_{m-1}}\frac{\zeta^{2(m-1)}(\zeta^2-b_{m}^2)}{\hat{R}_{b_{m}}(\zeta)}\d\zeta & \int_{b_{m-2}}^{a_{m-1}}\frac{\zeta^{2(m-3)}(\zeta^2-b_{m}^2)}{\hat{R}_{b_{m}}(\zeta)}\d\zeta &
\cdots &
\int_{b_{m-2}}^{a_{m-1}}\frac{(\zeta^2-b_{m}^2)}{\hat{R}_{b_{m}}(\zeta)}\d\zeta
&\int_{b_{m-2}}^{a_{m-1}}\frac{1}{\hat{R}_{b_{m}}(\zeta)}\d\zeta\\
\cdots&\cdots&\cdots&\cdots\\
\int_{-a_1}^{a_1}\frac{\zeta^{2(m-1)}(\zeta^2-b_{m}^2)}{\hat{R}_{b_{m}}(\zeta)}\d\zeta & \int_{-a_1}^{a_1}\frac{\zeta^{2(m-3)}(\zeta^2-b_{m}^2)}{\hat{R}_{b_{m}}(\zeta)}\d\zeta &
\cdots &
\int_{-a_1}^{a_1}\frac{(\zeta^2-b_{m}^2)}{\hat{R}_{b_{m}}(\zeta)}\d\zeta
&\int_{-a_1}^{a_1}\frac{1}{\hat{R}_{b_{m}}(\zeta)}\d\zeta
\end{pmatrix}_{m\times m}\\
&\times \det H_{b_m}\\
&=\det\begin{pmatrix}
\int_{b_{m-1}}^{a_{m}}\frac{1}{\hat{R}_{b_{m}}(\zeta)}\d\zeta
&\int_{b_{m-1}}^{a_{m}}\frac{\zeta^{2m}(b_{m}^2-\zeta^2)}{\hat{R}_{b_{m}}(\zeta)}\d\zeta & \int_{b_{m-1}}^{a_{m}}\frac{\zeta^{2(m-3)}(b_{m}^2-\zeta^2)}{\hat{R}_{b_{m}}(\zeta)}\d\zeta &
\cdots &
\int_{b_{m-1}}^{a_{m}}\frac{(b_{m}^2-\zeta^2)}{\hat{R}_{b_{m}}(\zeta)}\d\zeta
\\
\int_{b_{m-2}}^{a_{m-1}}\frac{1}{\hat{R}_{b_{m}}(\zeta)}\d\zeta
&\int_{b_{m-2}}^{a_{m-1}}\frac{\zeta^{2m}(b_{m}^2-\zeta^2)}{\hat{R}_{b_{m}}(\zeta)}\d\zeta & \int_{b_{m-2}}^{a_{m-1}}\frac{\zeta^{2(m-3)}(b_{m}^2-\zeta^2)}{\hat{R}_{b_{m}}(\zeta)}\d\zeta &
\cdots &
\int_{b_{m-2}}^{a_{m-1}}\frac{(b_{m}^2-\zeta^2)}{\hat{R}_{b_{m}}(\zeta)}\d\zeta
\\
\cdots&\cdots&\cdots&\cdots\\
\int_{-a_1}^{a_1}\frac{1}{\hat{R}_{b_{m}}(\zeta)}\d\zeta
&\int_{-a_1}^{a_1}\frac{\zeta^{2m}(b_{m}^2-\zeta^2)}{\hat{R}_{b_{m}}(\zeta)}\d\zeta & \int_{-a_1}^{a_1}\frac{\zeta^{2(m-3)}(b_{m}^2-\zeta^2)}{\hat{R}_{b_{m}}(\zeta)}\d\zeta &
\cdots &
\int_{-a_1}^{a_1}\frac{(b_{m}^2-\zeta^2)}{\hat{R}_{b_{m}}(\zeta)}\d\zeta
\end{pmatrix}_{m\times m}\\
&\times \det G_{b_m}\\
&-\det\begin{pmatrix}
\int_{b_{m-1}}^{a_{m}}\frac{1}{\hat{R}_{b_{m}}(\zeta)}\d\zeta
&\int_{b_{m-1}}^{a_{m}}\frac{\zeta^{2(m-1)}(b_{m}^2-\zeta^2)}{\hat{R}_{b_{m}}(\zeta)}\d\zeta & \int_{b_{m-1}}^{a_{m}}\frac{\zeta^{2(m-3)}(b_{m}^2-\zeta^2)}{\hat{R}_{b_{m}}(\zeta)}\d\zeta &
\cdots &
\int_{b_{m-1}}^{a_{m}}\frac{(b_{m}^2-\zeta^2)}{\hat{R}_{b_{m}}(\zeta)}\d\zeta
\\
\int_{b_{m-2}}^{a_{m-1}}\frac{1}{\hat{R}_{b_{m}}(\zeta)}\d\zeta
&\int_{b_{m-2}}^{a_{m-1}}\frac{\zeta^{2(m-1)}(b_{m}^2-\zeta^2)}{\hat{R}_{b_{m}}(\zeta)}\d\zeta & \int_{b_{m-2}}^{a_{m-1}}\frac{\zeta^{2(m-3)}(b_{m}^2-\zeta^2)}{\hat{R}_{b_{m}}(\zeta)}\d\zeta &
\cdots &
\int_{b_{m-2}}^{a_{m-1}}\frac{(b_{m}^2-\zeta^2)}{\hat{R}_{b_{m}}(\zeta)}\d\zeta
\\
\cdots&\cdots&\cdots&\cdots\\
\int_{-a_1}^{a_1}\frac{1}{\hat{R}_{b_{m}}(\zeta)}\d\zeta
&\int_{-a_1}^{a_1}\frac{\zeta^{2(m-1)}(b_{m}^2-\zeta^2)}{\hat{R}_{b_{m}}(\zeta)}\d\zeta & \int_{-a_1}^{a_1}\frac{\zeta^{2(m-3)}(b_{m}^2-\zeta^2)}{\hat{R}_{b_{m}}(\zeta)}\d\zeta &
\cdots &
\int_{-a_1}^{a_1}\frac{(b_{m}^2-\zeta^2)}{\hat{R}_{b_{m}}(\zeta)}\d\zeta
\end{pmatrix}_{m\times m}\\
&\times \det H_{b_m}\\
&=\det\begin{pmatrix}
\int_{b_{m-1}}^{a_{m}}\frac{1}{\hat{R}_{b_{m}}(\zeta)}\d\zeta
&\int_{b_{m-1}}^{a_{m}}\frac{\zeta^{2(m-2)}(b_{m}^2-\zeta^2)}{\hat{R}_{b_{m}}(\zeta)}\d\zeta & \int_{b_{m-1}}^{a_{m}}\frac{\zeta^{2(m-3)}(b_{m}^2-\zeta^2)}{\hat{R}_{b_{m}}(\zeta)}\d\zeta &
\cdots &
\int_{b_{m-1}}^{a_{m}}\frac{(b_{m}^2-\zeta^2)}{\hat{R}_{b_{m}}(\zeta)}\d\zeta
\\
\int_{b_{m-2}}^{a_{m-1}}\frac{1}{\hat{R}_{b_{m}}(\zeta)}\d\zeta
&\int_{b_{m-2}}^{a_{m-1}}\frac{\zeta^{2(m-2)}(b_{m}^2-\zeta^2)}{\hat{R}_{b_{m}}(\zeta)}\d\zeta & \int_{b_{m-2}}^{a_{m-1}}\frac{\zeta^{2(m-3)}(b_{m}^2-\zeta^2)}{\hat{R}_{b_{m}}(\zeta)}\d\zeta &
\cdots &
\int_{b_{m-2}}^{a_{m-1}}\frac{(b_{m}^2-\zeta^2)}{\hat{R}_{b_{m}}(\zeta)}\d\zeta
\\
\cdots&\cdots&\cdots&\cdots\\
\int_{-a_1}^{a_1}\frac{1}{\hat{R}_{b_{m}}(\zeta)}\d\zeta
&\int_{-a_1}^{a_1}\frac{\zeta^{2(m-2)}(b_{m}^2-\zeta^2)}{\hat{R}_{b_{m}}(\zeta)}\d\zeta & \int_{-a_1}^{a_1}\frac{\zeta^{2(m-3)}(b_{m}^2-\zeta^2)}{\hat{R}_{b_{m}}(\zeta)}\d\zeta &
\cdots &
\int_{-a_1}^{a_1}\frac{(b_{m}^2-\zeta^2)}{\hat{R}_{b_{m}}(\zeta)}\d\zeta
\end{pmatrix}_{m\times m}\\
&\times \det\begin{pmatrix}
\int_{b_{m-1}}^{a_{m}}\frac{\zeta^{2m}(\zeta^2-b_{m}^2)}{\hat{R}_{b_{m}}(\zeta)}\d\zeta
&\int_{b_{m-1}}^{a_{m}}\frac{\zeta^{2(m-1)}(b_{m}^2-\zeta^2)}{\hat{R}_{b_{m}}(\zeta)}\d\zeta & \int_{b_{m-1}}^{a_{m}}\frac{\zeta^{2(m-3)}(b_{m}^2-\zeta^2)}{\hat{R}_{b_{m}}(\zeta)}\d\zeta &
\cdots &
\int_{b_{m-1}}^{a_{m}}\frac{(b_{m}^2-\zeta^2)}{\hat{R}_{b_{m}}(\zeta)}\d\zeta
\\
\int_{b_{m-2}}^{a_{m-1}}\frac{\zeta^{2m}(\zeta^2-b_{m}^2)}{\hat{R}_{b_{m}}(\zeta)}\d\zeta
&\int_{b_{m-2}}^{a_{m-1}}\frac{\zeta^{2(m-1)}(b_{m}^2-\zeta^2)}{\hat{R}_{b_{m}}(\zeta)}\d\zeta & \int_{b_{m-2}}^{a_{m-1}}\frac{\zeta^{2(m-3)}(b_{m}^2-\zeta^2)}{\hat{R}_{b_{m}}(\zeta)}\d\zeta &
\cdots &
\int_{b_{m-2}}^{a_{m-1}}\frac{(b_{m}^2-\zeta^2)}{\hat{R}_{b_{m}}(\zeta)}\d\zeta
\\
\cdots&\cdots&\cdots&\cdots\\
\int_{-a_1}^{a_1}\frac{\zeta^{2m}(\zeta^2-b_{m}^2)}{\hat{R}_{b_{m}}(\zeta)}\d\zeta
&\int_{-a_1}^{a_1}\frac{\zeta^{2(m-1)}(b_{m}^2-\zeta^2)}{\hat{R}_{b_{m}}(\zeta)}\d\zeta & \int_{-a_1}^{a_1}\frac{\zeta^{2(m-3)}(b_{m}^2-\zeta^2)}{\hat{R}_{b_{m}}(\zeta)}\d\zeta &
\cdots &
\int_{-a_1}^{a_1}\frac{(b_{m}^2-\zeta^2)}{\hat{R}_{b_{m}}(\zeta)}\d\zeta
\end{pmatrix}_{m\times m}\\
\end{align*}
\begin{align*}
&=\det\begin{pmatrix}
\int_{b_{m-1}}^{a_{m}}\frac{\zeta^{2(m-1)}}{\hat{R}_{b_{m}}(\zeta)}\d\zeta & \int_{b_{m-1}}^{a_{m}}\frac{\zeta^{2(m-2)}}{\hat{R}_{b_{m}}(\zeta)}\d\zeta &
\cdots &
\int_{b_{m-1}}^{a_{m}}\frac{\zeta^2}{\hat{R}_{b_{m}}(\zeta)}\d\zeta
&\int_{b_{m-1}}^{a_{m}}\frac{1}{\hat{R}_{b_{m}}(\zeta)}\d\zeta
\\
\int_{b_{m-2}}^{a_{m-1}}\frac{\zeta^{2(m-1)}}{\hat{R}_{b_{m}}(\zeta)}\d\zeta & \int_{b_{m-2}}^{a_{m-1}}\frac{\zeta^{2(m-2)}}{\hat{R}_{b_{m}}(\zeta)}\d\zeta &
\cdots &
\int_{b_{m-2}}^{a_{m-1}}\frac{\zeta^2}{\hat{R}_{b_{m}}(\zeta)}\d\zeta
&\int_{b_{m-2}}^{a_{m-1}}\frac{1}{\hat{R}_{b_{m}}(\zeta)}\d\zeta
\\
\cdots&\cdots&\cdots&\cdots\\
\int_{-a_1}^{a_1}\frac{\zeta^{2(m-1)}}{\hat{R}_{b_{m}}(\zeta)}\d\zeta & \int_{-a_1}^{a_1}\frac{\zeta^{2(m-2)}}{\hat{R}_{b_{m}}(\zeta)}\d\zeta &
\cdots &
\int_{-a_1}^{a_1}\frac{\zeta^2}{\hat{R}_{b_{m}}(\zeta)}\d\zeta
&\int_{-a_1}^{a_1}\frac{1}{\hat{R}_{b_{m}}(\zeta)}\d\zeta
\end{pmatrix}_{m\times m}\\
&\times\det\begin{pmatrix}
\int_{b_{m-1}}^{a_{m}}\frac{\zeta^{2m}(\zeta^2-b_{m}^2)}{\hat{R}_{b_{m}}(\zeta)}\d\zeta
&\int_{b_{m-1}}^{a_{m}}\frac{\zeta^{2(m-1)}(b_{m}^2-\zeta^2)}{\hat{R}_{b_{m}}(\zeta)}\d\zeta & \int_{b_{m-1}}^{a_{m}}\frac{\zeta^{2(m-3)}(b_{m}^2-\zeta^2)}{\hat{R}_{b_{m}}(\zeta)}\d\zeta &
\cdots &
\int_{b_{m-1}}^{a_{m}}\frac{(b_{m}^2-\zeta^2)}{\hat{R}_{b_{m}}(\zeta)}\d\zeta
\\
\int_{b_{m-2}}^{a_{m-1}}\frac{\zeta^{2m}(\zeta^2-b_{m}^2)}{\hat{R}_{b_{m}}(\zeta)}\d\zeta
&\int_{b_{m-2}}^{a_{m-1}}\frac{\zeta^{2(m-1)}(b_{m}^2-\zeta^2)}{\hat{R}_{b_{m}}(\zeta)}\d\zeta & \int_{b_{m-2}}^{a_{m-1}}\frac{\zeta^{2(m-3)}(b_{m}^2-\zeta^2)}{\hat{R}_{b_{m}}(\zeta)}\d\zeta &
\cdots &
\int_{b_{m-2}}^{a_{m-1}}\frac{(b_{m}^2-\zeta^2)}{\hat{R}_{b_{m}}(\zeta)}\d\zeta
\\
\cdots&\cdots&\cdots&\cdots\\
\int_{-a_1}^{a_1}\frac{\zeta^{2m}(\zeta^2-b_{m}^2)}{\hat{R}_{b_{m}}(\zeta)}\d\zeta
&\int_{-a_1}^{a_1}\frac{\zeta^{2(m-1)}(b_{m}^2-\zeta^2)}{\hat{R}_{b_{m}}(\zeta)}\d\zeta & \int_{-a_1}^{a_1}\frac{\zeta^{2(m-3)}(b_{m}^2-\zeta^2)}{\hat{R}_{b_{m}}(\zeta)}\d\zeta &
\cdots &
\int_{-a_1}^{a_1}\frac{(b_{m}^2-\zeta^2)}{\hat{R}_{b_{m}}(\zeta)}\d\zeta
\end{pmatrix}_{m\times m}.
\end{align*}
The second equality follows from applying Laplace's theorem to expand the first two columns of the $m\times m$ determinant.
\end{proof}
 Using the expression for the derivative of the function $g_{{\beta_{m}}}(k)$, the function $g_{{\beta_{m}}}(k)$ can be precisely expressed as
\begin{equation}
\begin{array}{l}
g_{{\beta_{m}}}(k)=xk+4k^3t-x \int_{i{\beta_{m}}}^{k}\frac{\zeta^{2m}+\sum\limits_{j=1}^{m} \tilde{c}_j \zeta^{2(m-j)}}{R_{{\beta_{m}}}(\zeta)}\d\zeta
-12t \int_{i{\beta_{m}}}^{k}\frac{\zeta^{2(m+1)}+\frac{1}{2}(a_m^2+\beta_{m}^2+\sum\limits_{j=1}^{m-1}(a_j^2+b_j^2)) \zeta^{2m}+\sum\limits_{j=1}^{m} \hat{c}_j \zeta^{2(m-j)}}{R_{{\beta_{m}}}(\zeta)}\d\zeta.
\end{array}
\end{equation}
Together with equations \eqref{g32}--\eqref{g34}, we obtain the integral constants $\Omega_{{\beta_{m}},0}$, $\Omega_{{\beta_{m}},j}$ and $\Omega_{{\beta_{m}},-j}$
\begin{equation}
\begin{array}{l}
\Omega_{\beta_{m},m-1}=-2x \int_{i{\beta_{m}}}^{ia_m}\frac{\zeta^{2m}+\sum\limits_{j=1}^{m} \tilde{c}_j \zeta^{2(m-j)}}{R_{{\beta_{m}+}}(\zeta)}\d\zeta
-24t \int_{i{\beta_{m}}}^{i a_m}\frac{\zeta^{2(m+1)}+\frac{1}{2}(a_m^2+\beta_{m}^2+\sum\limits_{j=1}^{m-1}(a_j^2+b_j^2)) \zeta^{2m}+\sum\limits_{j=1}^{m} \hat{c}_j \zeta^{2(m-j)}}{R_{{\beta_{m}+}}(\zeta)}\d\zeta,
\\
\Omega_{\beta_{m},j}=\Omega_{\beta_{m},m-1}+\Omega_{\beta_{m},m-2}+\cdots+\Omega_{\beta_{m},j+1}-2x \int_{i{b_{j+1}}}^{ia_{j+1}}\frac{\zeta^{2m}+\sum\limits_{j=1}^{m} \tilde{c}_j \zeta^{2(m-j)}}{R_{{\beta_{m}+}}(\zeta)}\d\zeta
\\ \ \ \ \ \ \ \ \ \ -24t \int_{i{b_{j+1}}}^{ia_{j+1}}\frac{\zeta^{2(m+1)}+\frac{1}{2}(a_m^2+\beta_{m}^2+\sum\limits_{j=1}^{m-1}(a_j^2+b_j^2)) \zeta^{2m}+\sum\limits_{j=1}^{m} \hat{c}_j \zeta^{2(m-j)}}{R_{{\beta_{m}+}}(\zeta)}\d\zeta,\ j=0,1,\ldots, m-2,
\\
\Omega_{\beta_{m},-j}=\Omega_{\beta_{m},j}, \ j=1,2,\ldots,m-1.
\end{array}
\end{equation}
Then we choose the function $f_{{\beta_{m}}}(k)$ to simplify the jumps on $\Sigma_j$, $\Sigma_{m,\beta_{m}}$, $\Sigma_{-j}$ and $\Sigma_{-m,\beta_{m}}$
\begin{align}
&f_{\beta_{m}+}(k) f_{\beta_{m}-}(k)=\frac{1}{2r_j(k)},\ \ &k\in\Sigma_j,\ j=1,2,\ldots,m-1,\\
&f_{\beta_{m}+}(k) f_{\beta_{m}-}(k)=\frac{1}{2r_m(k)},\ \ &k\in\Sigma_{m,\beta_{m}},\\
&\frac{f_{\beta_{m}+}(k)}{f_{\beta_{m}-}(k)}=e^{i\Delta_{\beta_{m},j}},\ \ &k\in i[b_j,a_{j+1}],\ j=1,2,\ldots,m-1,\\
&\frac{f_{\beta_{m}+}(k)}{f_{\beta_{m}-}(k)}=e^{i\Delta_{\beta_{m},0}},\ \ &k\in i[-a_1,a_1],\\
&\frac{f_{\beta_{m}+}(k)}{f_{\beta_{m}-}(k)}=e^{i\Delta_{\beta_{m},-j}},\ \ &k\in i[-a_{j+1},-b_j],\ j=1,2,\ldots,m-1,\\
&f_{\beta_{m}+}(k) f_{\beta_{m}-}(k)=2r_j(k),\ \ &k\in\Sigma_{-j},\ j=1,2,\ldots,m-1,\\
&f_{\beta_{m}+}(k) f_{\beta_{m}-}(k)=2r_m(k),\ \ &k\in\Sigma_{-m,\beta_{m}},\\
&f_{\beta_{m}}(k)=1+\mathcal{O}(\frac{1}{k}),\ \ &k \rightarrow  \infty.\label{ftconstant}
\end{align}
It is easy to see that the function $f_{{\beta_{m}}}(k)$ is given by
\begin{equation}\label{ft}
\begin{aligned}
\begin{array}{l}
f_{\beta_{m}}(k)=\exp\le( \frac{R_{\beta_{m}}(k)}{2\pi i}\left(\sum\limits_{j=1}^{m-1}\int_{\Sigma_j}\frac{\log(\frac{1}{2r_j(\zeta)}) }{R_{\beta_{m}+}(\zeta)(\zeta-k)}\d\zeta+
\sum\limits_{j=1}^{m-1}\int_{ ib_j}^{ia_{j+1}}\frac{i\Delta_{\beta_{m},j}}{R_{\beta_{m}}(\zeta)(\zeta-k)}\d\zeta
+\int_{-ia_1}^{ia_1}\frac{i\Delta_{\beta_{m},0}}{R_{\beta_{m}}(\zeta)(\zeta-k)}\d\zeta\right.\right.\\
\left.\left.\ \ \ \ \ \ \ \ \ \ \
+\sum\limits_{j=1}^{m-1}\int_{\Sigma_{-j}}\frac{\log(2r_j(\zeta))}{R_{\beta_{m}+}(\zeta)(\zeta-k)}\d\zeta+
\sum\limits_{j=1}^{m-1}\int_{-ia_{j+1}}^{-ib_{j}}\frac{i\Delta_{\beta_{m},-j}}{R_{\beta_{m}}(\zeta)(\zeta-k)}\d\zeta
+\int_{\Sigma_{m,\beta_{m}}}\frac{\log(\frac{1}{2r_m(\zeta)}) }{R_{\beta_{m}+}(\zeta)(\zeta-k)}\d\zeta
\right.\right.\\
\left.\left.\ \ \ \ \ \ \ \ \ \ \
+\int_{\Sigma_{-m,\beta_{m}}}\frac{\log(2r_m(\zeta))}{R_{\beta_{m}+}(\zeta)(\zeta-k)}\d\zeta
\ri) \ri).
\end{array}
\end{aligned}
\end{equation}
The normalization condition \eqref{ftconstant} determines the constants $\Delta_{\beta_{m},0}$, $\Delta_{\beta_{m},j}$ and $\Delta_{\beta_{m},-j}$ by the following system
\begin{equation}\label{tDelta}
\begin{array}{l}
\sum\limits_{j=1}^{m-1}\int_{\Sigma_j}\frac{\log(\frac{1}{2r_j(\zeta)}) \zeta^l}{R_{\beta_{m}+}(\zeta)}\d\zeta+
\sum\limits_{j=1}^{m-1}\int_{ ib_j}^{ia_{j+1}}\frac{i\Delta_{\beta_{m},j}\zeta^l}{R_{\beta_{m}}(\zeta)}\d\zeta
+\int_{-ia_1}^{ia_1}\frac{i\Delta_{\beta_{m},0}\zeta^l}{R(\zeta)}\d\zeta
+\sum\limits_{j=1}^{m-1}\int_{\Sigma_{-j}}\frac{\log(2r_j(\zeta))\zeta^l}{R_{\beta_{m}+}(\zeta)}\d\zeta\\+
\sum\limits_{j=1}^{m-1}\int_{-ia_{j+1}}^{-ib_{j}}\frac{i\Delta_{\beta_{m},-j}\zeta^l}{R_{\beta_{m}}(\zeta)}\d\zeta
+\int_{\Sigma_{m,\beta_{m}}}\frac{\log(\frac{1}{2r_m(\zeta)})\zeta^{l} }{R_{\beta_{m}+}(\zeta)}\d\zeta
+\int_{\Sigma_{-m,\beta_{m}}}\frac{\log(2r_m(\zeta))\zeta^{l}}{R_{\beta_{m}+}(\zeta)}\d\zeta=0,\ l=0,1,\ldots,2m-2.
\end{array}
\end{equation}
Next, we introduce a new matrix-valued function
\begin{align}\label{transform}
\tilde{T}(k)=X(k)e^{ig_{{\beta_{m}}}(k)\sigma_3}{f_{{\beta_{m}}}(k)}^{\sigma_3}.
\end{align}
By using the functions $g_{{\beta_{m}}}(k)$ and  $f_{{\beta_{m}}}(k)$, the transformation \eqref{transform} results in a new RH problem for $\tilde{T}(k)$ as follows:
\begin{RHP} Find a $2\times2$ matrix-valued function $\tilde{T}(x,t;k)$ with the following properties
\begin{enumerate}
\item $\tilde{T}(x,t;k)$ is analytic for $k\in \C\backslash (i[-b_m, b_m]\cup\bigcup\limits_{j=m+1}^{n}(\Sigma_{j}\cup\Sigma_{-j}))$.
\item For $k\in i[-b_m, b_m]\cup\bigcup\limits_{j=m+1}^{n}(\Sigma_{j}\cup\Sigma_{-j})$, the boundary values $\tilde{T}_{\pm}(x,t;k)=\tilde{T}(x,t;k\mp 0)$ satisfy the following jump conditions
\begin{align}
\tilde{T}_+(k)=\tilde{T}_-(k) V_{\tilde{T}}(k)
\end{align}
where
\begin{equation}
\begin{array}{l}
V_{\tilde{T}}(k)= \begin{cases}
\displaystyle \begin{pmatrix} 1 & 0\\ 2i r_j(k) f_{{\beta_{m}}}^2(k) e^{2i(g_{\beta_{m}(k)}-\theta(x,t;k))} &1 \end{pmatrix},  k \in \Sigma_{j},\ j=m+1,m+2,\ldots,n,\\[3ex]
\displaystyle \begin{pmatrix} 1 & 0\\ 2i r_m(k) f_{{\beta_{m}}}^2(k) e^{2i(g_{\beta_{m}}(k)-\theta(x,t;k))} &1 \end{pmatrix},  k \in  i({\beta_{m}}, b_m],\\[3ex]
			\displaystyle \begin{pmatrix} \frac{e^{i\le(g_{{\beta_{m}}+}(k) - g_{{\beta_{m}}-}(k)\ri)}f_{{\beta_{m}}+}(k)}{f_{{\beta_{m}}-}(k)} & 0\\ i & \frac{e^{-i\le(g_{{\beta_{m}}+}(k) - g_{{\beta_{m}}-}(k)\ri)}f_{{\beta_{m}}-}(k)}{f_{{\beta_{m}}+}(k)} \end{pmatrix}, k \in  (\bigcup\limits_{j=1}^{m-1} \Sigma_j)\cup \Sigma_{m,{\beta_{m}}},\\
			\displaystyle \begin{pmatrix} \frac{e^{i\le(g_{{\beta_{m}}+}(k) - g_{{\beta_{m}}-}(k)\ri)}f_{{\beta_{m}}+}(k)}{f_{{\beta_{m}}-}(k)}& i\\ 0 & \frac{e^{-i\le(g_{{\beta_{m}}+}(k) - g_{{\beta_{m}}-}(k)\ri)}f_{{\beta_{m}}-}(k)}{f_{{\beta_{m}}+}(k)} \end{pmatrix}, k \in  (\bigcup\limits_{j=1}^{m-1} \Sigma_{-j})\cup\Sigma_{-m,{\beta_{m}}},\\
\displaystyle \begin{pmatrix} 1 &  2i r_j(k) f_{{\beta_{m}}}^{-2}(k) e^{-2i(g_{\beta_{m}}(k)-\theta(x,t;k))}\\0 &1 \end{pmatrix},  k \in  \Sigma_{-j},\ j=m+1,m+2,\ldots,n,\\
\displaystyle \begin{pmatrix} 1 &  2i r_m(k) f_{{\beta_{m}}}^{-2}(k) e^{-2i(g_{\beta_{m}}(k)-\theta(x,t;k))}\\0 &1 \end{pmatrix},  k \in  i[-b_m, -{\beta_{m}}),\\
			\displaystyle \begin{pmatrix} e^{i(\Omega_{{\beta_{m}},j}+\Delta_{{\beta_{m}},j}) } & 0\\0& e^{-i(\Omega_{{\beta_{m}},j}+\Delta_{{\beta_{m}},j})} \end{pmatrix},  k \in  i[b_j,a_{j+1}],\ j=1,2,\ldots,m-1,\\
\displaystyle \begin{pmatrix} e^{i(\Omega_{{\beta_{m}},0}+\Delta_{{\beta_{m}},0}) } & 0\\0& e^{-i(\Omega_{{\beta_{m}},0}+\Delta_{{\beta_{m}},0})} \end{pmatrix},  k \in  i[-a_1,a_1],\\
\displaystyle \begin{pmatrix} e^{i(\Omega_{{\beta_{m}},-j}+\Delta_{{\beta_{m}},-j}) } & 0\\0& e^{-i(\Omega_{{\beta_{m}},-j}+\Delta_{{\beta_{m}},-j})} \end{pmatrix},  k \in  i[-a_{j+1},-b_j]\ j=1,2,\ldots,m-1.
\end{cases}
\end{array}
\end{equation}
\item
$\tilde{T}(k) =  I  + \mathcal{O}\le(\frac{1}{k}\ri), \qquad k \rightarrow \infty.$
\end{enumerate}
\end{RHP}
The jump matrix  on $ \Sigma_{j},\ j=1,2,\ldots,m-1$, and $ \Sigma_{m,\beta_{m}}$ can be decomposed as follows
\begin{align*}
&\begin{pmatrix} \dfrac{e^{i(g_{\beta_{m}+}(k)-g_{\beta_{m}-}(k))}}{2r_j(k)f_{\beta_{m}-}(k)^{2}} & 0 \\ i  &  \dfrac{e^{-i(g_{\beta_{m}+}(k)-g_{\beta_{m}-}(k))}}{2r_j(k)f_{\beta_{m}+}(k)^{2}} \end{pmatrix}\\=
&\begin{pmatrix} 1 & \frac{-ie^{-2i(g_{\beta_{m}-}(k)-\theta(x,t;k))}}{ 2r_j(k) f_{\beta_{m}-}^{2}(k)}  \\ 0 & 1 \end{pmatrix}
	   \begin{pmatrix} 0 & i \\ i & 0 \end{pmatrix}
	   \begin{pmatrix} 1 &   \frac{-i e^{-2i(g_{\beta_{m}+}(k)-\theta(x,t;k))}}{2r_j(k) f_{\beta_{m}+}^{2}(k)}   \\ 0 & 1 \end{pmatrix},
\end{align*}
and on $\Sigma_{-j},\ j=1,2,\ldots,m-1$, $\Sigma_{-m,\beta_{m}}$ as
\begin{align*}
&\begin{pmatrix} \dfrac{f_{\beta_{m}+}(k)^{2} e^{i(g_{\beta_{m}+}(k)-g_{\beta_{m}-}(k))}}{2r_j(k)} & i \\ 0  &  \dfrac{f_{\beta_{m}-}(k)^{2} e^{-i(g_{\beta_{m}+}(k)-g_{\beta_{m}-}(k))}}{2r_j(k)} \end{pmatrix}\\=
&\begin{pmatrix} 1 & 0\\ \frac{-i  f_{\beta_{m}-}^{2}(k) e^{2i(g_{\beta_{m}-}(k)-\theta(x,t;k))}}{2r_j(k)} & 1 \end{pmatrix}
	   \begin{pmatrix} 0 & i \\ i & 0 \end{pmatrix}
	   \begin{pmatrix} 1 & 0\\ \frac{-if_{\beta_{m}+}^{2}(k) e^{2i(g_{\beta_{m}+}(k)-\theta(x,t;k))}}{2r_j(k)}  & 1 \end{pmatrix}.
\end{align*}
These factorizations motivate us to open lenses around $\Sigma_{j}$, $\Sigma_{-j}$, $\Sigma_{m,\beta_{m}}$ and $\Sigma_{-m,\beta_{m}}$, then we define the following transformation
\begin{equation}\label{St}
	 \tilde{S}(k) =\tilde{T}(k)
	\begin{cases}
	   \begin{pmatrix} 1 & \frac{-ie^{-2i (g_{{\beta_{m}}}(k)-\theta(x,t;k))}}{2\hat{r}_j(k) f_{\beta_{m}}^{2}(k)} \\ 0 & 1 \end{pmatrix},
	& k\in \text{lens right of $\Sigma_{j},\ j=1,2,\ldots,m-1$}, \\
	   \begin{pmatrix} 1 & \frac{ie^{-2i (g_{{\beta_{m}}}(k)-\theta(x,t;k))}}{2\hat{r}_j(k) f_{\beta_{m}}^{2}(k)} \\ 0 & 1 \end{pmatrix},
	& k\in \text{lens left of $\Sigma_{j},\ j=1,2,\ldots,m-1$}, \\
\begin{pmatrix} 1 & \frac{-ie^{-2i (g_{{\beta_{m}}}(k)-\theta(x,t;k))}}{2\hat{r}_m(k) f_{\beta_{m}}^{2}(k)} \\ 0 & 1 \end{pmatrix},
	& k\in \text{lens right of $\Sigma_{m,{\beta_{m}}}$}, \\
	   \begin{pmatrix} 1 & \frac{ie^{-2i (g_{{\beta_{m}}}(k)-\theta(x,t;k))}}{2\hat{r}_m(k) f_{\beta_{m}}^{2}(k)} \\ 0 & 1 \end{pmatrix},
	& k\in \text{lens left of $\Sigma_{m,{\beta_{m}}}$}, \\
	   \begin{pmatrix} 1 & 0 \\  \frac{-i f_{{\beta_{m}}}^2(k) e^{2i (g_{{\beta_{m}}}(k)-\theta(x,t;k))}}{2\hat{r}_j(k)} & 1 \end{pmatrix},
	& k\in \text{lens right of $\Sigma_{-j},\ j=1,2,\ldots,m-1$}, \\
	   \begin{pmatrix} 1 & 0 \\  \frac{i f_{{\beta_{m}}}^2(k) e^{2i (g_{{\beta_{m}}}(k)-\theta(x,t;k))}}{2\hat{r}_j(k)} & 1 \end{pmatrix},
	& k\in \text{lens left of $\Sigma_{-j},\ j=1,2,\ldots,m-1$}, \\
	   \begin{pmatrix} 1 & 0 \\  \frac{-i f_{{\beta_{m}}}^2(k) e^{2i (g_{{\beta_{m}}}(k)-\theta(x,t;k))}}{2\hat{r}_m(k)} & 1 \end{pmatrix},
	& k\in \text{lens right of $\Sigma_{-m,{\beta_{m}}}$}, \\
	   \begin{pmatrix} 1 & 0 \\  \frac{i f_{{\beta_{m}}}^2(k) e^{2i (g_{{\beta_{m}}}(k)-\theta(x,t;k))}}{2\hat{r}_m(k)} & 1 \end{pmatrix},
	& k\in \text{lens left of $\Sigma_{-m,{\beta_{m}}}$}, \\
	   I , & \text{elsewhere .}
	\end{cases}
\end{equation}

As a consequence, the function $\tilde{S}(k)$ satisfies the following conditions:
\begin{enumerate}
\item{}
$\tilde{S}(k)$ is analytic for $k\in \C\backslash ( \bigcup\limits_{j=1}^{m-1}(\tilde{\mathcal{C}}_j\cup \tilde{\mathcal{C}}_{-j} )\cup \tilde{\mathcal{C}}_{m,{\beta_{m}}}\cup \tilde{\mathcal{C}}_{-m,{\beta_{m}}}\cup i[-b_m, b_m]\cup\bigcup\limits_{j=m+1}^{n}(\Sigma_{j}\cup\Sigma_{-j}) ) $, where the contours $\tilde{\mathcal{C}}_j$, $\tilde{\mathcal{C}}_{-j}$, $\tilde{\mathcal{C}}_{m,{\beta_{m}}}$ and $\tilde{\mathcal{C}}_{-m,{\beta_{m}}}$ are shown in Figure 3.
\item{} The boundary values $\tilde{S}_{\pm}(k)$ satisfy the jump conditions
\begin{align}
\tilde{S}_+(k)=\tilde{S}_-(k) V_{\tilde{S}},\ \ k\in  \bigcup\limits_{j=1}^{m-1}(\tilde{\mathcal{C}}_j\cup \tilde{\mathcal{C}}_{-j} )\cup \tilde{\mathcal{C}}_{m,{\beta_{m}}}\cup \tilde{\mathcal{C}}_{-m,{\beta_{m}}}\cup i[-b_m, b_m]\cup\bigcup\limits_{j=m+1}^{n}(\Sigma_{j}\cup\Sigma_{-j}),
\end{align}
where the jump matrix $V_{\tilde{S}}$ is given by
\begin{equation}
\begin{array}{l}
V_{\tilde{S}}(k)= \begin{cases}
\begin{pmatrix} 1 & 0 \\  2i\hat{r}_j(k)f^2_{{\beta_{m}}}(k)e^{2i(g_{{\beta_{m}}}(k)-\theta(x,t;k))} & 1 \end{pmatrix}, k\in \Sigma_{j},\ j=m+1,m+2,\ldots,n,\\
\begin{pmatrix} 1 & 0 \\  2i\hat{r}_m(k)f^2_{{\beta_{m}}}(k)e^{2i(g_{{\beta_{m}}}(k)-\theta(x,t;k))} & 1 \end{pmatrix}, k\in i(\beta_{m}, b_{m}],\\
\begin{pmatrix} 1 & -i (2\hat{r}_m(k))^{-1} f_{\beta_{m}}^{-2}(k) e^{-2i(g_{\beta_{m}}(k)-\theta(x,t;k))} \\ 0 & 1 \end{pmatrix},
	 k\in\tilde{\mathcal{C}}_{m,\beta_{m}},\\
\begin{pmatrix} 1 & -i (2\hat{r}_j(k))^{-1} f_{\beta_{m}}^{-2}(k) e^{-2i(g_{\beta_{m}}(k)-\theta(x,t;k))} \\ 0 & 1 \end{pmatrix},
	 k\in\tilde{\mathcal{C}}_j ,\ j=1,2,\ldots,m-1,\\
\begin{pmatrix} 1 & 0 \\  \frac{-i f_{\beta_{m}}^2(k) e^{2i(g_{\beta_{m}}(k)-\theta(x,t;k))}}{2\hat{r}_j(k)} & 1 \end{pmatrix},
	 k\in  \tilde{\mathcal{C}}_{-j},\ j=1,2,\ldots,m-1, \\
\begin{pmatrix} 1 & 0 \\  \frac{-i f_{\beta_{m}}^2(k) e^{2i(g_{\beta_{m}}(k)-\theta(x,t;k))}}{2\hat{r}_m(k)} & 1 \end{pmatrix},
	 k\in  \tilde{\mathcal{C}}_{-m,\beta_{m}},\\
\begin{pmatrix} 1 & 2ir_j(k)f^{-2}_{{\beta_{m}}}(k)e^{-2i(g_{{\beta_{m}}}(k)-\theta(x,t;k))}\\ 0 & 1 \end{pmatrix}, k\in \Sigma_{-j},\ j=m+1,m+2,\ldots,n,\\
\begin{pmatrix} 1 & 2ir_m(k)f^{-2}_{{\beta_{m}}}(k)e^{-2i(g_{{\beta_{m}}}(k)-\theta(x,t;k))}\\ 0 & 1 \end{pmatrix}, k\in i[-b_{m}, -\beta_{m}),\\
\begin{pmatrix} 0 & i \\ i & 0 \end{pmatrix},
	 k\in (\bigcup\limits_{j=1}^{m-1}\Sigma_j\cup \Sigma_{-j})\cup \Sigma_{m,\beta_{m}}\cup \Sigma_{-m,\beta_{m}},\\
\begin{pmatrix} e^{i({\Omega_{\beta_{m},0}}+\Delta_{\beta_{m},0})} & 0 \\ 0 & e^{-i({\Omega_{\beta_{m},0}}+\Delta_{\beta_{m},0})} \end{pmatrix},
 k\in i[-a_1,a_1],\\
\begin{pmatrix} e^{i(\Omega_{\beta_{m},j}+\Delta_{\beta_{m},j})} & 0 \\ 0 & e^{-i(\Omega_{\beta_{m},j}+\Delta_{\beta_{m},j})} \end{pmatrix},
 k\in i[b_j,a_{j+1}],\ j=1,2,\ldots,m-1,\\
\begin{pmatrix} e^{i(\Omega_{\beta_{m},-j}+\Delta_{\beta_{m},-j})} & 0 \\ 0 & e^{-i(\Omega_{\beta_{m},-j}+\Delta_{\beta_{m},-j})} \end{pmatrix},
 k\in i[-a_{j+1},-b_j],\ j=1,2,\ldots,m-1,\\
	\end{cases}
\end{array}
\end{equation}
\item{}$\tilde{S}(k)= I  + \mathcal{O}\le(\frac{1}{k}\ri), \qquad k \rightarrow \infty.$
\end{enumerate}
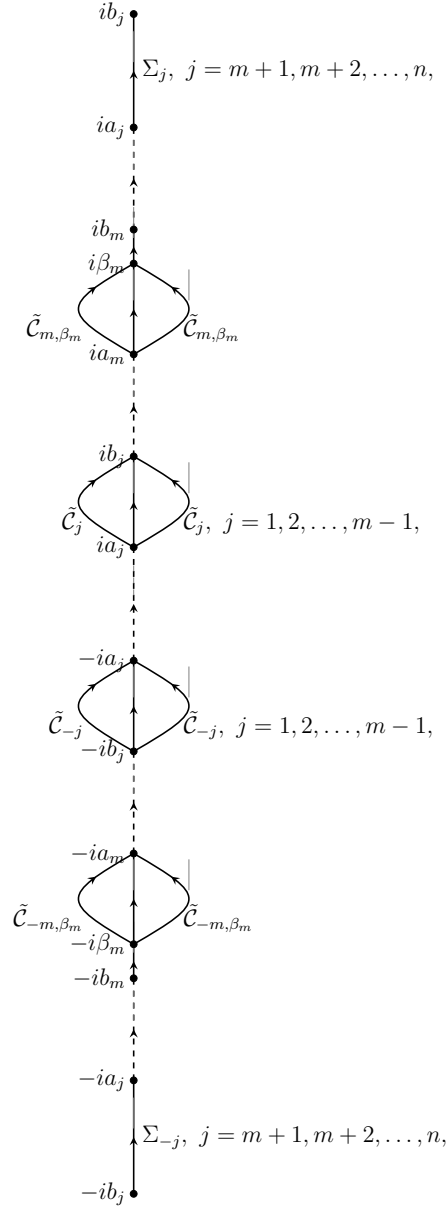
\begin{figure}
\centering
\scalebox{.75}{
\begin{tikzpicture}[>=stealth]op
\path (0,0) coordinate (O);

\coordinate (eta1) at (0,1);    \coordinate (eta1c) at ($-1*(eta1)$);
\coordinate (eta2) at (0,2.6);       \coordinate (eta2c) at ($-1*(eta2)$);
\coordinate (eta3) at (0,4.4);   \coordinate (eta3c) at ($-1*(eta3)$);
\coordinate (alpha) at (0,6);   \coordinate (alphac) at ($-1*(alpha)$);
\coordinate (eta4) at (0,6.6);   \coordinate (eta4c) at ($-1*(eta4)$);
\coordinate (eta5) at (0,8.4);   \coordinate (eta5c) at ($-1*(eta5)$);
\coordinate (eta6) at (0,10.4);   \coordinate (eta6c) at ($-1*(eta6)$);

\draw[->- = .5,  thick] (eta5)--(eta6)
 node[pos=.5, right] {$\Sigma_{j},\ j=m+1,m+2,\ldots,n,$}
  node[pos=0.5, pin={}] {};
\draw[->- = .5,dashed,  thick] (eta4)--(eta5)
  node[pos=0.5, pin={}] {};
\draw[->- = .5,thick] (alpha)--(eta4) node[pos=0.5, pin={}]{};
\draw[->- = .5,thick] (eta3)--(alpha)
  node[pos=0.5, pin={}] {};
\draw[->- = .5,dashed,  thick] (eta2)--(eta3)
  node[pos=0.5, pin={}] {};
\draw[->- = .5,  thick] (eta1)--(eta2)
  node[pos=0.5, pin={}] {};

\draw[->- = .5,  thick] (eta6c)--(eta5c)
node[pos=.5, right] {$\Sigma_{-j},\ j=m+1,m+2,\ldots,n,$}
  node[pos=0.5, pin={}] {};
\draw[->- = .5, dashed,  thick] (eta5c)--(eta4c)
  node[pos=0.5, pin={}] {};
\draw[->- = .5,thick] (eta4c)--(alphac)node[pos=0.5, pin={}]{};
\draw[->- = .5,  thick] (alphac)--(eta3c)
  node[pos=0.5, pin={}] {};
\draw[->- = .5, dashed,  thick] (eta3c)--(eta2c)
  node[pos=0.5, pin={}] {};
\draw[->- = .5,  thick] (eta2c)--(eta1c)
  node[pos=0.5, pin={}] {};

\draw[->- = .5, dashed,  thick] (eta1c)--(eta1)
  node[pos=0.5, pin={}] {};

\draw[->- = .7,thick] (eta1) .. controls + (30:1.5cm) and + (-30:1.5cm) .. (eta2)
  node[pos=.25, right] {$\tilde{\mathcal{C}}_j,\ j=1,2,\ldots,m-1,$}
  node[pos=0.5, pin={}] {};
\draw[->- = .7,thick] (eta1) .. controls + (150:1.5cm) and + (-150:1.5cm) .. (eta2)
  node[pos=.25, left] {$\tilde{\mathcal{C}}_j$};

\draw[->- = .7,thick] (eta3) .. controls + (30:1.5cm) and + (-30:1.5cm) .. (alpha)
  node[pos=.25, right] {$\tilde{\mathcal{C}}_{m,{\beta_{m}}}$}
  node[pos=0.5, pin={}] {};
\draw[->- = .7,thick] (eta3) .. controls + (150:1.5cm) and + (-150:1.5cm) .. (alpha)
  node[pos=.25, left] {$\tilde{\mathcal{C}}_{m,{\beta_{m}}}$};

\draw[->- = .7,thick] (eta2c) .. controls + (30:1.5cm) and + (-30:1.5cm) .. (eta1c)
  node[pos=.25, right] {$\tilde{\mathcal{C}}_{-j},\ j=1,2,\ldots,m-1,$}
  node[pos=0.5, pin={}] {};
\draw[->- = .7,thick] (eta2c) .. controls + (150:1.5cm) and + (-150:1.5cm) .. (eta1c)
  node[pos=.25, left] {$\tilde{\mathcal{C}}_{-j}$};

\draw[->- = .7,thick] (alphac) .. controls + (30:1.5cm) and + (-30:1.5cm) .. (eta3c)
  node[pos=.25, right] {$\tilde{\mathcal{C}}_{-m,{\beta_{m}}}$}
  node[pos=0.5, pin={}] {};
\draw[->- = .7,thick] (alphac) .. controls + (150:1.5cm) and + (-150:1.5cm) .. (eta3c)
  node[pos=.25, left] {$\tilde{\mathcal{C}}_{-m,{\beta_{m}}}$};

\draw[fill] (eta6) circle [radius=0.06] node[left] {$ib_j$};
\draw[fill] (eta5) circle [radius=0.06] node[left] {$ia_j$};
\draw[fill] (eta4) circle [radius=0.06] node[left] {$ib_m$};
\draw[fill] (alpha) circle [radius=0.06] node[left] {$i{\beta_{m}}$};
\draw[fill] (eta3) circle [radius=0.06] node[left] {$ia_m$};
\draw[fill] (eta2) circle [radius=0.06] node[left] {$ib_j$};
\draw[fill] (eta1) circle [radius=0.06] node[left] {$ia_j$};

\draw[fill] (eta2c) circle [radius=0.06] node[left] {$-ib_j$};
\draw[fill] (eta1c) circle [radius=0.06] node[left] {$-ia_j$};
\draw[fill] (eta3c) circle [radius=0.06] node[left] {$-ia_m$};
\draw[fill] (alphac) circle [radius=0.06] node[left] {$-i{\beta_{m}}$};
\draw[fill] (eta4c) circle [radius=0.06] node[left] {$-ib_m$};
\draw[fill] (eta6c) circle [radius=0.06] node[left] {$-ib_j$};
\draw[fill] (eta5c) circle [radius=0.06] node[left] {$-ia_j$};
\end{tikzpicture}
}

\label{openlenses2}
\caption{The jump contours $\tilde{\mathcal{C}}_j$, $\tilde{\mathcal{C}}_{-j}$, $ \tilde{\mathcal{C}}_{m,{\beta_{m}}}$ and $\tilde{\mathcal{C}}_{-m,{\beta_{m}}}$.}
\end{figure}
To eliminate the jumps on
$\bigcup\limits_{j=1}^{m-1}(\tilde{\mathcal{C}}_j\cup\tilde{\mathcal{C}}_{-j})\cup\tilde{\mathcal{C}}_{m,{\beta_{m}}}\cup\tilde{\mathcal{C}}_{-m,{\beta_{m}}}$, we need the following lemma.
\begin{lemma}\label{lemma5.2}
For $\xi_{2m-2}<\xi<\xi_{2m-1}$, the following inequalities are satisfied:
\begin{align}
&\Re 2i\le[g_{\beta_{m}}(k)-kx-4k^3t \ri]<-ct \ \mbox{ for }k \in i({\beta_{m}},b_m]\cup\bigcup\limits_{j=m+1}^{n}\Sigma_{j},\label{inequality1}\\
&\Re 2i\le[g_{\beta_{m}}(k)-kx-4k^{3}t \ri]>ct \ \mbox{ for }k \in \tilde{\mathcal{C}}_{m,{\beta_{m}}}\cup\bigcup\limits_{j=1}^{m-1}\tilde{\mathcal{C}}_j \backslash \{ia_1, ib_1, \ldots, ia_m, i{\beta_{m}} \}, \label{inequality2}\\
&\Re 2i\le[g_{\beta_{m}}(k)-kx-4k^3t \ri]>ct \ \mbox{ for }k \in i[-b_m,-{\beta_{m}})\cup\bigcup\limits_{j=m+1}^{n}\Sigma_{-j},\label{inequality3}\\
&\Re 2i\le[g_{\beta_{m}}(k)-kx-4k^{3}t \ri]<-ct \ \mbox{ for }k \in \tilde{\mathcal{C}}_{-m,{\beta_{m}}}\cup\bigcup\limits_{j=1}^{m-1}\tilde{\mathcal{C}}_{-j} \backslash \{-ia_1, -ib_1, \ldots, -ia_m, -i{\beta_{m}} \} \label{inequality4},
\end{align}
for some constant $c\in \mathbb{R}_{+}$.
\end{lemma}
\begin{proof}
Combining the representation of the derivative of $g_{{\beta_{m}}}(k)$ in \eqref{gad}, \eqref{tildec1}, \eqref{tildec2}, \eqref{hatc1} and \eqref{hatc2}, we obtain
\begin{align}
g_{{\beta_{m}}}'(k)-x-12k^2t=-\frac{12t(k^2+{\beta_{m}}^2)\prod\limits_{j=1}^{m}(k^2+k_j^2)}{R_{{\beta_{m}}}(k)},
\end{align}
where $k_1=k_1(x,t)\in (0,a_1)$ and $k_j=k_j(x,t)\in (b_{j-1},a_{j}),\ j=2,3,\ldots,m$. Hence, the following inequality is established for $k \in i({\beta_{m}},b_m]\cup\bigcup\limits_{j=m+1}^{n}\Sigma_{j}$
\begin{equation}
-12\frac{(k^2+{\beta_{m}}^2)\prod\limits_{j=1}^{m}(k^2+k_j^2)}{R_{{\beta_{m}}}(k)}>0,
\end{equation}
which implies \eqref{inequality1}.

For $k\in \Sigma_{m,{\beta_{m}}}\cup \bigcup\limits_{j=1}^{m-1}\Sigma_{j}$
\begin{equation}
\Im\le( -12\frac{(k^2+{\beta_{m}}^2)\prod\limits_{j=1}^{m}(k^2+k_j^2)}{R_{{\beta_{m}}+}(k)}\ri)>0,
\end{equation}
we get $\Im (g_{{\beta_{m}}}(k)-xk-4k^3t)<0$ on $\bigcup\limits_{j=1}^{m-1}\tilde{\mathcal{C}}_j\cup \tilde{\mathcal{C}}_{m,{\beta_{m}}}$, which implies \eqref{inequality2}. In a similar way, the inequalities \eqref{inequality3} and \eqref{inequality4} can be proven.
\end{proof}

As time tends to infinity, the jumps for $\tilde{S}(k)$ converges to the jumps for the following outer model problem $\tilde{S}^{\infty}(k)$  exponentially fast outside small neighbourhoods of endpoints $\pm ia_1$, $\pm i b_1$, $\pm ia_2$, $\ldots$, $\pm i{\beta_{m}}$.
\begin{RHP}\label{RHP6}
		Find a $2\times 2$ matrix-valued function $\tilde{S}^{\infty}(k)$ with the following properties
\begin{enumerate}
\item {}  $\tilde{S}^{\infty}(k)$ is analytic in $\mathbb{C}\setminus i[-{\beta_{m}},{\beta_{m}}]$.
\item {} For $k\in i[-{\beta_{m}},{\beta_{m}}]$, the boundary values $\tilde{S}^{\infty}_{\pm}(k)$ satisfy the following jump relation
\begin{gather}
\label{Stinfinity1}
\tilde{S}^{\infty}_+(k) = \tilde{S}^{\infty}_-(k)
\begin{cases}
\begin{pmatrix}0 & i\\ i & 0 \end{pmatrix}, k \in\bigcup\limits_{j=1}^{m-1}(\Sigma_{j}\cup\Sigma_{-j})\cup\Sigma_{m,{\beta_{m}}}\cup\Sigma_{-m,{\beta_{m}}},  \\
\begin{pmatrix} e^{i(\Omega_{{\beta_{m}},0}+\Delta_{{\beta_{m}},0})} &0 \\ 0 & e^{-i(\Omega_{{\beta_{m}},0}+\Delta_{{\beta_{m}},0})}\end{pmatrix},   k \in i[-a_1,a_1],  \\
\begin{pmatrix} e^{i(\Omega_{{\beta_{m}},j}+\Delta_{{\beta_{m}},j})} &0 \\ 0 & e^{-i(\Omega_{{\beta_{m}},j}+\Delta_{{\beta_{m}},j})}\end{pmatrix},   k \in i[b_j,a_{j+1}],\ j=1,2,\ldots,m-1,  \\
\begin{pmatrix} e^{i(\Omega_{{\beta_{m}},-j}+\Delta_{{\beta_{m}},-j})} &0 \\ 0 & e^{-i(\Omega_{{\beta_{m}},-j}+\Delta_{{\beta_{m}},-j})}\end{pmatrix},   k \in i[-a_{j+1},-b_{j}],\ j=1,2,\ldots,m-1.
\end{cases}
\end{gather}
\item {}$\tilde{S}^{\infty}(k)= I +\mathcal{O}\le(\frac{1}{k}\ri), k\to\infty.$
\end{enumerate}
\end{RHP}
Motivated by the conditions \eqref{Stinfinity1},  we introduce the function $\tilde{h}(k)$ as follows:
\begin{equation}
\begin{array}{l}
\tilde{h}(k)=\frac{R_{{\beta_{m}}}(k)}{2\pi i}(\sum\limits_{j=1}^{m-1}\int_{ ib_j}^{ia_{j+1}}\frac{i(\Omega_{{\beta_{m}},j}+\Delta_{{\beta_{m}},j})}{R_{{\beta_{m}}}(\zeta)(\zeta-k)}\d\zeta+\int_{ -ia_1}^{ia_1}\frac{i(\Omega_{{\beta_{m}},0}+\Delta_{{\beta_{m}},0})}{R_{{\beta_{m}}}(\zeta)(\zeta-k)}\d\zeta+
 \sum\limits_{j=1}^{m-1}\int_{-ia_{j+1}}^{-ib_j}\frac{i(\Omega_{{\beta_{m}},-j}+\Delta_{{\beta_{m}},-j})}{R_{{\beta_{m}}}(\zeta)(\zeta-k)}\d\zeta),
\end{array}
\end{equation}
it is easy to check that $\tilde{h}(k)$ satisfies the following jump relation
\begin{align}
&\tilde{h}_{+}(k)+\tilde{h}_{-}(k)=0,\ & k\in \bigcup\limits_{j=1}^{m-1}(\Sigma_{j}\cup\Sigma_{-j})\cup\Sigma_{m,{\beta_{m}}}\cup \Sigma_{-m,{\beta_{m}}},\\
&\tilde{h}_{+}(k)-\tilde{h}_{-}(k)=i(\Omega_{{\beta_{m}},0}+\Delta_{{\beta_{m}},0}),\ &k\in i[-a_1, a_1],\\
&\tilde{h}_{+}(k)-\tilde{h}_{-}(k)=i(\Omega_{{\beta_{m}},j}+\Delta_{{\beta_{m}},j}),\ &k\in i[b_j,a_{j+1}], j=1,2,\ldots,m-1,\\
&\tilde{h}_{+}(k)-\tilde{h}_{-}(k)=i(\Omega_{{\beta_{m}},-j}+\Delta_{{\beta_{m}},-j}),\ &k\in i[-a_{j+1},-b_j], j=1,2,\ldots,m-1.
\end{align}
Suppose that $R_{\beta_{m}}(k)$ has the following asymptotic expansion as k tends to infinity
\begin{gather}
R_{\beta_{m}}(k)=k^{2m}+\tilde{R}_1 k^{2(m-1)}+\tilde{R}_2 k^{2(m-2)}+\ldots+\tilde{R}_m+\mathcal{O}\le(\frac{1}{k^2}\ri),\ k\to \infty,
\end{gather}
where $\tilde{R}_{j},\ j=1,2,\ldots,m$, can be expressed by $a_j$, $b_{j}$ and $\beta_{m}$.
As $k\to \infty$, $\tilde{h}(k)$ has the asymptotic formula:
\begin{align}
\tilde{h}(k)=\sum\limits_{j=0}^{2m-1}\tilde{h}_j k^j+\mathcal{O}\le(\frac{1}{k}\ri),
\end{align}
where
\begin{equation}\label{th0123}
\begin{aligned}\begin{array}{l}
\tilde{h}_{2m-1}=-\frac{1}{2\pi }\le(\sum\limits_{j=1}^{m-1}\int_{ ib_{j}}^{ia_{j+1}}\frac{\Omega_{\beta_{m},j}+\Delta_{\beta_{m},j}}{R_{\beta_{m}}(\zeta)}\d\zeta+\int_{ -ia_1}^{ia_1}\frac{\Omega_{\beta_{m},0}+\Delta_{\beta_{m},0}}{R_{\beta_{m}}(\zeta)}\d\zeta
+\sum\limits_{j=1}^{m-1}\int_{ -ia_{j+1}}^{-ib_{j}}\frac{\Omega_{\beta_{m},-j}+\Delta_{\beta_{m},-j}}{R_{\beta_{m}}(\zeta)}\d\zeta\ri),\\
\tilde{h}_{2m-1-2l}=-\frac{1}{2\pi }\le(\sum\limits_{j=1}^{m-1}\int_{ ib_{j}}^{ia_{j+1}}\frac{(\Omega_{\beta_{m},j}+\Delta_{\beta_{m},j})(\zeta^{2l}+\sum\limits_{p=1}^{l}\tilde{R}_p\zeta^{2(l-p)})}{R_{\beta_{m}}(\zeta)}\d\zeta+\int_{ -ia_1}^{ia_1}\frac{(\Omega_{\beta_{m},0}+\Delta_{\beta_{m},0})(\zeta^{2l}+\sum\limits_{p=1}^{l}\tilde{R}_p\zeta^{2(l-p)})}{R_{\beta_{m}}(\zeta)}\d\zeta\right.\\
\left.\ \ \ \ \ \ \ \  \ \ \ \ \ \ \ \ \
+\sum\limits_{j=1}^{m-1}\int_{ -ia_{j+1}}^{-ib_{j}}\frac{(\Omega_{\beta_{m},-j}+\Delta_{\beta_{m},-j})(\zeta^{2l}+\sum\limits_{p=1}^{l}\tilde{R}_p\zeta^{2(l-p)})}{R_{\beta_{m}}(\zeta)}\d\zeta\ri),\ l=1,2,\ldots,m-1,\\
\tilde{h}_{2m-2}=-\frac{1}{2\pi }\le(\sum\limits_{j=1}^{m-1}\int_{ ib_{j}}^{ia_{j+1}}\frac{(\Omega_{\beta_{m},j}+\Delta_{\beta_{m},j})\zeta}{R_{\beta_{m}}(\zeta)}\d\zeta+\int_{ -ia_1}^{ia_1}\frac{(\Omega_{\beta_{m},0}+\Delta_{\beta_{m},0})\zeta}{R_{\beta_{m}}(\zeta)}\d\zeta
+\sum\limits_{j=1}^{m-1}\int_{ -ia_{j+1}}^{-ib_{j}}\frac{(\Omega_{\beta_{m},-j}+\Delta_{\beta_{m},-j})\zeta}{R_{\beta_{m}}(\zeta)}\d\zeta\ri),\\
\tilde{h}_{2m-2-2l}=-\frac{1}{2\pi }\le(\sum\limits_{j=1}^{m-1}\int_{ ib_{j}}^{ia_{j+1}}\frac{(\Omega_{\beta_{m},j}+\Delta_{\beta_{m},j})(\zeta^{2l}+\sum\limits_{p=1}^{l}\tilde{R}_p\zeta^{2(l-p)})\zeta}{R_{\beta_{m}}(\zeta)}\d\zeta+\int_{ -ia_1}^{ia_1}\frac{(\Omega_{\beta_{m},0}+\Delta_{\beta_{m},0})(\zeta^{2l}+\sum\limits_{p=1}^{l}\tilde{R}_p\zeta^{2(l-p)})\zeta}{R_{\beta_{m}}(\zeta)}\d\zeta\right.\\
\left.\ \ \ \ \ \ \ \  \ \ \ \ \ \ \ \ \
+\sum\limits_{j=1}^{m-1}\int_{ -ia_{j+1}}^{-ib_{j}}\frac{(\Omega_{\beta_{m},-j}+\Delta_{\beta_{m},-j})(\zeta^{2l}+\sum\limits_{p=1}^{l}\tilde{R}_p\zeta^{2(l-p)})\zeta}{R_{\beta_{m}}(\zeta)}\d\zeta\ri), \ l=1,2,\ldots,m-1.
\end{array}
\end{aligned}
\end{equation}
In order to get the solution of RH problem \ref{RHP6}, we introduce a  Riemann surface $\mathfrak{X}_{{\beta_{m}}}$ of genus $2m-1$ (see Figure 4 ).
\begin{figure}[th]
\centering
\scalebox{.9}{
\begin{tikzpicture}[>=stealth]
\path (0,0) coordinate (O);

\coordinate (TL) at (-2,5);
\coordinate (TR) at (7,5);
\coordinate (BL) at (-4,3);
\coordinate (BR) at (5,3);
\coordinate (INF1) at (5.6,4.5);

\coordinate (shift) at (0,-2.3);
\coordinate (TL2) at ($  (TL) + (shift)  $);
\coordinate (TR2) at ($  (TR) + (shift)  $);
\coordinate (BL2) at ($  (BL) + (shift)  $);
\coordinate (BR2) at ($  (BR) + (shift)  $);
\coordinate (INF2) at ($ (INF1) + (shift) $);

\coordinate (eta1) at (2.2,4);
\coordinate (eta2) at (3.2,4);
\coordinate (eta3) at (4.2,4);
\coordinate (alpha) at (5.2,4);

\coordinate (-eta1) at (0.2,4);
\coordinate (-eta2) at (-0.8,4);
\coordinate (-eta3) at (-1.8,4);
\coordinate (-alpha) at (-2.8,4);
\draw (TL) -- (TR) -- (BR) -- (BL) -- cycle;
\node[ label={[label distance= -0.3cm, below, xshift= -0.1cm]$\times$}]  at (INF1) {$\infty^+$};
\draw (TL2) -- (TR2) -- (BR2) -- (BL2) -- cycle;
\node[ label={[label distance= -0.3cm, below, xshift= -0.1cm]$\times$}]  at (INF2) {$\infty^-$};

\draw[->-=0.5] (eta1) -- (eta2);
\draw[->-=0.5] (eta3) -- (alpha);
\draw[->-=0.5] (-eta2) -- (-eta1);
\draw[->-=0.5] (-alpha) -- (-eta3);

\draw[->-=0.5] ($(eta2)+(shift)$) -- ($ (eta1) +(shift) $);
\draw[->-=0.5] ($(alpha)+(shift)$) -- ($ (eta3) +(shift) $);
\draw[->-=0.5] ($(-eta1)+(shift)$) -- ($ (-eta2) +(shift) $);
\draw[->-=0.5] ($(-eta3)+(shift)$) -- ($ (-alpha) +(shift) $);

\foreach \pos/\label in {eta1/ia_j, eta2/ib_j, eta3/ia_m, alpha/i{\beta_{m}}, -eta1/-ia_{m-j},-eta2/-ib_{m-j},-eta3/-ia_m, -alpha/-i{\beta_{m}}}{
\node[circle,fill=black, inner sep=0pt,minimum size=3pt,label=below:{\tiny $\label$}] at  (\pos) {};
\node[circle,fill=black, inner sep=0pt,minimum size=3pt,label=below:{\tiny $\label$}] at  ($ (\pos)+(shift) $)  {};
}
\coordinate (eta0) at (1.2,4);
\coordinate (eta01) at (1,4);
\coordinate (eta02) at (1.4,4);
\draw[fill] (eta0) circle [radius=0.02];
\draw[fill] (eta01) circle [radius=0.02];
\draw[fill] (eta02) circle [radius=0.02];
\coordinate (eta00) at (1.2,1.7);
\coordinate (eta001) at (1,1.7);
\coordinate (eta002) at (1.4,1.7);
\draw[fill] (eta00) circle [radius=0.02];
\draw[fill] (eta001) circle [radius=0.02];
\draw[fill] (eta002) circle [radius=0.02];
\foreach \pos in {eta1,eta2,-eta1,-eta2}{
\draw[dashed, black!30] (\pos) -- ($ (\pos) + (shift) $);
}
\foreach \pos in {eta3,alpha,-eta3,-alpha}{
\draw[dashed, black!30] (\pos) -- ($ (\pos) + (shift) $);
}
\draw[->- = .25, red] ($ 0.5*(-alpha)+0.5*(-eta3)  $) .. controls + (100:1cm) and + (100:1cm) .. ($ 0.5*(alpha)+0.5*(eta3) $);
\draw[->- = .25, red] ($ 0.5*(-alpha)+0.5*(-eta3)  $) .. controls + (100:.7cm) and + (100:.7cm) .. ($ 0.5*(eta2)+0.5*(eta1) $);
\draw[->- = .25, red] ($ 0.5*(-alpha)+0.5*(-eta3)  $) .. controls + (100:.4cm) and + (100:.4cm) .. ($ 0.5*(-eta1)+0.5*(-eta2) $);

\draw[->- = .25, red]  ($ 0.5*(-eta1)+0.5*(-eta2)  + (shift) $) .. controls + (-100:0.4cm) and + (-100:.4cm) .. ($ 0.5*(-alpha)+0.5*(-eta3) +(shift) $);
\draw[->- = .25, red]  ($ 0.5*(eta1)+0.5*(eta2)  + (shift) $) .. controls + (-100:0.7cm) and + (-100:.7cm) .. ($ 0.5*(-alpha)+0.5*(-eta3) +(shift) $);
\draw[->- = .25, red]  ($ 0.5*(eta3)+0.5*(alpha)  + (shift) $) .. controls + (-100:1cm) and + (-100:1cm) .. ($ 0.5*(-alpha)+0.5*(-eta3) +(shift) $);

\draw[red!30, dashed] ($ 0.5*(-eta1)+0.5*(-eta2)  $) -- ++ (shift);
\draw[red!30, dashed] ($ 0.5*(-eta3)+0.5*(-alpha)  $) -- ++ (shift);
\draw[red!30, dashed] ($ 0.5*(eta1)+0.5*(eta2)  $) -- ++ (shift);
\draw[red!30, dashed] ($ 0.5*(eta3)+0.5*(alpha)  $) -- ++ (shift);

\node[above, red] at (-1.3,3.8) {\small $\scriptstyle{\tilde{{\mathcal B}}_j}$};
\node[above, red] at (1.2,4) {\small $\scriptstyle{\tilde{{\mathcal B}}_{m-1+j}}$};
\node[above, red] at (3.7,4.05) {\small $\scriptstyle{\tilde{{\mathcal B}}_{2m-1}}$};
\draw[->- = .25, blue] ( $0.5*(eta1)+0.5*(eta2) $) ellipse (0.8cm and .4cm);
\draw[->- = .25, blue] ( $0.5*(eta3)+0.5*(alpha) $) ellipse (0.8cm and .4cm);
\draw[->- = .25, blue] ( $0.5*(-eta1)+0.5*(-eta2) $) ellipse (0.8cm and .4cm);
\node[blue, above] at (4.9,3.9) {\small $\scriptstyle\tilde{\mathcal A}_{2m-1}$};
\node[blue, above] at (2.8,3.85) {\small $\scriptstyle\tilde{\mathcal A}_{m-1+j}$};
\node[blue, above] at (0.01,3.85) {\small $\scriptstyle\tilde{\mathcal A}_{j}$};
\end{tikzpicture}
}
\caption{The homology basis for the Riemann surface
$\mathfrak{X}_{{\beta_{m}}}$ associated with
$R_{{\beta_{m}}}(k) = \sqrt{(k^2+a_m^2)(k^2 +\beta_{m}^2) \prod\limits_{j=1}^{m-1}(k^2+a_j^2)(k^2 +b_j^2)}$.
}
\label{fig:homology2}
\end{figure}
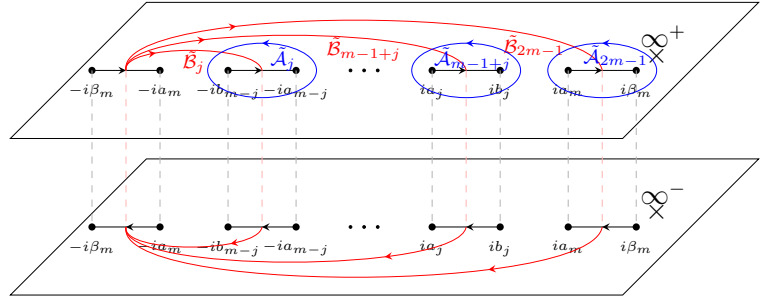
Define
\begin{align}
\tilde{\omega}_j(k)=\int_{-ia_{m}}^{k}\tilde{\psi}_j(\zeta)\d \zeta,\ \ \tilde{\psi}_j(\zeta)=\frac{\sum\limits_{l=1}^{2m-1}\tilde{c}_{jl}\zeta^{2m-1-l}}{R_{{\beta_{m}}}(\zeta)}, \ j=1,2,\ldots,2m-1,
\end{align}
where the coefficients $\tilde{c}_{jl}$ are determined by
$$\oint_{\tilde{\mathcal{A}}_l}\d{\tilde{\omega}}_j(\mathcal{P})=\delta_{jl},\ l,j=1,2,\ldots,2m-1.$$
Define a $(2m-1)\times(2m-1)$ period matrix $\tilde{B}$ by
\begin{gather}
\tilde{B}_{jl}=\oint_{\tilde{\mathcal{B}}_l}\d{\tilde{\omega}}_j(\mathcal{P}),\ l,j=1,2,\ldots,2m-1,
\end{gather}
where the integral paths $\tilde{\mathcal{A}}_l$ and $\tilde{\mathcal{B}}_l$ are shown in Figure 4.
Then we introduce the Abel mapping $\tilde{A}_{j}(\mathcal{P})$ by
\begin{gather}\label{tA}
\tilde{A}_{j}(\mathcal{P})=\int_{\tilde{\mathcal{P}}_0}^{\mathcal{P}}\d \tilde{\omega}_j(\mathcal{\tilde{P}}),
\end{gather}
where the fixed point $\tilde{\mathcal{P}}_0$ on the Riemann surface
$\mathfrak{X}_{{\beta_{m}}}$ satisfies $\tilde{\pi}(\tilde{\mathcal{P}}_0)=-ia_{m}$ and the Abelian integral $\tilde{\mathbf{A}}(k)=(\tilde{A}_1(k), \tilde{A}_2(k), \ldots, \tilde{A}_{2m-1}(k))^{\top}$ is considered in the upper sheet of $\mathfrak{X}_{{\beta_{m}}}$.
Define the additional Abel integrals as follows:
\begin{align}
\tilde{\varsigma}_j(k)=\int_{-i{\beta_{m}}}^{k}\tilde{\varphi}_j(\zeta)\d \zeta, \ \tilde{\varphi}_j(\zeta)=\frac{\sum\limits_{l=1}^{4m-1}\tilde{s}_{jl}\zeta^{4m-1-l}}{R_{{\beta_{m}}}(\zeta)},\ j=1,2,\ldots,2m-1,
\end{align}
where $\tilde{s}_{jl}$ are chosen so that
\begin{align}
&\tilde{\varsigma}_j(k)\to k^{j}+\mathcal{O}(1),\text{as}\  k \to \infty^+,\\
&\oint_{\tilde{\mathcal{A}}_{j}}\tilde{\varphi}_j(\zeta)\d  \zeta=0,\ j,l=1,2,\ldots,2m-1.
\end{align}
As we did in section 3, we define
\begin{align}\label{tUVW}
\tilde{U}_{j,l}=\oint_{\tilde{\mathcal{B}}_{l}}\tilde{\varphi}_{j}(\zeta)\d  \zeta,\ j,l=1,2,\ldots,2m-1,
\end{align}
and
\begin{align}\label{tJ123}
\tilde{J}_{j}=\lim\limits_{k\to \infty}\tilde{\varsigma}_{j}(k)-k^{j},\ j=1,2,\ldots,2m-1.
\end{align}
We observe that
\begin{align}\label{tvarsigmajump}
&\tilde{\varsigma}_{j+}(k)+\tilde{\varsigma}_{j-}(k)=\tilde{U}_{j,l},\ k\in \Sigma_{-(m-l)},\ l=1,2,\ldots,m-1,\ j=1,2,\ldots, 2m-1,\\
&\tilde{\varsigma}_{j+}(k)+\tilde{\varsigma}_{j-}(k)=\tilde{U}_{j,m-1+l},\ k\in \Sigma_{l},\ l=1,2,\ldots,m-1,\ j=1,2,\ldots, 2m-1,\\
&\tilde{\varsigma}_{j+}(k)+\tilde{\varsigma}_{j-}(k)=\tilde{U}_{j,2m-1},\ k\in \Sigma_{-m,\beta_m}
\ j=1,2,\ldots, 2m-1.
\end{align}
Next, we introduce the Riemann-theta function
\begin{align}
{\Theta}(\tilde{\mathbf{v}})=\sum\limits_{\tilde{\mathbf{w}}\in \Z^{2m-1}}e^{\pi i(\tilde{B}\tilde{\mathbf{w}},\tilde{\mathbf{w}})+2\pi i (\tilde{\mathbf{w}},\tilde{\mathbf{v}})},
\end{align}
where $\tilde{\mathbf{v}}=(\tilde{v}_1,\tilde{v}_2,\ldots,\tilde{v}_{2m-1})^{\top}$, $\tilde{\mathbf{w}}=(\tilde{w}_1,\tilde{w}_2,\ldots,\tilde{w}_{2m-1})^{\top}$ and $(\tilde{\mathbf{w}},\tilde{\mathbf{v}})=\tilde{w}_1 \tilde{v}_1+\tilde{w}_2 \tilde{v}_2+\cdots+\tilde{w}_{2m-1} \tilde{v}_{2m-1}$.
Finally, we define a function $\tilde{\gamma}(k)$ which is analytic in $\C\backslash (\bigcup\limits_{j=1}^{m-1}(\Sigma_{j}\cup \Sigma_{-j})\cup \Sigma_{m, \beta_{m}}\cup \Sigma_{-m, \beta_m})$ by
\begin{align}
\tilde{\gamma}(k)=\le(\frac{(k-i{\beta_{m}})(k+ia_m)}{(k+i{\beta_{m}})(k-ia_m)}\prod\limits_{j=1}^{m-1}\frac{(k-ib_j)(k+ia_j)}{(k+ib_j)(k-ia_j)} \ri)^{\frac{1}{4}}.
\end{align}	
Define the zeros of function $\tilde{\gamma}(k)-\frac{1}{\tilde{\gamma}(k)}$ by  $\infty^{+}$, $\tilde{\mathcal{P}}_1$, $\tilde{\mathcal{P}}_2$, $\ldots$, $\tilde{\mathcal{P}}_{2m-1}$, let $\tilde{D}$ denote the divisor $\tilde{D}=\sum\limits_{j=1}^{2m-1}\tilde{\mathcal{P}}_j$, the constant vector $\tilde{\mathbf{d}}$ is given by
\begin{align}\label{td}
\tilde{\mathbf{d}}=\tilde{\mathbf{A}}(\tilde{D})+\mathbf{\tilde{K}},
\end{align}
where $\tilde{\mathbf{K}}$ is the Riemann-theta constant vector which entries are defined by
\begin{align}
\tilde{K}_{j}=\frac{1}{2}\sum_{l=1}^{2m-1}\tilde{B}_{lj}-\frac{j}{2}.
\end{align}
Using the Riemann-theta function, we define a $2\times2$ matrix-valued function $\tilde{Q}(k)$ which entries are given by
\begin{equation}
\begin{array}{l}
\tilde{Q}_{11}(x;k)=\frac{1}{2}(\tilde{\gamma}(k)+\frac{1}{\tilde{\gamma}(k)})\frac{{\Theta}(\tilde{\mathbf{A}}(\infty)+\tilde{\mathbf{d}})}{{\Theta}(\tilde{\mathbf{A}}(k)+\tilde{\mathbf{d}})} \frac{{\Theta}(\tilde{\mathbf{A}}(k)+\tilde{\mathbf{d}}-\frac{1}{2 \pi i}\sum\limits_{j=1}^{2m-1}\tilde{\mathbf{U}}_j\tilde{h}_j)}{{\Theta}(\tilde{\mathbf{A}}(\infty)+\tilde{\mathbf{d}}-\frac{1}{2 \pi i}\sum\limits_{j=1}^{2m-1}\tilde{\mathbf{U}}_j\tilde{h}_j)} \exp(-\sum\limits_{j=1}^{2m-1}(\tilde{\varsigma}_j-\tilde{J}_j)\tilde{h}_j),\\
\tilde{Q}_{12}(x;k)=\frac{1}{2}(\tilde{\gamma}(k)-\frac{1}{\tilde{\gamma}(k)})\frac{{\Theta}(\tilde{\mathbf{A}}(\infty)+\tilde{\mathbf{d}})}{{\Theta}(\tilde{\mathbf{A}}(k)-\tilde{\mathbf{d}})} \frac{{\Theta}(\tilde{\mathbf{A}}(k)-\tilde{\mathbf{d}}+\frac{1}{2 \pi i}\sum\limits_{j=1}^{2m-1}\tilde{\mathbf{U}}_j\tilde{h}_j)}{{\Theta}(\tilde{\mathbf{A}}(\infty)+\tilde{\mathbf{d}}-\frac{1}{2 \pi i}\sum\limits_{j=1}^{2m-1}\tilde{\mathbf{U}}_j\tilde{h}_j)} \exp(\sum\limits_{j=1}^{2m-1}(\tilde{\varsigma}_j+\tilde{J}_j)\tilde{h}_j),\\
\tilde{Q}_{21}(x;k)=\frac{1}{2}(\tilde{\gamma}(k)-\frac{1}{\tilde{\gamma}(k)})\frac{{\Theta}(\tilde{\mathbf{A}}(\infty)+\tilde{\mathbf{d}})}{{\Theta}(\tilde{\mathbf{A}}(k)-\tilde{\mathbf{d}})} \frac{{\Theta}(\tilde{\mathbf{A}}(k)-\tilde{\mathbf{d}}-\frac{1}{2 \pi i}\sum\limits_{j=1}^{2m-1}\tilde{\mathbf{U}}_j\tilde{h}_j)}{{\Theta}(\tilde{\mathbf{A}}(\infty)+\tilde{\mathbf{d}}+\frac{1}{2 \pi i}\sum\limits_{j=1}^{2m-1}\tilde{\mathbf{U}}_j\tilde{h}_j)} \exp(-\sum\limits_{j=1}^{2m-1}(\tilde{\varsigma}_j+\tilde{J}_j)\tilde{h}_j),\\
\tilde{Q}_{22}(x;k)=\frac{1}{2}(\tilde{\gamma}(k)+\frac{1}{\tilde{\gamma}(k)})\frac{{\Theta}(\tilde{\mathbf{A}}(\infty)+\tilde{\mathbf{d}})}{{\Theta}(\tilde{\mathbf{A}}(k)+\tilde{\mathbf{d}})} \frac{{\Theta}(\tilde{\mathbf{A}}(k)+\tilde{\mathbf{d}}+\frac{1}{2 \pi i}\sum\limits_{j=1}^{2m-1}\tilde{\mathbf{U}}_j\tilde{h}_j)}{{\Theta}(\tilde{\mathbf{A}}(\infty)+\tilde{\mathbf{d}}+\frac{1}{2 \pi i}\sum\limits_{j=1}^{2m-1}\tilde{\mathbf{U}}_j\tilde{h}_j)} \exp(\sum\limits_{j=1}^{2m-1}(\tilde{\varsigma}_j-\tilde{J}_j)\tilde{h}_j),
\end{array}
\end{equation}
where $\tilde{\mathbf{U}}_j=(\tilde{U}_{j,1}, \tilde{U}_{j,2},\ldots, \tilde{U}_{j,2m-1} )^{\top}\in \C^{2m-1}$.
Finally, we can construct the matrix solution $\tilde{S}^{\infty}(k)$ of RH problem \ref{RHP6} as follows:
\begin{align}
\tilde{S}^{\infty}(k)=e^{-\tilde{h}_0\sigma_3} \tilde{Q}(k) e^{\tilde{h}(k)\sigma_3}.
\end{align}
Next, we construct the local parametrix at $k=i\beta_{m}$. With the help of the local behavior of the function $g_{\beta_{m}}(k)-kx-4k^3t$ at $k=i\beta_{m}$, we introduce a local variable $\lambda=\lambda(k,\xi)$
$$-2i({g}_{\beta_{m}}(k)-\theta(x,t;k) )=\frac{4}{3}t\lambda^{3/2},\ k \in B^{(i\beta_{m})}_{\rho} = \le\{ k \in \mathbb{C} \le| \,  \le|k - i\beta_{m}\ri|< \rho \ri.  \ri\},$$
and the branch cut for $\lambda^{3/2}$ is $\lambda\in (-\infty,0)$, corresponds to $k\in i(a_m,\beta_{m})$.
Next, we introduce the model parametrix  $\Psi_{Ai}(\lambda)$ as in \cite{Girotti-2}, the function $\Psi_{Ai}(\lambda)$ satisfies the jump conditions
\begin{align}
\Psi_{Ai+}(\lambda)=\Psi_{Ai-}(\lambda)\begin{cases}
\begin{pmatrix}
0&i\\i&0
\end{pmatrix},\ \ \lambda\in (-\infty,0),\\
\begin{pmatrix}
1&-i\\0&1
\end{pmatrix},\ \ \lambda\in (\infty e^{\pm\frac{2\pi i}{3}},0),\\
\begin{pmatrix}
1&0\\i&1
\end{pmatrix},\ \ \lambda\in (0,+\infty),
\end{cases}
\end{align}
with the asymptotics at infinity
$$ \Psi_{Ai}(\lambda)=\lambda^{-\frac{\sigma_3}{4}}\frac{1}{\sqrt{2}}\begin{pmatrix}1&-1\\1&1 \end{pmatrix} (  I +\mathcal{O}(\frac{1}{\lambda^{\frac{3}{2}}})) e^{\frac{2}{3}\lambda^{\frac{3}{2}}\sigma_3}, \ \ \text{uniformly in}\ \arg \lambda\in[-\pi,\pi].$$
Finally, we construct the local parametrix $\tilde{P}^{i\beta_{m}}(k)$ around the endpoint $i\beta_{m}$ by
\begin{align}
\tilde{P}^{i\beta_{m}}(k)=\tilde{D}(k) \Psi_{Ai}( t^{\frac{2}{3}}\lambda(k;\xi))e^{i( g_{\beta_{m}}(k)-\theta(x,t;k))\sigma_3} ( f_{\beta_{m}}(k) \sqrt{2\hat{r}_m(k)})^{\sigma_3}, \ \ k\in B^{(i\beta_{m})}_{\rho},
\end{align}
where $\tilde{D}(k)$ is given by
$$ \tilde{D}(k)=\tilde{S}^{\infty}(k)( f_{\beta_{m}}(k) \sqrt{2\hat{r}_m(k)})^{-\sigma_3}\frac{1}{\sqrt{2}}\begin{pmatrix}1&1\\-1&1 \end{pmatrix}(t^{\frac{2}{3}}\lambda)^{\frac{\sigma_3}{4}}.$$
On the boundary $\partial B^{(i\beta_{m})}_{\rho}$ we have the following matching condition
$$\tilde{S}(k) (\tilde{P}^{i\beta_{m}}(k))^{-1}= I +\mathcal{O}(\frac{1}{t}),\ \ \text{as}\ t\to\infty.$$
The local parametrix of other endpoints can be constructed similarly as \cite{Girotti-2}.
We define the following error function $\tilde{\mathcal{E}}(k)$:
 \begin{gather}
	\tilde{\mathcal{E}}(k) = \tilde{S}(k) \left(\tilde{P}(k) \right)^{-1},
\text{}\\
	\tilde{P}(k) = \begin{cases}
		\tilde{S}^{\infty}(k),  & k\in\mathbb{C}\backslash (\bigcup\limits_{j=1}^{m}B^{(\pm i a_j)}_{\rho}\cup \bigcup\limits_{j=1}^{m-1}B^{(\pm i b_j)}_{\rho}\cup B^{(\pm i\beta_{m})}_{\rho} ), \\
        \tilde{P}^{ia_j}(k), & k\in B^{(ia_j)}_{\rho},\ j=1,2,\ldots,m, \\
        \tilde{P}^{ib_j}(k), & k\in B^{(ib_j)}_{\rho},\ j=1,2,\ldots,m-1, \\
		\tilde{P}^{i\beta_{m}} (k), & k\in B^{(i\beta_{m})}_{\rho}, \\
		\tilde{P}^{-i\beta_{m}}(k), & k\in B^{(-i\beta_{m})}_{\rho}, \\
        \tilde{P}^{-ia_j}(k), & k\in B^{(-ia_j)}_{\rho},\ j=1,2,\ldots,m, \\
        \tilde{P}^{-ib_j}(k), & k\in B^{(-ib_j)}_{\rho},\ j=1,2,\ldots,m-1.
	\end{cases}
\end{gather}
Lemma \ref{lemma5.2} and the matching conditions of local parametrix $\tilde{P}^{\pm ia_j}(k)$, $\tilde{P}^{\pm ib_j}(k)$, $\tilde{P}^{\pm i\beta_{m}} (k)$ allow to conclude that
\begin{gather}
\tilde{\mathcal{E}}(x,t;k)= I +\mathcal{O}(t^{-1}),\ \text{as}\  t\ \to \infty.
\end{gather}
Tracing back all the transformations, we get the following theorem.
\begin{theorem}
In the region $\xi_{2m-2}<\xi<\xi_{2m-1}$, ${\beta_{m}}(\xi)\in (a_{m},b_{m})$ and satisfies the Whitham modulation equation \eqref{alpha}, the large-time asymptotics is expressed by the Riemann-theta function of genus $2m-1$
\begin{gather}
\begin{array}{l}
u(x,t)={\alpha}_2\frac{{\Theta}(\tilde{\mathbf{A}}(\infty)-\tilde{\mathbf{d}}+\frac{1}{2\pi i}\sum\limits_{j=1}^{2m-1}\tilde{\mathbf{U}}_j\tilde{h}_j)}{{\Theta}(\tilde{\mathbf{A}}(\infty)+\tilde{\mathbf{d}}-\frac{1}{2\pi i}\sum\limits_{j=1}^{2m-1}\tilde{\mathbf{U}}_j\tilde{h}_j)} \frac{{\Theta}(\tilde{\mathbf{A}}(\infty)+\tilde{\mathbf{d}})}{{\Theta}(\tilde{\mathbf{A}}(\infty)-\tilde{\mathbf{d}})} \exp(\sum\limits_{j=1}^{2m-1}2\tilde{J}_j\tilde{h}_j-2\tilde{h}_0)+\mathcal{O}(\frac{1}{t}),
\end{array}
\end{gather}
where $\tilde{\mathbf{A}}$, $\tilde{\mathbf{d}}$ and $\tilde{\mathbf{U}}_j$ are defined by \eqref{tA}, \eqref{td}, \eqref{tUVW}, respectively, the constants $\tilde{h}_0$, $\tilde{h}_1$, $\ldots$, $\tilde{h}_{2m-1}$ are given by \eqref{th0123}, the functions $\tilde{J}_1$, $\tilde{J}_2$, $\ldots$, $\tilde{J}_{2m-1}$ are shown in equation \eqref{tJ123} and ${\alpha}_2=\beta_{m}(\xi)-a_m+\sum\limits_{j=1}^{m-1}(b_j-a_j).$
\end{theorem}

\section{The second genus $2m-1$ sector : $ \xi_{2m-1}<\xi<\xi_{2m}$}
When  $\xi_{2m-1}<\xi<\xi_{2m}$, where $\xi_{2m-1}$, $m=2,3,\ldots, n$, are defined in \eqref{xi2},  $\xi_{2m}$, $m=2,3,\ldots, n-1$, are defined in \eqref{xi3},  and $\xi_{2n}=\infty$, in this case, we find that ${\beta_{m}}(\xi)=b_m$. We define a new scalar function $g_{b_m}(k)$ as follows:
\begin{equation}
\begin{array}{l}
g_{b_m}(k)=xk+4k^3t-x \int_{ib_m}^{k}\frac{\zeta^{2m}+\sum\limits_{j=1}^{m}\grave{c}_j\zeta^{2(m-j)}}{R_{b_m}(\zeta)}\d\zeta-12t \int_{ib_m}^{k}\frac{\zeta^{2m+2}+\frac{1}{2}\sum\limits_{j=1}^{m}(a_j^2+b_j^2) \zeta^{2m}+\sum\limits_{j=1}^{m}\acute{c}_j\zeta^{2(m-j)}}{R_{b_m}(\zeta)}\d\zeta,
\end{array}
\end{equation}
where $R_{b_m}(k)=\sqrt{\prod\limits_{j=1}^{m}(k^2+a_j^2)(k^2+b_j^2)}$, with the constants $\grave{c}_j$ and $\acute{c}_j$ are chosen so that
\begin{align}
&\int_{ib_j}^{ia_{j+1}}\frac{\zeta^{2m}+\sum\limits_{l=1}^{m}\grave{c}_l\zeta^{2(m-l)}}{R_{b_m}(\zeta)}\d\zeta=0,\ j=1,2,\ldots,m-1,\\ &\int_{-ia_1}^{ia_1}\frac{\zeta^{2m}+\sum\limits_{l=1}^{m}\grave{c}_l\zeta^{2(m-l)}}{R_{b_m}(\zeta)}\d\zeta=0,\\
&\int_{ib_j}^{ia_{j+1}}\frac{\zeta^{2m+2}+\frac{1}{2}\sum\limits_{l=1}^{m}(a_l^2+b_l^2) \zeta^{2m}+\sum\limits_{l=1}^{m}\acute{c}_l\zeta^{2(m-l)}}{R_{b_m}(\zeta)}\d\zeta=0, j=1,2,\ldots,m-1,\\ &\int_{-ia_1}^{ia_1}\frac{\zeta^{2m+2}+\frac{1}{2}\sum\limits_{l=1}^{m}(a_l^2+b_l^2) \zeta^{2m}+\sum\limits_{l=1}^{m}\acute{c}_l\zeta^{2(m-l)}}{R_{b_m}(\zeta)}\d\zeta=0.
\end{align}
Then the function $g_{b_m}(k)$ satisfies the following conditions
\begin{align}
&g_{b_m+}(k)+ g_{b_m-}(k)= 2kx+8k^3t, & k\in\bigcup\limits_{j=1}^{m}(\Sigma_j\cup\Sigma_{-j}),\label{g41}\\
&g_{b_m+}(k)-g_{b_m-}(k)=\Omega_{b_m,0} ,& k \in  i[-a_1,a_1], \label{g42}\\
&g_{b_m+}(k)-g_{b_m-}(k)=\Omega_{b_m,j} ,& k \in  i[b_j,a_{j+1}],\ j=1,2,\ldots,m-1, \label{g43}\\
&g_{b_m+}(k)-g_{b_m-}(k)=\Omega_{b_m,-j} ,& k \in  i[-a_{j+1},-b_j],\ j=1,2,\ldots,m-1, \label{g44}
\end{align}
with the integral constants $\Omega_{b_m,0}$, $\Omega_{b_m,j}$ and $\Omega_{b_m,-j}$ are given by
\begin{equation}
\begin{array}{l}
\Omega_{b_m,m-1}=-2x \int_{i{b_m}}^{ia_m}\frac{\zeta^{2m}+\sum\limits_{l=1}^{m}\grave{c}_l\zeta^{2(m-l)}}{R_{b_m+}(\zeta)}\d\zeta
-24t \int_{i{b_m}}^{i a_m}\frac{\zeta^{2(m+1)}+\frac{1}{2}\sum\limits_{l=1}^{m}(a_l^2+b_l^2) \zeta^{2m}+\sum\limits_{l=1}^{m}\acute{c}_l\zeta^{2(m-l)}}{R_{b_m+}(\zeta)}\d\zeta,
\\
\Omega_{b_m,j}=\Omega_{b_m,m-1}+\Omega_{b_m,m-2}+\cdots+\Omega_{b_m,j+1}-2x \int_{i{b_{j+1}}}^{ia_{j+1}}\frac{\zeta^{2m}+\sum\limits_{l=1}^{m}\grave{c}_l\zeta^{2(m-l)}}{R_{b_m+}(\zeta)}\d\zeta
\\ \ \ \ \ \ \ \ \ \ -24t \int_{i{b_{j+1}}}^{ia_{j+1}}\frac{\zeta^{2(m+1)}+\frac{1}{2}\sum\limits_{l=1}^{m}(a_l^2+b_l^2) \zeta^{2m}+\sum\limits_{l=1}^{m}\acute{c}_l\zeta^{2(m-l)}}{R_{b_m+}(\zeta)}\d\zeta,\ j=0,1,\ldots, m-2,
\\
\Omega_{b_m,-j}=\Omega_{b_m,j}, \ j=1,2,\ldots,m-1.
\end{array}
\end{equation}
The following analysis is similar to the study in section 5. Next we have the following theorem.
\begin{thm}
In the regime $t \to +\infty$, $\xi_{2m-1}<\xi<\xi_{2m}$, the mKdV soliton gas solution $u(x,t)$ has the following asymptotic behaviour
\begin{gather}
\begin{array}{l}
u(x,t)=\alpha_3\frac{{\Theta}(\mathbf{\grave{A}}(\infty)-\mathbf{\grave{d}}+\frac{1}{2\pi i}\sum\limits_{j=1}^{2m-1}\mathbf{\grave{U}}_{j}\grave{h}_{j})}{{\Theta}(\mathbf{\grave{A}}(\infty)+\mathbf{d}-\frac{1}{2\pi i}\sum\limits_{j=1}^{2m-1}\mathbf{\grave{U}}_{j}\grave{h}_{j})} \frac{{\Theta}(\mathbf{\grave{A}}(\infty)+\mathbf{\grave{d}})}{{\Theta}(\mathbf{\grave{A}}(\infty)-\mathbf{\grave{d}})} \exp(2\sum\limits_{j=1}^{2m-1}\grave{J}_j\grave{h}_j-2\grave{h}_0)+\mathcal{O}(\frac{1}{t}),
\end{array}
\end{gather}
where $\mathbf{\grave{A}}$, $\mathbf{\grave{d}}$, $\mathbf{\grave{U}}_{j}$, $\grave{J}_j$ and $\grave{h}_j$ have the same form as $\mathbf{\tilde{A}}$, $\mathbf{\tilde{d}}$, $\mathbf{\tilde{U}}_{j}$, $\tilde{J}_j$ and $\tilde{h}_j$, the only difference being the substitution of $b_m$ for $\beta_m$, and $\alpha_3=\sum\limits_{j=1}^{m}(b_{j}-a_{j}).$
\end{thm}
		
\section*{Acknowledgments}
This work is supported by the National Natural Science Foundation of China (Grant Nos.  12471234, 12201573, 12171439).

\section*{Conflict of interest}
The authors declare no conflicts of interest.

\section*{Data availability statement}
Data sharing is not applicable to this article as no new data were created
or analyzed during the current study.


\begin{thebibliography}{90}
\bibitem{MC} M. A. Alejo, C. Mu\~{n}oz, Nonlinear stability of MKdV breathers, Comm. Math. Phys. 324 (2013) 233--262.
\bibitem{Grava-3}M. Bertola, T. Grava, G. Orsatti, Soliton shielding of the focusing nonlinear Schr\"odinger equation, Phys. Rev. Lett. 130 (2023) 127201.
\bibitem{CKSTT} J. Colliander, M. Keel, G. Staffilani, H. Takaoka, T. Tao, Sharp global well-posedness for KdV and modified KdV on R and T, J. Amer. Math. Soc. {16} (2003) 705--749.
\bibitem{CL1}G. Chen, J. Q. Liu, Soliton resolution for the focusing modified KdV equation, Ann. Inst. H. Poincar\'{e} C Anal. Non Lin\'{e}aire 38 (2021) 2005--2071.
\bibitem{CL2}G. Chen, J. Q. Liu, Long-time asymptotics to the modified KdV equation in weighted
Sobolev spaces, Forum Math. Sigma 10 (2022) e66 1--52.
\bibitem{DeiftItsZhou}P. Deift, A. Its, X. Zhou, A Riemann-Hilbert approach to asymptotic problems arising in the theory of random matrix models, and also in the theory of integrable statistical mechanics, Ann. of Math. 146 (1997) 149-235.
\bibitem{DeiftZ}P. Deift, X. Zhou, A steepest descent method for oscillatory Riemann--Hilbert problems. Asymptotics for the MKdV equation, Ann. of Math. 137 (1993) 295--368.
\bibitem{DZZ}S. Dyachenko, D. Zakharov, V. Zakharov, Primitive potentials and bounded solutions of the KdV equation, Physica D 333 (2016) 148--156.
\bibitem{E1}G. A. El, The thermodynamic limit of the Whitham equations, {Phys. Lett. A} 311 (2003) 374--383.
\bibitem{E2}G. A. El, A. M. Kamchatnov, Kinetic equation for a dense soliton gas, {Phys. Rev. Lett.} 95 (2005) 204101.
\bibitem{E3}G. A. El, A. M. Kamchatnov, M. V. Pavlov, S. A. Zykov, Kinetic equation for a soliton gas and its hydrodynamic reductions, J. Nonlinear Sci. 21 (2010) 151--191.
\bibitem{E4}G. A. El, A. Tovbis, Spectral theory of soliton and breather gases for the focusing nonlinear Schr\"odinger equation, Phys. Rev. E 21 (2020) 052207.
\bibitem{GA}A. A. Gelash, D. S. Agafontsev, Strongly interacting soliton gas and formation of rogue waves, Phys. Rev. E 98 (2018) 042210.
\bibitem{GAZ}A. A. Gelash, D. S. Agafontsev, V. E. Zakharov, G. A. El, S. Randoux, P. Suret, Bound state soliton gas dynamics underlying the spontaneous modulational instability, Phys. Rev. Lett. 123 (2019) 234102.
\bibitem{GPR}P. Germain, F. Pusateri, F. Rousset, Asymptotic stability of solitons for mKdV equation, Adv. Math. 299 (2016) 272--330.
\bibitem{Girotti-1}M. Girotti, T. Grava, R. Jenkins, K. T. R. McLaughlin, Rigorous asymptotics of a KdV soliton gas, Commun. Math. Phys. 384 (2021) 733--784.
\bibitem{Girotti-2}M. Girotti, T. Grava, R. Jenkins, K. T. R. McLaughlin, A. Minakov, Soliton versus the gas: Fredholm determinants, analysis, and the rapid oscillations behind the kinetic equation, Commun. Pure Appl. Math. 76 (2023) 3233--3299.
\bibitem{TM} T. Grava, A. Minakov, On the long-time asymptotic behavior of the modified
              {K}orteweg--de {V}ries equation with step-like initial data, SIAM J. Math. Anal. 52 (2020) 5892--5993.
\bibitem{HXF}X. F. Han, X. E. Zhang, H. H. Dong, Large x asymptotics of the soliton gas for the nonlinear {S}chr\"odinger equation, Stud. Appl. Math. 154 (2025) e70027.
\bibitem{KPV}C. Kenig, G. Ponce, L. Vega, Well-posedness and scattering results for the generalized Korteweg-de Vries equation via the contraction principle, Comm.Pure Appl. Math. 46 (1993) 527--620.
\bibitem{KVSD}V. Kotlyarov, D. Shepelsky, Planar unimodular Baker-Akhiezer function for the nonlinear Schr\"{o}dinger equation, Ann. Math. Sci. Appl. 2 (2017) 343--384.
\bibitem{Ling} L. M. Ling, X. Sun, The multi elliptic-localized solutions and their asymptotic
              behaviors for the m{K}d{V} equation, Stud. Appl. Math. 150 (2023) 135--183.
\bibitem{SRADE2024}P. Suret, S. Randoux, A. Gelash, D. Agafontsev, B. Doyon, G. El, Soliton gas: Theory, numerics, and experiments, Phys. Rev. E 109 (2024) 061001.
\bibitem{Yan}D. D. Yan, X. G. Geng, R. M. Li, M. X. Jia, Large-space and large-time asymptotics of the complex mKdV soliton gas with any even genus,submitted.
\bibitem{TW2022}A. Tovbis, F.D. Wang, Recent developments in spectral theory of the focusing NLS soliton and breather gases: the thermodynamic limit of average densities, fluxes and certain meromorphic differentials; periodic gases, J. Phys. A 55 (2022) 424006.
\bibitem{Wang}D. S. Wang, D. H. Zhu, X. D. Zhu, Genus two KdV soliton gases and their long-time asymptotics.arXiv:2410.22634.
\bibitem{Zak71}V. Zakharov, Kinetic equation for solitons, Sov. Phys. JETP. 33 (1971) 538--541.
\bibitem{ZhangYan}G. Q. Zhang, Z. Y. Yan, Focusing and defocusing mKdV equations with nonzero boundary conditions: Inverse scattering transforms and soliton interactions, Physica D 410 (2020) 132521.		
\end{thebibliography}
	\end{document}